\documentclass{lmcs} %%% last changed 2014-08-20
\pdfoutput=1

\usepackage{lastpage}
\lmcsdoi{18}{1}{14}
\lmcsheading{}{\pageref{LastPage}}{}{}%
{Dec.~08,~2020}{Jan.~19,~2022}{}

\usepackage[utf8]{inputenc}


\usepackage{booktabs}
\usepackage{amsmath}
\usepackage{amsfonts}
\usepackage{amssymb}

\usepackage{color}
\usepackage{bbding}

\usepackage{ltl}

\usepackage{tikz}
\usepackage{pgfplots}
\usepackage{makecell}
\usepackage{multirow}

\usepackage{stmaryrd}
\usepackage{wasysym}

\usepackage{xspace}

\makeatletter
\DeclareRobustCommand*\cal{\@fontswitch\relax\mathcal}
\makeatother

\def\plotholder{\colorbox{lightgray}{\parbox{\linewidth}{
\centering
\vspace{1cm}
\Huge Plot Removed

\vspace{2cm}

\Large Enable plotting in main.tex
\vspace{1cm}} }}
\newif\ifhideplots
\newcommand{\showplot}[1]{\ifhideplots \plotholder \else #1 \fi}

\definecolor{darkred}{rgb}{0.7 0 0}
\definecolor{lightred}{rgb}{1 0.2 0}
\definecolor{ired}{rgb}{0.8 0.36 0.36}
\definecolor{darkgreen}{rgb}{0 0.5 0}
\definecolor{darkblue}{rgb}{0 0 0.7}
\definecolor{lightblue}{rgb}{0.390 0.582 0.925}
\definecolor{darkgray}{rgb}{0.3 0.3 0.3}
\definecolor{gray}{rgb}{0.5 0.5 0.5}
\definecolor{lightgray}{rgb}{0.90 0.90 0.90}

\newcommand{\revised}[1]{\textcolor{red}{#1}}

\def\caly{{\cal Y}}
\def\cali{{\cal I}}

\def\calp{{\cal P}}

\def\calt{{\cal T}}

\def\calr{{\cal R}}
\def\cale{{\cal E}}

\def\hybrid{\mathit{H}}
\def\gtt{\mathit{GTT}}
\def\dom{\mathit{dom}}

\newcommand\subi{_\mathit{I}}
\newcommand\subo{_\mathit{O}}

\def\prefix{\ensuremath{\mathsf{Prefix}}}
\def\prefixconf{\ensuremath{\mathsf{PrefConf}}}
\def\conf{\ensuremath{\mathsf{Conf}}}

\def\wit{\ensuremath{\mathsf{Wit}}}

\def\skorc{\textsf{SkorConf}}

\def\tracec{\textsf{TraceConf}}
\def\hybc{\textsf{HybridConf}}

\def\hybpc{\textsf{HybridPermConf}}
\def\hybdc{\textsf{HybridDoubleConf}}

\def\hybconformant{hybrid conformant}
\def\hybcleanname{hybrid-conformance clean }
\def\skorcleanname{Skorokhod-conformance clean }

\def\hybclean{\textsf{HybridClean}}
\def\skorclean{\textsf{SkorClean}}
\def\robclean{\textsf{RobustClean}}

\def\ret{\ensuremath{\mathsf{Ret}}}

\def\reals{\mathbb{R}}
\def\nat{\mathbb{N}}
\def\realsnn{\mathbb{R}_{\geq 0}}
\def\rats{\mathbb{Q}}
\def\ratsnn{\mathbb{Q}_{\geq 0}}
\def\bools{\mathbb{B}}

\def\eqdef{\stackrel{\triangle}{=}}

\definecolor{cadmiumgreen}{rgb}{0.0, 0.42, 0.24}

\definecolor{eyecancerpink}{rgb}{1.0, 0.0, 1.0}

\newcommand{\myomit}[1]{}

\def\te{\tau,\epsilon}
\def\tei{\tau\subi,\epsilon\subi}

\def\dy{d_{\caly}}

\def\prefixwit{\prefix\wit}

\def\idret{(\mathsf{id},\mathsf{id})}

\def\retime{\calr \cale \calt}

\def\synch{\textsf{Sync}}

\newcommand{\rexIstA}{\ensuremath{i_{\mathit{st}}}}
\newcommand{\rexIstB}{\ensuremath{i_{\mathit{st}}}}
\newcommand{\rexIdevA}{\ensuremath{i_{\mathit{dev}}}}
\newcommand{\rexIdevB}{\ensuremath{i_{\mathit{ddev}}}}
\newcommand{\rexOstC}{\ensuremath{o(\rexIstB) }}
\newcommand{\rexOstD}{\ensuremath{o'(\rexIstB) }}
\newcommand{\rexOdevC}{\ensuremath{o(\rexIdevB) }}
\newcommand{\rexOdevD}{\ensuremath{o'(\rexIdevB) }}

\renewcommand{\emptyset}{\varnothing}

\newcommand{\NOgas}{\ensuremath{\mathrm{NO}}}
\newcommand{\NOtwo}{\ensuremath{\mathrm{NO}_2}}
\newcommand{\NOx}{\ensuremath{\mathrm{NO}_x}}

\newcommand{\abs}[1]{\lvert {#1} \rvert}

\newcommand{\NEDC}{{\sc{NEDC}}}
\newcommand{\UDC}{{\sc{UDC}}}
\newcommand{\EUDC}{{\sc{EUDC}}}
\newcommand{\PreCon}{{\sc{PreCon}}}

\newcommand{\SineNEDC}{{\sc{SineNEDC}}}
\newcommand{\PermNEDC}{{\sc{PermNEDC}}}
\newcommand{\DoubleNEDC}{{\sc{DoubleNEDC}}}

\newcommand{\AP}{\textit{AP}\xspace}
\newcommand{\stl}{\textsf{STL}\xspace}

\newcommand{\hyperltl}{\textsf{HyperLTL}\xspace}
\newcommand{\hyperstl}{\textsf{HyperSTL}}
\newcommand{\hyperstlstar}{\textsf{HyperSTL*}\xspace}
\newcommand{\stlstar}{\textsf{STL*}\xspace}
\newcommand{\mtl}{\textsf{MTL}\xspace}
\newcommand{\hypermtl}{\textsf{HyperMTL}\xspace}

\newcommand{\defeq}{\eqdef}

\newcommand{\valuet}{\mathsf{Value}}
\newcommand{\T}{\mathsf{T}}
\newcommand{\F}{\mathsf{F}}
\newcommand{\U}{\mathsf{U}}

\renewcommand{\revised}[1]{#1}

\begin{document}

\title[Conformance and Hyperproperties for Doping Detection in Time and Space]{Conformance Relations and Hyperproperties for \texorpdfstring{\\}{}Doping Detection in Time and Space}

\author[S. Biewer]{Sebastian Biewer\rsuper{a}}
\author[R. Dimitrova]{Rayna Dimitrova\rsuper{b}}
\author[M. Fries]{Michael Fries\rsuper{c}}
\author[M. Gazda]{Maciej Gazda\rsuper{d}}
\author[T. Heinze]{Thomas Heinze\rsuper{c}}
\author[H. Hermanns]{Holger Hermanns\rsuper{a,e}}
\author[M. R. Mousavi]{Mohammad Reza Mousavi\rsuper{f}}

\address{Saarland University --- Computer Science, Saarland Informatics Campus, Saarbr{\"u}cken, Germany}
\address{CISPA Helmholtz Center for Information Security, Saarbr{\"u}cken,  Germany}
\address{Automotive Powertrain Institute, htw saar --- University of Applied Sciences, Saarbrücken, Germany}
\address{Department of Computer Science, University of Sheffield, Sheffield, UK}
\address{Institute of Intelligent Software, Guangzhou, China}
\address{Department of Informatics, King's College London, London, UK}

%% required for running head on odd and even pages, use suitable
%% abbreviations in case of long titles and many authors:

%%%%%%%%%%%%%%%%%%%%%%%%%%%%%%%%%%%%%%%%%%%%%%%%%%%%%%%%%%%%%%%%%%%%%%%%%%%

%% the abstract has to PRECEDE the command \maketitle:
%% be sure not to issue the \maketitle command twice!

\begin{abstract}
We present a novel and generalised notion of doping cleanness for cyber-physical systems that
allows for  perturbing the inputs and observing the perturbed outputs both in the time---and value---domains.
We instantiate our definition using existing notions of conformance for cyber-physical systems.
\revised{As a formal basis for monitoring conformance-based cleanness, we develop the temporal logic \hyperstlstar, an extension of Signal Temporal Logics with trace quantifiers and a freeze operator.}
We show that our generalised definitions are essential in a data-driven method for doping detection and
apply our definitions to a case study concerning diesel emission tests.
\end{abstract}

\maketitle

% TODO: merge in single tex file for submission

\section{Introduction}%
\label{sec:into}
System doping, in our terminology, is an intentional intervention causing a change in the system's normal behaviour against the interests of the user or other stakeholders (such as the society at large). Examples of system doping are widespread and range from vendors' enforcing a monopoly on chargers and spare parts (by checking for and refusing third-party chargers and spare parts, respectively) to tampering with exhaust emission in order to detect and pass emission tests. Doping can be the result of embedding a piece of code or smuggling a piece of electronic circuit into the system and it can be caused by the original developers or by hackers. Software and system doping has been studied in the past couple of years and rigorous theories for it have been developed~\cite{BartheDFH16:isola,DBLP:conf/esop/DArgenioBBFH17,DBLP:conf/cpsweek/BiewerDH18}. These theories were subsequently adopted in order to detect doping, or formally, to check system cleanness~\cite{DBLP:conf/lpar/HermannsBDK18,DBLP:conf/qest/BiewerDH19} (corresponding to the absence of doping).

 In the present paper, we extend the theory of doping to the setting of cyber-physical systems (CPS) by exploiting the notions of conformance testing for CPS~\cite{abbas2014formal,Deshmukh17,khakpour2015notions}.  The existing theories of software doping define doping in terms of drastic deviations in output as a result of minor deviations in input, where the term ``deviation'' refers to differences in validity of propositions or values of variables. However, the current notions come short of properly dealing with the issues of retiming and delays, which are commonly present in the signals of CPS%, and measuring their doping effect
. We observe that this is an essential aspect of detecting doping for cyber-physical  systems: often the traces to be tested for doping have subtly different timing behaviour,
e.g., due to measurement and calibration errors or due to the slight deviations of human actors in acting upon the planned scenarios. The insufficient treatment of retiming and delays can both lead to false negatives, i.e., missing cases of doping, as well as false  positives, i.e., reporting spurious doping cases.

To address these issues, we exploit the notion of conformance to devise a general theory
of being clean from doping
and instantiate that theory with some existing notions of conformance for hybrid
systems. We show how these notions can account for retiming and lead to more precise
notions of cleanness.
Furthermore, we show how the retiming can be synchronized between input and output, leading to a refined notion of cleanness with a rigorous account of the relation between the retiming of input and output.

We illustrate the usefulness of our theory by empirical analysis of diesel engine exhaust emissions in the context of one of the official test cycles, the New European Driving Cycle (NEDC)~\cite{nedc}. In particular, we show that catering for retiming is essential in effectively exploiting the actual driving cycles for performing doping analysis. %due to human errors and hence,
We thus demonstrate that our new theory remedies a major shortcoming in the existing notions from the literature.
To facilitate the presentation, we use  throughout the remainder of this paper the following simple running example, which is inspired by our case study.

\begin{figure}[t]
\centering
\includegraphics[scale=.65, clip]{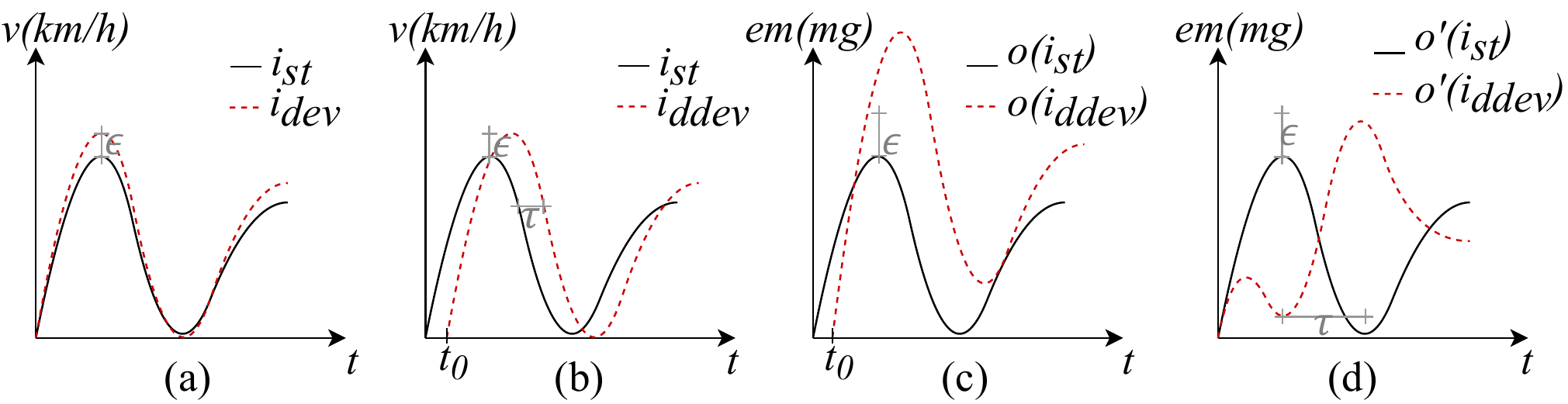}
\vspace{-.5cm}
\caption{Running Example: Specified (a) and Actual (b) Test Cycles and Emission Footprints obtained from Different (Fictitious) Vehicles (c) and (d).}%
\label{fig:running}
\end{figure}%
\begin{exa}\label{ex:intro}
Fig.~\ref{fig:running}.(a) shows two test cycles (evolution of speed over time), designed to detect whether the exhaust emission control of a particular vehicle is doped. The test cycle \rexIstA, depicted with a black solid line, is the standard one prescribed by the (fictitious) official regulation, while test cycle \rexIdevA, depicted by a red dotted line, is a slight deviation thereof. If the exhaust emissions measured during the test cycle \rexIdevA\ turn out to be significantly higher than the ones measured in test cycle \rexIstA, then we can conclude that the exhaust emission system is potentially doped, since it appears tailored to the standard test cycle.

Fig.~\ref{fig:running}.(b) addresses a notorious problem of testing cars: a human tester is supposed to drive the car as just described, however, she can do this only up to a certain imprecision. Assume her driving of \rexIdevA\ exhibits a slight time shift $\tau$ relative to the test cycle, as in \rexIdevB, while \rexIstB\ is being driven as intended.

The result of a test is the emission footprint measured at the exhaust pipe of the car. Fig.~\ref{fig:running}.(c) and Fig.~\ref{fig:running}.(d) show two different possible test results (obtained from different cars) for the scenario in Fig.~\ref{fig:running}.(b). Intuitively, the footprints in Fig.~\ref{fig:running}.(c) provide significant evidence for doping --- a slightly different test cycle has resulted in significantly larger footprint. However, due to the time shift on the input side in Fig.~\ref{fig:running}.(b) the point-wise difference of the two driven test-cycles has grown very large. As we show in the remainder of this paper, the existing theory of doping fails to detect such a clear evidence, due to the minor delay during the execution of the driving cycle. The emission footprint in Fig.~\ref{fig:running}.(d) is another (synthetic) example of a significant deviation which cannot be detected for the input in Fig.~\ref{fig:running}.(b) using existing theories; this latter footprint sheds some light on the intricate design decisions in the theory we develop in this paper.
\end{exa}

The contributions for this paper can be summarized as follows:
\begin{itemize}
\item We define a \emph{general notion of conformance} that can express different ways of comparing execution traces by allowing deviations both in value and in time;
\item We define a general \emph{notion of cleanness for hybrid systems}, and show that it subsumes the existing notion of robust cleanness~\cite{DBLP:conf/esop/DArgenioBBFH17};
\item We define the notion of synchronized retiming, which provides a rigorous tool for relating the retiming of input and output and use it to produce a refined notion of cleanness;
\item We provide a logical account of cleanness (based on the notion of hybrid conformance) by developing a temporal logic,  called \hyperstlstar, that extends Signal Temporal Logic (STL) with a freeze \revised{operator} and quantifiers over traces; and
\item We demonstrate the usefulness of the proposed generic framework by applying it to \emph{software doping tests} in the automotive domain, where we show that the new cleanness definition is able to flag a case of software doping that goes unnoticed when robust cleanness is used.
\end{itemize}

\noindent
This paper substantially extends the theoretical material and the experimental results published in the earlier conference publication~\cite{DimitrovaFORTE2020}. In particular, the following contributions of the present paper are new with respect to the earlier conference publication:
\begin{itemize}
\item the notion of cleanness with synchronised retiming in Section~\ref{sec:synchronized} and its application in Section~\ref{sec:case_study} for doping detection in our case study,
\item the logical approach to cleanness in Section~\ref{sec:timed-hyper}, the introduction of our logical formalism  \hyperstlstar, its monitoring, and its application to specify hybrid conformance,  and
%the application of the \hyperstlstar\ monitoring technique to our case study, and
\item the redesign of our experimental setup to comply with the NEDC test cycle (with preconditioning and control of ambient temperature), as well as several new experiments to make full use of our theories of retiming.  These new experiments led to much more substantial and decisive evidence for our case study.
\end{itemize}

\section{Related Work}%
\label{sec:rel}
The term ``software doping'' was coined around 2015~\cite{tagesschau} in media uncovering the diesel exhaust emissions scandal. An informal problem formulation~\cite{BartheDFH16:isola} pointed out the general phenomenon of intentionally added hidden software behaviour, which is not in the interest of the consumer.
Shortly after, this observation has been complemented by a set of formal \emph{cleanness} definitions~\cite{DBLP:conf/esop/DArgenioBBFH17} laying the theoretical foundations upon which formal methods to detect such software behaviour can be used.
It is possible to detect missing functionality and undesired existing functionality.
The definitions support both sequential programs and nondeterministic reactive programs.
To check satisfaction of the definitions, it is necessary to compare two (or more) execution traces of the same system.
Such properties are called \emph{hyperproperties}~\cite{ClarksonS08} (whereas classical properties are \emph{trace properties}).
Tool support for analysing hyperproperties typically requires high computational effort~\cite{ClarksonFKMRS14:post,FinkbeinerRS15:cav}.
There exist several temporal logics for analysing satisfaction of trace properties of various kinds of systems, one of them being \emph{Linear Temporal Logic} (LTL)~\cite{DBLP:conf/focs/Pnueli77} for systems producing outputs in discrete time steps and properties that do not consider the time passing between outputs.
LTL has been extended to the logic HyperLTL, which can express hyperproperties by allowing explicit quantification of execution traces in front of an LTL formula~\cite{ClarksonFKMRS14:post}.
Tools for model-checking boolean circuits, satisfiability and monitoring of HyperLTL specifications have been developed~\cite{agbo2016,brsibo2017,FinkbeinerRS15:cav,DBLP:conf/concur/FinkbeinerH16,DBLP:conf/cav/FinkbeinerHS17,DBLP:conf/rv/FinkbeinerHST17,DBLP:conf/tacas/FinkbeinerHST18,DBLP:conf/tacas/HahnST19}.

\emph{Signal Temporal Logic} (STL)~\cite{DBLP:conf/formats/MalerN04} is an extension of LTL that adds support for time constraints and real-valued signals.
Tools exist that automatically try to falsify STL formulas~\cite{DBLP:conf/cav/Donze10,DBLP:conf/tacas/AnnpureddyLFS11}.
There has been an extension of STL to HyperSTL in a similar fashion as it was done for HyperLTL~\cite{HyperSTL}.
The syntax of HyperSTL, however, is not able to express the cleanness definitions (for deterministic systems) in a way that allows (efficient) falsification. Freeze Temporal Logic~\cite{AlurH94} introduces a freeze quantifier to ``record'' the moment of time while evaluating a formula and later use the recorded time to measure time lapse. The concept of freeze quantifier  has been used to specify properties of data-dependent systems (using models such as register automata)~\cite{DemriL09} and dynamical systems~\cite{FreezeSTL,Deshmukh17}. In the present paper, we introduce the concept of freeze quantifier~\cite{FreezeSTL} into a logic that is inspired by HyperSTL~\cite{HyperSTL}. The combination turns out to be expressive enough to specify cleanness with respect to hybrid conformance. Moreover, we use a monitoring procedure inspired by an earlier extension of STL with freeze operator~\cite{RobustnessMonitoringFreezeSTL} to check for cleanness with respect to hybrid conformance.
\emph{Robust cleanness} is defined for distance functions on inputs and outputs~\cite{DBLP:conf/esop/DArgenioBBFH17}.
When used with temporal logics the distance functions are restricted to those compatible with the logics.
To be fully independent, robust cleanness analysis has been embedded into the theory of model-based testing~\cite{DBLP:conf/qest/BiewerDH19} with input-output conformance~\cite{Tretmans-thesis,DBLP:journals/cn/Tretmans96}.

Notions of conformance for discrete event systems have been discussed for almost a century. The earliest work on this topic dates back to 1960's when researchers studied model-based testing of digital circuits using  Finite State Machine models~\cite{hennie64,lee96}. Concurrency theory contributed ideas to this field, such as decoupling (i.e., removing the synchronised assumption between) inputs and outputs and observing failures to engage in a communication (and more specifically quiescence)~\cite{DeNicola84,Tretmans-thesis}. A theory of conformance testing for systems with continuous dynamics was developed by Michiel van Osch~\cite{vanOsch2006}; this theory did not gain much popularity in practice, partly because of its insufficient treatment of approximation (e.g., differences in values and retiming). Pappas and Girard~\cite{Girard08,Girard2011} proposed the use of Metric Bisimulation for conformance checking in dynamical systems and Pappas and Fainekos~\cite{Fainekos09} developed a falsification framework for the same purpose. This research led to two notions of conformance used in the present paper, namely hybrid conformance by Abbas and Fainekos~\cite{abbas2014formal} and Skorokhod conformance by Deshmukh, Majumdar, and Prabhu~\cite{Deshmukh17}.

%There, arbitrary distance functions, for the type of trajectories the underlying theory provides, can be used.
%If the distances are computable, real-world doping testing is possible with the provided testing algorithm.
%Variance in time can be introduced via the distance functions, for example by using hybrid closeness~\cite{DBLP:journals/corr/AbbasHFDKU14} or Skorokhod metric~\cite{DBLP:conf/hybrid/Davoren09}\Sebastian{Please check this reference}.
%Computational costs for these metrics become higher with increasing time tolerance\Sebastian{Must check}, which may make them a bottleneck of the testing procedure.
%For hybrid closeness, there can be a gain in efficiency by sticking to model-based testing theory that considers the time variance by design.
%This has been done in earlier work~\cite{DBLP:journals/scp/AraujoCMMS18}.
%The framework can, however, not directly be used for cleanness analysis, because it cannot check if the conformance of the inputs holds.
%Moreover, different conformances for input and output may be used.
%Those combinations must be studied specifically for the cleanness property, which is done in this work.

\section{Preliminaries}%
\label{sec:prelim}
\subsection{Semantic domain.}

In this section, we provide definitions regarding semantic domain, conformance, and robust cleanness.
We begin with the definition of our semantic domain, called generalised timed traces~\cite{Gazda2019}. This definition subsumes both discrete-time state sequences and continuous-time trajectories. A generalised timed trace is a function with a discrete or continuous domain (called time domain) and a co-domain which is a metric space. Intuitively, a generalized timed trace maps each element of its time domain to a state. We require  that the set of possible states is a metric space  since we study conformance notions that compare traces based on the distance between the states of the traces.

\begin{defi}\label{def:gtt}
	Let $(\caly,d_\caly)$ be a metric space. A ${\caly}$-valued \emph{generalised timed trace (GTT)} is a function
	$\mu: \calt \to \caly$ such that $\calt \subseteq \reals_{\geq 0}$.  We call $\calt$ the  \emph{time domain} of $\mu$, denoted  $\dom(\mu)$.
	$\mathit{GTT}(\caly)$ is the set of all ${\caly}$-valued generalised timed traces.
\end{defi}

For a GTT $\mu:\calt \to \caly$ and time $t_0 \in \calt$, by $\mu[\dots t_0]$ we denote the prefix of $\mu$ up to $t_0$, i.e., the restriction $\mu |_{t \in \calt: t\leq t_0}$; likewise, by  $\mu[t_s \dots t_e]$, we shall denote the restriction $\mu |_{t \in \mathcal{T} : t_s \leq t \leq t_e}$

A hybrid system is a mapping from generalised (input) timed traces to sets of generalised (output) \revised{timed} traces.

\begin{defi}%
	\label{def:hybridsystem}
 A ${\caly}$-valued hybrid system is a function $\hybrid : \mathit{GTT}(\caly)  \to$    $ \calp(\mathit{GTT}(\caly))$ such that for all $\mu \in  \mathit{GTT}(\caly)$  and all
	$\mu' \in \hybrid(\mu) $ it holds that $\dom(\mu') =\dom( \mu)$.
	We define  $\mathcal{H}(\caly)$ to be the set of all ${\caly}$-valued hybrid systems.

	In addition, we distinguish deterministic hybrid systems whose output values range over singleton sets only. In what follows, we  identify deterministic hybrid systems with functions of the type
	$ \mathit{GTT}(\caly)  \to \mathit{GTT}(\caly)$.
\end{defi}

For simplicity, we assume that the input and output domain are defined on the same metric spaces. The generalisation to different spaces is straightforward.

\subsection{Conformance relations.}
Recently, a number of notions of conformance for cyber-physical systems have been proposed~\cite{aerts2016mbt,khakpour2015notions}. It turns out that these notions, two of which are quoted below, can provide a rigorous basis for doping detection.

Note that throughout the paper, the variables $\tau$ and $\epsilon$ (with possible subscripts) always range over non-negative real numbers.

\begin{defi}%
	\label{def:confrel}
	 We say that $\caly$-valued GTTs $\mu_1:\calt_1 \to \caly$ and $\mu_2:\calt_2 \to \caly$ are:

	\begin{itemize}
		\item \emph{trace conformant} with tolerance threshold for signal value $\epsilon$, notation $\tracec_{\epsilon}(\mu_1,\mu_2)$, if $\calt_1=\calt_2$ and for all $t \in \calt_1$, $\dy(\mu_1(t),\mu_2(t)) \leq \epsilon$

		\item \emph{\hybconformant} with thresholds $\tau$ and $\epsilon$, denoted $\hybc_{\te}(\mu_1,\mu_2)$, if:
		\begin{itemize}
			\item $\forall t_1 \in \calt_1 \,\exists t_2 \in \calt_2:\,
		|t_2-t_1|\leq \tau \,\wedge\,
		\dy(\mu_2(t_2),\mu_1(t_1))\leq\epsilon$
		\item $\forall t_2 \in \calt_2 \,\exists t_1 \in \calt_1:\, |t_1-t_2|\leq\tau \,\wedge\, \dy(\mu_1(t_1),\mu_2(t_2))\leq\epsilon$
	\end{itemize}
\item \emph{Skorokhod conformant} with tolerance thresholds $\tau$ and $\epsilon$, notation $\skorc_{\te}(\mu_1,\mu_2)$, if $\calt_1$ and $\calt_2$ are intervals and there is a strictly increasing continuous bijection $r:\calt_1 \to \calt_2$ called retiming, such that:
\begin{itemize}
	\item for all $t\in \calt_1$, $|r(t)-t| \leq \tau$, and
	\item for all $t\in \calt_1$, $\dy(\mu_1(t),\mu_2(r(t))) \leq \epsilon$.
\end{itemize}
	\end{itemize}

\end{defi}

\noindent
We show in the proposition below and also in our generalisation results in Section~\ref{subsec:retiming}, that these notions are closely related. However, they also have some fundamental differences, that can be illustrated using the example in Fig.~\ref{fig:running}.

\begin{exa}\label{ex:conformance}
Consider again the
example shown in Fig.~\ref{fig:running}.

We can see that in Fig.~\ref{fig:running}.(a) \rexIstA\ and \rexIdevA\
are trace conformant with value threshold $\epsilon$,
as they only exhibit point-wise deviations by values less than $\epsilon$.
In contrast, \rexIstB\ and \rexIdevB\ in Fig.~\ref{fig:running}.(b)
are not  trace conformant, yet they are hybrid conformant with time and value margins $\tau$ and $\epsilon$, respectively. The key difference is that the
inputs
depicted in
Fig.~\ref{fig:running}.(b) are very different if compared point-wise, but if one allows for retiming, they are close enough in value after retiming.

The outputs \rexOstD\ and \rexOdevD\
in Fig.~\ref{fig:running}.(d) illustrate the fundamental difference between hybrid and Skorokhod conformance: although the order of rising and falling signals are reversed in the two  trajectories, they are still hybrid conformant, because hybrid conformance disregards the order. However, Skorokhod conformance requires an order-preserving retiming, and hence distinguishes these two trajectories. On the other hand, such retiming exists, e.g., for
\rexIstB\ and \rexIdevB\
in Fig.~\ref{fig:running}.(b), witnessing their Skorokhod conformance.
\end{exa}

We shall use the following notation. We write $\conf_1 \sqsubseteq \conf_2$ whenever for all $\mu_1:\calt_1 \to \caly$ and $\mu_2:\calt_2 \to \caly$, we have $
\conf_1(\mu_1,\mu_2)\implies
\conf_2(\mu_1,\mu_2)
$.
We write $\conf_1 \sqsubset \conf_2$ whenever $\conf_1 \sqsubseteq \conf_2$ and
$\neg \, \conf_2 \sqsubseteq \conf_1$.

\begin{prop}%
\label{prop:relationships}
For any $\tau,\epsilon \in \reals_{\geq 0}$, the following relations hold:
\begin{center}
\[
\tracec_{\epsilon}\sqsubset \skorc_{\te}\sqsubset
\hybc_{\te}
\]
\end{center}
\end{prop}

\subsection{Robust cleanness.}
We shall now state the original definition of robust cleanness from~\cite{DBLP:conf/esop/DArgenioBBFH17}, adapted to our framework of hybrid systems. It is based on Definition 7 and Proposition 19 from~\cite{DBLP:conf/esop/DArgenioBBFH17}; the phrasing below abstracts from the so-called parameters of interest
and standard inputs. Moreover it is cast in the setting of generalised timed traces rather than discrete-step programs, and stated using trace conformance with different thresholds for inputs and outputs, $\kappa_I$ and $\kappa_O$.

Intuitively, a hybrid system is robustly clean if for every pair of input prefixes on which no difference in the inputs exceeding $\kappa_I$ has occurred so far (i.e., all sub-prefixes are trace conformant), the corresponding sets of output prefixes are also conformant with respect to $\kappa_O$. As we consider nondeterministic systems, Hausdorff distance is used to compare sets of outputs (see~\cite{DBLP:conf/esop/DArgenioBBFH17} for details).

\begin{defi}\label{def:robust_clean}
	A hybrid system $H$ is robustly clean, denoted\\ $\robclean(\kappa_I,\kappa_O)$, whenever:	\\
	$\begin{array}{l}
	\forall i_1, i_2 \in \gtt(\caly):
	\forall t\in\dom(i_1) \cup \dom(i_2):
	\\
	\big(\forall t' \leq t:\,
		\tracec_{\kappa_I}(i_1[\dots t'], i_2[\dots t'])
	\,\implies\,\\
%	\phantom{\forall t\in\dom(i_1) \cup \dom(i_2):  }
	\quad\ \big((\forall o_1\in\hybrid(i_1) \, \exists o_2 \in \hybrid(i_2):\,\,
		\tracec_{\kappa_O}(o_1[\dots t], o_2[\dots t]))\;\wedge \\
%	\phantom{\forall t\in\dom(i_1) \cup \dom(i_2):  }
	\quad\ \ (\forall o_2\in\hybrid(i_2) \, \exists o_1 \in \hybrid(i_1):\,\,
	\tracec_{\kappa_O}(o_1[\dots t], o_2[\dots t]))\big)
	\end{array}
	$

\end{defi}

Note that in the above definition we do not require that $dom(i_1) = \dom(i_2)$. In practice, robust cleanness is typically applied to pairs of traces that are both defined over $\mathbb N$. Here, however, for the sake of generality we impose no such restriction. In particular, when the time domains of two traces are different, for example disjoint,  the predicate $\robclean$ will trivially evaluate to $\mathit{true}$.

\begin{exa}\label{example:robust-clean}
Consider the traces depicted in Fig.~\ref{fig:running}. The input prefixes  \rexIstB\ and \rexIdevB\ are given in Fig.~\ref{fig:running}.(b), and the corresponding pair of outputs is shown in  Fig.~\ref{fig:running}.(c).
\revised{For $t \leq t_0$,  we have $\rexIdevB(t)=0$ and $\rexOdevC(t)=0$.}
The trace \rexIstB\ results in output \rexOstC\ and \rexIdevB\ results in \rexOdevC.
\revised{Suppose that  $\epsilon < |\rexIstB(t_0) - \rexIdevB(t_0)|$,  and for every $t < t_0$ it holds that $|\rexOstC(t) - \rexOdevC(t)| \leq |\rexIstB(t) - \rexIdevB(t)|$. % chktex 8
Furthermore,   $\epsilon < |\rexOstC(t_1) - \rexOdevC(t_1)|$ at some time $t_1 \geq t_0$. % chktex 8
Consider Def.~\ref{def:robust_clean} instantiated with $\kappa_I = \kappa_O = \epsilon$.
Clearly,  for every $t < t_0$ the implication in  Def.~\ref{def:robust_clean} is satisfied, since
$|\rexOstC(t') - \rexOdevC(t')| \leq |\rexIstB(t') - \rexIdevB(t')|$ for all $t' \leq t < t_0$,  and % chktex 8
 $\kappa_I = \kappa_O$.
For $t \geq t_0$ the implication is true as well:
for $t' < t_0$ the reasoning is as above,  and
for all $t' \geq t_0$ the left-hand side of the implication is false.
Hence, regardless of the difference in the output values at $t_1$, this pair of inputs satisfies the condition of $\robclean(\epsilon,\epsilon)$, and, if these are the only traces in a hybrid system $H$ then we can conclude that $H$ is $\robclean(\epsilon,\epsilon)$.
}
\end{exa}

\section{Conformance-Based Cleanness}%
\label{sec:definition}
We now define a general notion of conformance-based cleanness and provide two instantiations based on the conformance notions defined in the previous section.

\subsection{Motivation}%
\label{subsec:def-motivation}
The need for considering disturbance in time as well as in value is motivated by our running example from Fig.~\ref{fig:running}.
One of the challenges in performing doping tests for cyber-physical systems is that in such systems timing is rarely perfectly precise, due to imprecision in measurements, or caused by the interaction with the physical world. As illustrated in Example~\ref{ex:intro}, for instance, when checking for software doping in a car~\cite{DBLP:conf/qest/BiewerDH19}, the input to the system is the value of the car's speed over time, which is under the control of a driver, and can thus vary from one execution to the other, even if the driver is trying to execute the same input sequence. Clearly, those variations can be in value, as well as in time.

\begin{exa}\label{ex:input-delay}
  Consider the test setup sketched in Fig.~\ref{fig:running}. There, \rexIstB\ and \rexIdevB, depicted in Fig.~\ref{fig:running}.(b) define the speed of a car as a function of time. These two input sequences follow a trajectory of values differing by a small margin $\epsilon$ (the difference in value allowed by the standard defining the doping tests), but also shifted by a small unit of time $\tau$. Observe further that $|\rexIstB(t_0) - \rexIdevB(t_0)| \gg \epsilon$. Thus, without allowing for deviations in time when comparing these input sequences, they will be considered sufficiently different, and as a result their respective exhaust emission outputs will fall out of the comparison when checking for doping according to Def.~\ref{def:robust_clean}, even if the % chktex 8
  outputs $\hybrid(\rexIstB(t))$ and $\hybrid(\rexIdevB(t))$ are vastly different, as depicted in Fig.~\ref{fig:running}.(c). This results in a false negative, i.e., failing to detect a clearly doped system.
\end{exa}

In the above example, we demonstrated that not accounting for timing disturbances when relating input trajectories can result in false negatives in doping detection. Dually, using the traditional comparison for output traces can result in false positives by requiring overly strict matching of outputs.

The above  example motivates the need to
account for timing deviations in trajectories. Intuitively, for input trajectories this relaxation results in considering more traces as conforming, and thus enforcing more comparisons when checking if a system is clean. For output trajectories this means relaxing the conformance requirement by considering two output sequences as conforming even if their values are not perfectly aligned in time.
Furthermore, different types of timing deviations need to be considered in different scenarios, for example, depending on whether the order in which values occur is important or not.

\begin{exa}\label{ex:output-skorokhod}
Consider the testing workflow from Example~\ref{ex:intro} and Fig.~\ref{fig:running}, where inputs \rexIstB\ and \rexIdevB\ are passed to a car.
In the second experiment, depicted in Fig.~\ref{fig:running}.(d), \revised{the outputs of the car are} \rexOstD\ and \rexOdevD, which are hybrid conformant for $\epsilon$ and $\tau$. Hence,  this observation of the system is classified as clean under hybrid output conformance.
However, the output \rexOdevD\ is clearly suspicious, as the values in \rexOdevD\ and \rexOstD\  are reversed.
This motivates considering conformance notions that require retimings to be order-preserving. Indeed, using Skorokhod conformance we can detect that the system is doped.
\end{exa}

The above examples show that in order to be useful in a diverse set of applications, a software cleanness  theory should  allow for using a variety of conformance notions. To this end, we next take a more general view on conformance notions, in order to be able to develop a generic conformance-based cleanness framework.

\subsection{Retimings and a more generic view on conformance notions}%
\label{subsec:retiming}
So far, we have defined three specific notions of conformance which either coincide, or are closely inspired by ones that have appeared in the literature. In order to define a general framework for cleanness, we also wish to treat notions of conformance in a more generic manner. To this end, we propose an abstract definition of conformance predicates. As conformance predicates admit variations in time, as well as in value,  our definition is based on \emph{retimings}, a device that will play a key role in the context of this work.
In its general form a retiming is a pair of functions between two time domains. Intuitively, given two GTTs, a retiming will define a mapping from points in each of the traces to points in the other trace. Note that in general the mappings are not required to be injective; this way we can cater for notions of conformance allowing for the so-called local disorder phenomenon (in particular hybrid conformance --- see Proposition~\ref{prop:confviaret}).

\begin{defi}
	A \emph{retiming} is a pair of functions between two time domains, i.e., a pair of the form $(r_1,r_2)$, where $r_1: \calt_1 \to \calt_2$ and $r_2: \calt_2 \to \calt_1$,  with time domains  $ \calt_1,\calt_2 \subseteq \reals_{\geq 0}$. %,  where
Given two time domains $\calt_1$ and $\calt_2$, we denote the set of all retimings between
$\calt_1$ and $\calt_2$
with $\retime(\calt_1,\calt_2)$.
\end{defi}

Retiming is explicitly  present in the definition of Skorokhod conformance; there, each Skorokhod retiming is required to be a strictly increasing continuous bijection. We can express a Skorokhod retiming $r$ as an instance of  our definition as the pair $(r,r^{-1})$. In fact, one can also define hybrid conformance, as well as a whole class of conformance notions, using a suitable \emph{family} of retimings.

A family of retimings $\ret$ can be further constrained by $\tau$ to a subset $\ret_\tau$ of $\ret$ containing only functions that shift time by at most $\tau$ time units. In order to use a family of retimings for concrete sequences $\mu_1$ and $\mu_2$, it is necessary to consider only functions that match the domains of the sequences. This leads to a generic notion of conformance  associated with a given family of retimings $\ret$, a given time threshold $\tau$ and a given value threshold $\epsilon$.

\begin{defi}%
	\label{def:ret-induced-conf}
	Let $\ret$ be a family of retimings, and let \[
	\begin{array}{lll}
	\ret_{\tau} &\!\eqdef\!&
	\{(r_1,r_2) \in \ret \,\,|\,\,
	%	\\& &
	\forall t \in \dom(r_i):\, \,|r_i(t)-t| \leq \tau \,\,\,(i=1,2)
	\},	\\
	\ret_{\tau}(\calt_1,\calt_2) &\!\eqdef\!&
\ret_{\tau} \cap \retime(\calt_1,\calt_2).
\end{array}
\]
	A \emph{conformance notion} with time threshold $\tau$ and value threshold $\epsilon$ induced by $\ret$ is a predicate $\conf^{\,\ret}_{\tau,\epsilon}$ on pairs of GTTs such that, for
	$\mu_1:\calt_1 \to \caly$, 	$\mu_2:\calt_2 \to \caly$:
	\[
	\begin{array}{rl}

	\conf^{\,\ret}_{\tau,\epsilon}(\mu_1,\mu_2) \iff
	 \exists (r_1,r_2) \in \ret_{\tau}(\calt_1,\calt_2): &
	 	\forall t \in \calt_1:\, \dy(\mu_1(t),\mu_2 \circ r_1(t)) \leq \epsilon
	\\	\,\wedge\, &
	\forall t \in \calt_2:\,  \dy(\mu_2(t),\mu_1 \circ r_2(t)) \leq \epsilon.
	\end{array}
	\]

	A conformance notion with \emph{unbounded time deviation}  and value threshold $\epsilon$ induced by $\ret$ is a predicate $\conf^{\,\ret}_{\infty,\epsilon}$ on pairs of GTTs such that, for
	$\mu_1:\calt_1 \to \caly$, 	$\mu_2:\calt_2 \to \caly$:
	\[
	\begin{array}{rl}
	\conf^{\,\ret}_{\infty,\epsilon}(\mu_1,\mu_2) \iff
	 \exists (r_1,r_2) \in \ret(\calt_1,\calt_2): &
	 	\forall t \in \calt_1:\, \dy(\mu_1(t),\mu_2 \circ r_1(t)) \leq \epsilon
	\\	\,\wedge\, &
	\forall t \in \calt_2:\,  \dy(\mu_2(t),\mu_1 \circ r_2(t)) \leq \epsilon.
	\end{array}
	\]
\end{defi}

Unless we state explicitly otherwise, we consider conformance notions with finite time threshold $\tau \in \realsnn$.
Using the above definition, we can easily express the specific notions of conformance defined in the previous section by selecting a suitable family of retimings.

\begin{prop}%
	\label{prop:confviaret}
	The conformance predicates below coincide with the notions of conformance induced by the corresponding families of retimings:

	\begin{itemize}
		\item $\tracec_{\epsilon}$ is induced by the family of retimings containing only identity functions:
		$\ret_{\mathsf{id}} = \{\idret \mid \mathsf{id} : \calt \to \calt \text{is the identity on some }\calt \subseteq \mathbb{R}_{\geq 0}\}$.
		\item $\skorc_{\tau,\epsilon}$ is induced by the family of retimings\\
		$\ret = \{(r,r^{-1}) \mid r \textrm{ is a strictly increasing continuous bijection} \myomit{and \max((r,r^{-1})) \leq \tau}\}$.
		\item $\hybc_{\tau,\epsilon}$ is induced by pairs of arbitrary functions.
	\end{itemize}
\end{prop}

\noindent
Definition~\ref{def:ret-induced-conf} also enables us to define other notions  of conformance, such as, for instance a ``shift conformance'', which, intuitively, shifts all time points by a given constant $c \in\reals$, i.e.,
$\ret_c = \{(r,r^{-1}) \mid \myomit{r : \realsnn \to \realsnn}\, r(t)=t+c \}$.

Next, we define a generic notion of cleanness, parametrised by conformance predicates for the input and for the output traces.  Instantiating these predicates with existing or new conformance notions, yields different conformance-based notions of cleanness that can capture a variety of cleanness specifications.

\subsection{Definition of Conformance-based Cleanness}\label{subsec:defs}
We now extend the notion of robust cleanness~\cite{DBLP:conf/esop/DArgenioBBFH17} to allow for ``small'' variations in time, in addition to the variations in value. To this end, the new notion makes use of two conformance predicates, one that postulates when two input traces should be considered close enough, and another one that specifies when two output traces are close enough.

Our starting point, the notion of robust cleanness in Definition~\ref{def:robust_clean}, is based on comparison of matching prefixes of a pair of input traces and the corresponding prefixes of the associated output traces. As we now want to accommodate for distance in time, we (1) compare prefixes using a conformance relation, and (2) allow for variation in the length of the compared prefixes that is within the corresponding time-distance threshold. More precisely,  when comparing two prefixes, we allow for discarding start and end segments of length at most $\tau$.

This intuition is formalized by the predicate $\prefixconf$ for relaxed comparison of GTT prefixes using a notion of conformance $\conf$ with tolerance threshold $\tau$ for time disturbance. We use cascaded notation to define $\prefixconf$ as a higher-order function taking $\conf$ as its first argument.  The predicate $\prefixconf$ compares two prefixes $\mu_1$ and $\mu_2$ by requiring that there exist traces  $\mu_1[t_1^s\dots t_1^e]$ and $ \mu_2[t_2^s\dots t_2^e]$  obtained from them, that are conformant with respect to $\conf$. These traces are obtained by possibly removing a sub-prefix of length at most $\tau$, and/or removing extending with a suffix of length at most $\tau$.

\begin{defi}\label{def:prefixconf}
Let $\conf$ be a  notion of conformance on GTTs with tolerance threshold $\tau \in \realsnn$ for time disturbance. For any pair of GTTs $\mu_1:\calt_1 \to \caly$, $\mu_2:\calt_2 \to \caly$, and $t \in \calt=\calt_1 \cup \calt_2$,  the predicate $\prefixconf$ is defined as:
 \[
\begin{array}{lll}
 \prefixconf (\mu_1,\mu_2,t) &\!\!\!\!\!\ \iff\ \!\!\!\!\!\!&
%\prefix(\conf,\mu_1,\mu_2,t) &\,\iff\,&
\exists t_1^s \in\! [0,\tau] \cap \calt_1,
	\exists t_1^e \in\!
	[t-\tau,t+\tau] \cap \calt_1,\\
	& & \exists t_2^s  \in\! [0,\tau] \cap \calt_2, \exists t_2^e \in\!
	[t-\tau,t+\tau] \cap \calt_2\!\!: \\
	& &
	\conf(\mu_1[t_1^s\dots t_1^e], \mu_2[t_2^s\dots t_2^e]).
\end{array}
\]
\end{defi}

For conformance notions with unbounded timing deviation $\prefixconf$ coincides with $\conf$.

The predicate $\prefixconf$ provides a generic notion of prefix-con\-for\-mance. By instantiating it with conformance relations $\conf\subi$ and $\conf\subo$ for input and output traces respectively, we define the notion of $(\conf\subi,\conf\subo)$-cleanness.

For deterministic systems $(\conf\subi,\conf\subo)$-cleanness requires that for all pairs of input prefixes for which all sub-prefixes are prefix-conformant w.r.t.\ $\conf\subi$, the corresponding pair of output prefixes are prefix-conformant w.r.t.\ $\conf\subo$.

\begin{defi}%
	\label{def:confclean-det}
	A deterministic system $\hybrid$ is $(\conf\subi,\conf\subo)$-clean if
	\[\begin{array}{lll}
	\forall i_1, i_2 \in & \gtt(\caly) :\,
	\forall t\in\dom(i_1) \cup \dom(i_2): & \\
	&	(\forall t' \leq t:\,
	\prefixconf\subi(i_1, i_2,t')\big)
	\,\implies\,%\\&
\prefixconf\subo(\hybrid(i_1), \hybrid(i_2),t). &
	\end{array}
	\]
\end{defi}

\noindent The above definition naturally generalises to nondeterministic hybrid systems, by comparing  sets of possible output prefixes using Hausdorff distance as in~\cite{DBLP:conf/esop/DArgenioBBFH17}.

\begin{defi}%
\label{def:confclean-nondet}

	A system $\hybrid$ is $(\conf\subi,\conf\subo)$-clean if
	\[\begin{array}{l}
	\forall i_1, i_2 \in \gtt(\caly):
	\forall t\in\dom(i_1) \cup \dom(i_2): \\ \big(\forall t' \leq t:\,
	\prefixconf\subi(i_1, i_2,t')\big)
	\,\implies\,\\
%	\phantom{\forall t\in\dom(i_1) \cup \dom(i_2):  }
	\qquad\ \big((\forall o_1\in\hybrid(i_1) \, \exists o_2 \in \hybrid(i_2):\,\,
	\prefixconf\subo(o_1, o_2,t ))\;\wedge \\
	%\phantom{\forall t\in\dom(i_1) \cup \dom(i_2):  }
	\qquad\ \ (\forall o_2\in\hybrid(i_2) \, \exists o_1 \in \hybrid(i_1):\,\,
	\prefixconf\subo(o_1, o_2,t ))\big).
	\end{array}
	\]
\end{defi}

Robust cleanness~\cite{DBLP:conf/esop/DArgenioBBFH17} can be now formulated as conformance-based cleanness, which establishes that  $(\conf\subi,\conf\subo)$-cleanness is a generalisation. Using hybrid conformance, we define hybrid-conformance cleanness, and similarly, plugging in Skorokhod conformance, we define Skorokhod-conformance cleanness. Formally:
\begin{itemize}
\item A hybrid system $H$ is robustly clean, denoted  $\robclean(\kappa_I,\kappa_O)$, if and only if  $H$ is $(\tracec_{\kappa_I},\tracec_{\kappa_O})$-clean.
\item A  hybrid system $H$ is \emph{\hybcleanname with conformance thresholds $(\tau_I,\epsilon_I,$ $\tau_O,\epsilon_O)$}, which we denote by $\hybclean(\tau_I,\epsilon_I,\tau_O,\epsilon_O)$, if and only if $H$ is $(\hybc_{\tau_I,\epsilon_I},$  $\hybc_{\tau_O,\epsilon_O})$-clean. % chktex 40
\item A  hybrid system $H$ is \emph{\skorcleanname with conformance thresholds $(\tau_I,\epsilon_I,\tau_O,\epsilon_O)$}, denoted $\skorclean(\tau_I,\epsilon_I,\tau_O,\epsilon_O)$, if and only if $H$ is $(\skorc_{\tau_I,\epsilon_I},$ $\skorc_{\tau_O,\epsilon_O})$-clean. % chktex 40
\end{itemize}

\subsection{Properties}
We will now
establish some key relations between the cleanness notions defined previously. We begin by lifting the implication between conformance relations to implication between cleanness notions defined using those relations.
\begin{prop}
Suppose that $\conf^{\,1}_I \sqsupseteq \conf^{\,2}_I$ and $\conf^{\,1}_O \sqsubseteq \conf^{\,2}_O$. Then for any system $H$:
%\begin{center}
$H$ is $(\conf^{\,1}_I,\conf^{\,1}_O)$-clean $\,\implies\,$
$H$ is $(\conf^{\,2}_I,\conf^{\,2}_O)$-clean.
%\end{center}
\end{prop}

The proposition above has two important corollaries. The first one explains the relationships between the original robust cleanness, and notions of cleanness based on Skorokhod conformance and hybrid conformance, in particular stating the conservative generalisation property for the latter notions. The second corollary compares cleanness notions with different conformance thresholds.

\begin{cor}
For all $\tau_I,\tau_O,\epsilon_I,\epsilon_O \in \mathbb R_{\geq 0}$, the following implications hold:
{
	\begin{enumerate}
\item	$\robclean(\epsilon_I,\epsilon_O) \implies$
	$\skorclean(0,\epsilon_I,\tau_O,\epsilon_O) \implies$
	$\hybclean(0,\epsilon_I,\tau_O,\epsilon_O)$,%\\
	\item $\hybclean(\tau_I,\epsilon_I,0,\epsilon_O)\,\implies\,$
	$\skorclean(\tau_I,\epsilon_I,0,\epsilon_O)\,\implies\,$
	$\robclean(\epsilon_I,\epsilon_O)$.
	\end{enumerate}
}
\noindent	Also,
%		\begin{center}
			{%\small
		$\robclean(\epsilon_I,\epsilon_O)=\skorclean(0,\epsilon_I,0,\epsilon_O)=
		\hybclean(0,\epsilon_I,0,\epsilon_O)$}
%	\end{center}
and hence $\skorclean$ and $\hybclean$ are conservative extensions of robust cleanness.
\end{cor}

\begin{cor}
For all $\epsilon\subi,\epsilon'\subi, \epsilon\subo,\epsilon'\subo,\tau\subi,\tau'\subi, \tau\subo,\tau'\subo$ that satisfy the inequalities\\
$\epsilon'\subi \leq \epsilon\subi,\quad
\tau'\subi \leq \tau\subi,\quad
\epsilon'\subo \geq \epsilon\subo,\quad
\tau'\subo \geq \tau\subo$
the following implications hold:

	\begin{enumerate}
\item	$\robclean(\epsilon\subi,\epsilon\subo) \,\implies\,$
	$\robclean(\epsilon'\subi,\epsilon'\subo)$;
\item	$\hybclean(\epsilon\subi,\tau\subi,\epsilon\subo,\tau\subo) \,\implies\,$
	$\hybclean(\epsilon'\subi,\tau'\subi,\epsilon'\subo,\tau'\subo)$;
\item	$\skorclean(\epsilon\subi,\tau\subi,\epsilon\subo,\tau\subo) \,\implies\,$
	$\skorclean(\epsilon'\subi,\tau'\subi,\epsilon'\subo,\tau'\subo)$.
	\end{enumerate}
\end{cor}

\begin{exa}
Consider the testing workflow in Fig.~\ref{fig:running}. The inputs passed to a car are \rexIstB\ and \rexIdevB, depicted in Fig.~\ref{fig:running}.(b). One of the test results is presented in Fig.~\ref{fig:running}.(c), where \rexIstB\ reveals output \rexOstC\ and \rexIdevB\ reveals \rexOdevC.
We assume that  $\epsilon < |\rexIstB(t_0) - \rexIdevB(t_0)|$ and $\epsilon < |\rexOstC(t_1) - \rexOdevC(t_1)|$ at some time $t_1 \geq t_0$.
\begin{itemize}
\item \revised{As we saw in Example~\ref{example:robust-clean},  for inputs \rexIstB\ and \rexIdevB, the car that emits the outputs depicted in Fig.~\ref{fig:running}.(c) is deemed $\robclean(\epsilon,\epsilon)$.} Note, that in the presence of other inputs the car used for testing might not be $\robclean(\epsilon,\epsilon)$.
\item As explained in Example~\ref{ex:conformance}, \rexIstB\ and \rexIdevB\ are hybrid conformant for $\epsilon$ and $\tau$, i.e., the predicate $\prefixconf_I$  on the left-hand side of the implication in Def.~\ref{def:confclean-det} holds. $\prefixconf_O$, however, fails at time $t_1$ for signals \rexOstC\ and \rexOdevC.
Hence, the system tested in Fig.~\ref{fig:running}.(c) is not $\hybclean(\epsilon,\tau,\epsilon,\tau)$.
\end{itemize}

\end{exa}

\noindent
We now discuss  testing and falsification of conformance-based cleanness.
For systems with discrete time domains the existing methods for verifying~\cite{DBLP:conf/esop/DArgenioBBFH17} or testing~\cite{DBLP:conf/qest/BiewerDH19} robust cleanness can be readily applied.

In the case of hybrid cleanness, existing methods for testing hybrid conformance, such as~\cite{AbbasHFDKU14} and~\cite{AraujoCMMS18} can be extended to testing and falsification of hybrid cleanness of hybrid systems consisting of traces with finite time domains. Methods for checking Skorokhod conformance were presented in~\cite{Deshmukh17}. Due to the quantification over all time-points $t'$ in our Definition~\ref{def:confclean-det} and Definition~\ref{def:confclean-nondet}, it is not clear how to  directly extend them to testing Skorokhod cleanness.

\section{Cleanness with Synchronized Retimings}%
\label{sec:synchronized}
\subsection{Practical motivation}
An intuitive and useful notion of doping cleanness should capture precisely what we expect from a clean \revised{system} subject to disturbances in time and value. In this regard, one can observe that even the more discriminating $\skorclean$ predicate has certain drawbacks. The following example motivates why one may want to resort to the finer definition to be proposed in this section.

\begin{exa}\label{ex:synchronized}
Consider the scenario of particle emission cleanness presented in Example~\ref{ex:conformance} and the input (velocity)- and output trajectories depicted in
Fig.~\ref{fig:weneedsync}.  Assume that for some input trajectory $i_1$, the vehicle shows the output (emission) profile $o_1$; for a second input $i_2(t)=i_1(t-\tau)$, consider two possible \revised{output} trajectories: one output is $o_1(t-\tau)$, i.e., it is shifted in the same manner as input; this is assumed to be the best response to $i_2$. The other output is of the form $o_1(t-\tau-\delta)$, where $\delta>0$ can be arbitrarily small, i.e., it is the optimal output with an arbitrary small shift to the right.
Skorokhod-conformance cleanness with $\tau_I=\tau_O$ would accept the first output, but it would reject the second one.
A potential solution could be to  increase the value of $\tau_O$ so that it is significantly larger than $\tau_I$, but this increases the imprecision by  accepting too many trajectories shifted far to the left from $o_1(t-\tau)$.

Intuitively, when the input shifts by some $\tau$, we would like to compare the corresponding output trajectory with the one that is shifted accordingly. In the above-mentioned case, one would therefore ideally like to perform conformance check of output against $o_1(t-\tau)$, rather than $o_1(t)$.
\end{exa}

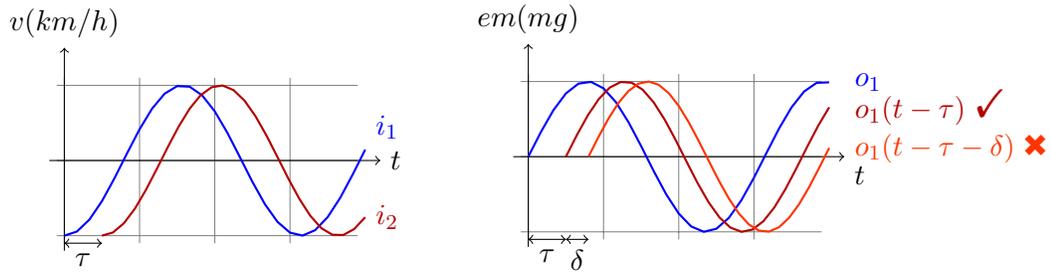
\begin{figure}[h!]
\begin{minipage}{0.4\textwidth}
\begin{tikzpicture}[domain=0:4]
\draw[very thin,color=gray] (-0.1,-1.1) grid (3.9,1.1);
\draw[->] (-0.2,0) -- (4.2,0) node[right] {$t$};
\draw[->] (0,-1.2) -- (0,1.5) node[above] {${v(km/h)}$};
\draw[color=blue,thick]   plot (\x,{-cos(2*\x r)})   node[above right] {$i_1$};
\draw[domain=0.5:4,color=darkred,thick]   plot (\x,{-cos(2*(\x-0.5) r)})   node[right\myomit{,yshift=-.6cm}] {$i_2$};

\draw[<->] (0,-1.1) -- (0.5,-1.1) node[below,xshift=-0.25cm] {$\tau$};
\end{tikzpicture}
\end{minipage}
\begin{minipage}{0.4\textwidth}
\begin{tikzpicture}[domain=0:4]
\draw[very thin,color=gray] (-0.1,-1.1) grid (3.9,1.1);
\draw[->] (-0.2,0) -- (4.2,0) node[below right] {$t$};
\draw[->] (0,-1.2) -- (0,1.5) node[above] {${em(mg)}$};
\draw[color=blue,thick]   plot (\x,{sin(2*\x r)})   node[right,xshift=0.2cm] { $o_1$};
\draw[domain=0.5:4,color=darkred,thick]   plot (\x,{sin(2*(\x-0.5) r)})   node[right,xshift=0.2cm\myomit{yshift=-.6cm}] { $o_1(t-\tau)$~~~~~~~\Checkmark};
\draw[domain=0.8:4,color=lightred,thick]   plot (\x,{sin(2*(\x-0.8) r)})   node[right,xshift=0.2cm\myomit{,yshift=-.6cm}] { $o_1(t-\tau-\delta)$~~{\scriptsize\XSolidBold}};
\draw[<->] (0,-1.1) -- (0.5,-1.1) node[below,xshift=-0.25cm] {$\tau$};
\draw[<->] (0.5,-1.1) -- (0.8,-1.1) node[below,xshift=-0.15cm] {$\delta$};
\end{tikzpicture}
\end{minipage}
\caption{Imprecision problem in Skorokhod-conformance cleanness without synchronisation of retimings. While $o_1(t-\tau)$ is the best expected response to $i_2$, no trajectory to the right of $o_1(t-\tau)$ is accepted when $\tau_I=\tau_O=\tau$.}%
\label{fig:weneedsync}
\end{figure}

\subsection{Formal theory of synchronized retiming.}
In order to alleviate this imprecision, we propose a definition of conformance-based cleanness with synchronised retimings, in which we do not check the conformance of the resulting outputs directly, but rather check conformance of each of the outputs against the transformation of the other output with the  retiming that is expected, based on the retiming of the corresponding input. Note that the expected retiming of the output is not always precisely the same as that of the input. Instead, we assume that the set of expected output retimings to a given input retiming is available through a synchronisation function.

As mentioned earlier, we can include in the set of conforming output trajectories the best expected response  $o_1(t-\tau)$ by allowing a sufficiently large $\tau\subo$, but this comes at the price of introducing imprecision in the conformance relation. By shifting the reference point of conformance comparison, our cleanness with synchronised retimings avoids this imprecision. What is more important, by performing the synchronisation independently of $\tau\subo$, we introduce the opportunity to constrain the set of conforming output traces to those traces that are as close as desired to the ideal expected output behaviour.

We proceed to formalise the enhanced notion of cleanness with synchronisation outlined above. We start with two auxiliary definitions.

Definition~\ref{def:ret-induced-conf} entails that whenever the conformance predicate \conf\  holds for certain pair of timed traces, it is witnessed by at least one relevant retiming $(r_1, r_2)$. The following operator ``extracts'' all such  witness retimings:
%To implement synchronisation into our theory we change the conformance definition \conf\ (Definition~\ref{def:ret-induced-conf}) by lifting the predicate checking for existence of a pair of retiming functions to the set of all  retiming functions witnessing the conformance:
\[
\begin{array}{ll}
%\conf\text{-Wit}^{\,\ret}_{\tau,\epsilon}(\mu_1,\mu_2) \,\eqdef\,
\wit \ \conf^{\,\ret}_{\tau,\epsilon}(\mu_1,\mu_2) \,\eqdef\,
& \{(r_1,r_2)\in \ret_{\tau}(\calt_1,\calt_2) \,\mid
\,\\
& \forall t \in \calt_1:\, \dy(\mu_1(t)-\mu_2 \circ r_1(t)) \leq \epsilon  \,\wedge\,\\
& \forall t \in \calt_2:\, \dy(\mu_1\circ r_2(t)-\mu_2(t)) \leq \epsilon \}.
\end{array}
\]

Similarly, we define the collection of all retimings that witness prefix conformance (\prefixconf\ predicate, Definition~\ref{def:prefixconf}):

%Similarly, we lift the predicate \prefixconf\ which checks conformance for sequence prefixes to the set of all  pairs of retiming functions proving conformance for these prefixes:
\[
\begin{array}{lll}
%\prefixconf\text{-Wit}^{\,\ret}_{\tau,\epsilon}(\mu_1,\mu_2,t) &\,\eqdef\,&
\prefixwit \ \conf(\mu_1,\mu_2,t) &\,\eqdef\,&
% \{(r_1,r_2)\in \conf\text{-Wit}^{\,\ret}_{\tau,\epsilon}(\mu_1[t_1^s\dots t_1^e], \mu_2[t_2^s\dots t_2^e])  \,\mid
 \{(r_1,r_2)\in \wit \ \conf(\mu_1[t_1^s\dots t_1^e], \mu_2[t_2^s\dots t_2^e])  \,\mid
\,\\
& & t_1^s \in [0,\tau] \cap \calt_1,
	 t_1^e \in
	[t-\tau,t+\tau] \cap \calt_1,\\
	& &  t_2^s  \in [0,\tau] \cap \calt_2, t_2^e \in
	[t-\tau,t+\tau] \cap \calt_2	\}.
\end{array}
\]

Note that the domains of the retimings in
$\prefixwit \ \conf(\mu_1,\mu_2,t)$
can be smaller than the domains of $\mu_1$ and $\mu_2$.

Synchronisation is realised through a function $\synch$ specifying all allowed pairs of output retimings for a given pair of input retimings. With this, we can extend the definition of $(\conf\subi,\conf\subo)$ cleanness as follows. Given two inputs that are conformant, i.e., $\prefixconf\subi(i_1,i_2)$, we may pick any pair $(r_1,r_2)$ from the set
$\prefixwit\ \conf\subi(i_1,i_2)$
of pairs of retiming functions for which the input conformance holds. This pair induces another pair $(r'_1,r'_2) \in \synch(r_1,r_2)$ of retiming functions for the output timeline. For those, the prefix conformance predicates $\conf_O(o_1 \circ r'_2,o_2)$ and  $\conf_O(o_1,o_2 \circ r'_1) $ must hold. This is formally expressed in the following definition.
\begin{defi}
	A deterministic system $\hybrid$ is \emph{$(\conf\subi,\conf\subo)$--clean with synchronised retiming through
	$\synch: \retime(\calt_1,\calt_2) \to \calp(\retime(\calt_1,\calt_2)) $} if the following holds:
	\[\begin{array}{ll}
	\forall i_1, i_2 \in \gtt(\caly):\,\,&
	\forall t\in\dom(i_1) \cup \dom(i_2) \\	& (\forall t' \leq t:\,
	\prefixconf\subi(i_1, i_2,t')\big)
	\,\implies\,\\&
	\exists (r_1,r_2) \in \prefixwit\ \conf\subi(i_1,i_2,t):\\
	&\exists (r'_1,r'_2) \in \synch(r_1,r_2):\\
	&
	\prefixconf\subo(\hybrid(i_1) \circ r'_2, \hybrid(i_2),t)
	\wedge\\&
	\prefixconf\subo(\hybrid(i_1), \hybrid(i_2)\circ r'_1,t)
	\end{array}
	\]

\end{defi}

Through the function $\synch$, which is a parameter to the above definition, we can specify the allowed retimings for the output, such as, for example a scaling of the input retiming when the timelines of the input and the output have different scales. It is the responsibility of the cleanness tester or verifier to accurately specify the expected behaviour, as an inappropriately chosen $\synch$ function can result in declaring doped systems to be clean.

One important aspect of this definition is that by selecting a suitable synchronisation function $\synch$ we can incorporate in the cleanness check any available knowledge regarding the expected output behaviours for conforming input trajectories. The following proposition states how the conformance-based notions of cleanness can be recovered by choosing appropriate retimings.

\begin{prop}
For a given class of retimings $\ret$, an arbitrary induced conformance relation on inputs  $\conf\subi=\conf^{\,\ret}_{\tau,\epsilon}$, and a conformance relation on outputs $\conf\subo$, there exists a  synchronised retiming
	$\synch$ such that \emph{$(\conf\subi,\conf\subo)$}-cleanness (as per Definition~\ref{def:confclean-det}) is a special instance of \emph{$(\conf\subi,\conf\subo)$}-cleanness through $\synch$.
\end{prop}

  By setting $\synch(r_1,r_2) = \{\idret\}$ we obtain the corresponding notion of $(\conf\subi,\conf\subo)$--cleanness, which, in particular, means that cleanness with synchronised retimings is also a conservative generalization of robust cleanness.

\begin{exa}\label{ex:synchronizedFormal}
Consider the behaviour introduced in Example~\ref{ex:synchronized}. As for the retiming witnessing conformance between $i_1$ and $i_2$, let us take the most obvious one i.e. $(r_1,r_2) = (t-\tau,t+\tau)$.
If $(r_1,r_2)$ covers the whole domain of the output trace, then we can use $\synch(r_1,r_2) = \{(r_1,r_2)\}$, according to which the output should be retimed in the same way as the input is.
By reusing the same retiming of input for output, $o_1(t - \tau - \delta)$ conforms to the retimed  output $o_1$ with respect to the margin $\delta$.
\end{exa}

We use this theory in our experimental setup in Section~\ref{sec:case_study} and show how it can lead to a more accurate analysis of emission data in practice.

\section{Expressing Cleanness in  Timed Hyper Logics}%
\label{sec:timed-hyper}
In this section we introduce a logic that is capable of characterizing the notions of  robust and hybrid cleanness.  Since robust cleanness can be characterized in the logic \hyperltl~\cite{ClarksonFKMRS14:post}, the logic we propose is a temporal logic for hyperproperties. Our semantic domain consists of generalized timed traces, and thus, our logic  extends Signal Temporal Logic (\stl)~\cite{DBLP:conf/formats/MalerN04} with quantifiers over traces.  In order to be able to express deviations in time, our logic uses freeze quantifiers as the mechanism for comparing values at different time points. More precisely,  the proposed logic is obtained by extending \stlstar~\cite{FreezeSTL} with trace quantifiers. In the remainder of the section we provide the formal definition of the logic and discuss its applicability in the context of specifying and monitoring cleanness of hybrid systems.

\subsection{Preliminaries}
\revised{For the presentation in this section it will be convenient to consider hybrid systems as sets of GTTs,  where  each GTT represents a pair of input and output GTTs.  The reason for this is that we will define a logic whose formulas refer to both the inputs and the outputs of a hybrid system over time,  and are therefore interpreted over sets of such combined traces that contain both the input to the system and the system's output.
}

\revised{
Formally,  we will represent a
 ${\caly'}$-valued hybrid system $\hybrid : \mathit{GTT}(\caly')  \to \calp(\mathit{GTT}(\caly'))$
as a subset of $\mathit{GTT}(\caly)$ where $\caly = \caly' \times \caly'$, defined as
$\{ \mu \in \mathit{GTT}(\caly' \times \caly') \mid
\exists \mu_I,  \mu_O \in \mathit{GTT}(\caly'),
\mu_O \in \hybrid(\mu_I),
\mu(t) = (\mu_I(t),\mu_O(t)) \text{ for all }t \in \dom(\mu_I)\}$.
The definition of the GTTs $\mu$ in this set is possible since according to Definition~\ref{def:hybridsystem} we have that  for all $\mu_I \in  \mathit{GTT}(\caly')$  and all $\mu_O \in \hybrid(\mu_I) $ it holds that
$\dom(\mu_I) =\dom( \mu_O)$.
}

\revised{
For the rest of this subsection, whenever we refer to a GTT $\mu \in \hybrid$ of a $\caly'$-valued hybrid system $\hybrid$,  we mean a function  $\mu: \calt \to \caly$ defined as above, with $\caly  = \caly'\times\caly'$.
Given $\mu \in \hybrid$ such that $\mu(t) = (\mu^I(t),\mu^O(t))$ for $t \in \dom(\mu)$, we denote with
$\mu^I : \calt \to \caly'$ the projection of $\mu$ on the input component and with $\mu^O : \calt \to \caly'$ its projection on the output component.
}

Let $\mu: \calt \to \caly$ be a GTT\@.
If $\calt$ is an interval of the form $[0,b] \subseteq \realsnn$ with $b \in \rats_{>0}$, or of the form $[0,\infty)$,  and
$\caly = \reals^n$ for some $n \in \nat$,
we say that $\mu$ is a \emph{real-valued signal},
and define $\mathit{length}(\mu) = b$,  respectively $\mathit{length}(\mu) = \infty$,
to be the \emph{time length} of $\mu$.
If $\calt$ is instead a strictly increasing sequence $t_0,t_1,t_2,\ldots$ of rational numbers such that $t_0=0$ we say that $\mu$ is a \emph{timed word} and similarly define its \emph{time length} as $\mathit{length}(\mu) = \max \calt$ if $\calt$ is finite and $\mathit{length}(\mu) = \infty$ otherwise.

Let $X$ be a finite set of \emph{real-valued variables}. We denote with $\reals^X$ the set of possible valuations of $X$.  In the rest of this section we assume that the range of all GTTs that we consider is $\reals^X$ for a given finite set of real-valued variables.  We will assume that the variables in $X$ are indexed, i.e.,  $X = \{x_1,x_2,\ldots x_n\}$ for some $n \in \nat$, and use $\reals^n$ instead of $\reals^X$ with the expected interpretation.
An \emph{atomic predicate} over $X$ is a function $\alpha : \reals^X \to \bools$.
\revised{Recall from Definition~\ref{def:gtt} that a GTT $\mu:\calt \to \caly$ is defined for a metric space $(\caly,d_\caly)$.   When $\caly  = \reals^X$ for some set of variables $X$,  we assume that the metric $d$ can be expressed as an arithmetic expression $d_\caly(X,X')$ over the variables $X \cup X'$, where $X' = \{x' \mid x \in X\}$.
More precisely,  we have that the expression $d_\caly(X,X')$ evaluates to $u \in \reals$ for valuations $v \in \reals^X$ and $v'\in \reals^{X'}$ of $X$ and $X'$,  respectively,  if and only if $d_\caly(v,v') = u$.
For a $\caly'$-valued hybrid system $H$ we denote with $d_I$ and $d_O$ the arithmetic expressions that define the metrics associated with the underlying metric spaces for the input and output values of $H$.}

\subsection{The logic \hyperstlstar}

We now define the logic \hyperstlstar, which extends the logic \stlstar~\cite{FreezeSTL}  with \emph{quantifiers over traces},  that are used to relate multiple GTTs in a hybrid system.  To this end,  let $V_{\mathit{trace}}$ be a countably infinite set of \emph{trace variables}.  For  a set $X$ of real-valued variables  and a given trace variable $\pi \in V_{\mathit{trace}}$,  let $X_\pi = \{x_\pi \mid x \in X\}$  be the set of variables indexed with $\pi$.

Let $\mathcal I = \{1,\ldots,m\}$ for some $m \in \nat$ be a finite index set. As in the logic \stlstar, the index set $\mathcal I$ consists of the indices of the positions in the \emph{frozen time vector}.  Intuitively,  at each position of the frozen time vector a time point can be stored.  For a trace variable $\pi \in V_{\mathit{trace}}$ and an index  $i \in \mathcal I$,  let $X_\pi^{*_i} = \{x_\pi^{*_i} \mid x \in X\}$  be the set of variables indexed with $\pi$ and $*_i$.

\subsubsection{Syntax}

Let $X$ be a finite set of real-valued variables, and \AP be a set of atomic predicates over the set of indexed variables $\bigcup_{\pi \in V_{\mathit{trace}}} X_{\pi} \cup \bigcup_{\pi \in V_{\mathit{trace}}, i \in \mathcal I} X_{\pi}^{*_i}$.

\hyperstlstar  formulas are defined according to the following grammar.
\[
\begin{array}{lll}
\Phi &::=& \exists \pi.\Phi \;|\; \forall \pi. \Phi \;|\;  \varphi,\\
\varphi &::= & \alpha \;|\; \top \;|\; \neg \varphi \;|\; \varphi \vee \varphi \;|\; \varphi \LTLuntil_{J}\varphi \;|\; \varphi \LTLsince_{J}\varphi \;|\; *_{i} \varphi,
\end{array}
\]
where
$\pi \in V_{\mathit{trace}}$ is a trace variable,
$\alpha$ is an atomic predicate from \AP,
$J \subseteq \realsnn$ is an interval with endpoints in $\ratsnn \cup \{\infty\}$, and
$i \in \mathcal I$ is an index.

The operators $\LTLuntil$ and $\LTLsince$ are the temporal operators \emph{Until} and \emph{Since}. The $*_{i}$ operator, for $i \in \mathcal I$ is the \emph{signal-value freeze} operator. Their semantics is formally defined below.
When the interval $J$ is of the form $[0,\infty)$ we often omit it for convenience.

The Boolean constant $\bot$ (false),  additional Boolean operators, as well as additional temporal operators are defined in the usual way. More concretely, we define
$\LTLfinally_{J} \varphi \defeq \top \LTLuntil_{J} \varphi$,
$\LTLglobally_{J} \varphi \defeq \neg \LTLfinally_{J} \neg \varphi$,
$\LTLpastfinally_{J} \varphi \defeq \top \LTLsince_{J} \varphi$,  and
$\LTLpastglobally_{J} \varphi \defeq \neg \LTLpastfinally_{J} \neg \varphi$.
%, and
%$\varphi\LTLweakuntil\psi \defeq \varphi\LTLuntil \psi \;\vee \LTLglobally\varphi$.

A \hyperstlstar formula is \emph{well-formed} if each occurrence of a trace quantifier introduces a unique variable name,  and it is \emph{closed} if every occurrence of a variable in $X_{\pi}$ is in the scope of a quantifier for $\pi$.  We will consider only well-formed \hyperstlstar formulas.

Note that in \hyperstlstar, unlike~\cite{FreezeSTL}, we also allow the past operator $\LTLsince$, as well as arbitrary intervals $J$ in the operators $\LTLuntil$ and $\LTLsince$.
We define $\hyperstlstar_{\mathit{fin}}$ to be the fragment of $\hyperstlstar$ such that every interval $J$ is of the form $[a,b]$, where $a, b \in \ratsnn$ and $a < b$.

\subsubsection{Semantics}

\hyperstlstar formulas are interpreted over trace assignments and register valuations. A \emph{trace assignment} is a partial function with finite domain from $V_{\mathit{trace}}$ to the set of GTTs in a given hybrid system.
\revised{Formally,  given a hybrid system $\hybrid$ represented as a set of input-output GTTs,  a trace assignment $\Pi$ is a partial function $\Pi :  V_{\mathit{trace}} \to \hybrid$.}
\emph{Register valuations}  are $|\mathcal I|$-dimensional vectors over $\realsnn$.

Let $\Pi$ be a trace assignment with domain $U_{\mathit{trace}} \subseteq V_{\mathit{trace}}$, and let $\alpha$ be an atomic predicate defined over the variables in
$\bigcup_{\pi \in U_{\mathit{trace}}} X_{\pi} \cup \bigcup_{\pi \in U_{\mathit{trace}}, i \in \mathcal I} X_{\pi}^{*_i}$.  Consider a time point $t \in \realsnn$ such that for each $\pi \in V_{\mathit{trace}}$  for which a variable from $X_\pi$ occurs in $\alpha$ it holds that $t \in \dom(\Pi(\pi))$,  and a register valuation $T$ such that for each pair $\pi \in V_{\mathit{trace}}$ and $i \in \mathcal I$ such that a variable from $X_\pi^{*_i}$ occurs  in  $\alpha$ it holds that $T(i) \in \dom(\Pi(\pi))$.  Then, the value of the atomic predicate $\alpha$ at the tuple $(\Pi,T,t)$ is defined as:
\[\alpha(\Pi,T,t) \defeq
\alpha((\Pi(\pi)(t))_{\pi \in U_{\mathit{trace}}},
(\Pi(\pi)(T(i)))_{\pi \in U_{\mathit{trace}},i\in \mathcal I}).\]
Intuitively,  the atomic predicate is evaluated using the signal values at time point $t$ and at the time points stored in the frozen time vector $T$.  If for some of the indexed variables that occur in $\alpha$ the corresponding time point is not in the time domain of the corresponding trace,  then the value of the atomic predicate is \emph{undefined}.

To define the semantics of \hyperstlstar, we define the function $\valuet$ that maps a formula $\Psi$,  a trace assignment $\Pi$,  a register assignment $T$ and a time point $t$ to a value in the set $\{\T,\F,  \U\}$, which indicates whether $\Psi$ is true ($\T$),  false ($\F$) or undefined ($\U$) at $(\Pi,T,t)$.
Formally,  for a hybrid system $H$, a trace assignment $\Pi$, a register valuation $T$, and $t \in \realsnn$, the value $\valuet(\Psi,H,\Pi,T,t)$ is defined by induction on the structure of \hyperstlstar formulas.

\begin{itemize}
\item If $\Psi = \alpha$,  then
\[\valuet(\Psi,H,\Pi,T,t) = \begin{cases}
\alpha(\Pi,T,t) & \text{if the value of } \alpha \text{ is defined at } (\Pi,T,t),\\
\U & \text{otherwise.}
\end{cases}\]
\item  If $\Psi = \top$,  then $\valuet(\Psi,H,\Pi,T,t)  = \T$.
\item  If $\Psi = \neg \varphi$, then
\[\valuet(\Psi,H,\Pi,T,t) = \begin{cases}
\T & \text{if } \valuet(\varphi,H,\Pi,T,t)  = \F,\\
\F & \text{if } \valuet(\varphi,H,\Pi,T,t)  = \T,\\
\U & \text{if }  \valuet(\varphi,H,\Pi,T,t)  = \U.
\end{cases}\]
\item If $\Psi = \varphi_1 \vee \varphi_2$,   then
\[\valuet(\Psi,H,\Pi,T,t) = \begin{cases}
\T & \text{if } \valuet(\varphi_1,H,\Pi,T,t)  = \T \text{ or } \valuet(\varphi_2,H,\Pi,T,t)  = \T,\\
\F & \text{if } \valuet(\varphi_1,H,\Pi,T,t)  = \F \text{ and } \valuet(\varphi_2,H,\Pi,T,t)  =\F,\\
\U & \text{otherwise}.
\end{cases}\]
\item If $\Psi  = \varphi_1 \LTLuntil_{J}\varphi_2$, then
\[\valuet(\Psi,H,\Pi,T,t) = \begin{cases}
\T & \text{if } \text{for some } t' \geq t \text{ such that } t'-t \in J: \
\valuet(\varphi_2,H,\Pi,T,t') = \T \text{ and }\\&
\phantom{if }\text{for all } t'' \in [t,t'):\ \valuet(\varphi_1,H,\Pi,T,t'') \in \{\T,\U\} \text{ or } \\
&\phantom{\text{ if for all } t'' \in [t,t'):\ } \revised{t''-t \in J \text{ and }}\valuet(\varphi_2,H,\Pi,T,t'') = \T ,\\
\F & \text{otherwise}.
\end{cases}\]
\item If $\Psi  = \varphi_1 \LTLsince_{J}\varphi_2$, then
\[\valuet(\Psi,H,\Pi,T,t) = \begin{cases}
\T & \text{if } \text{for some } t' \in [0,t] \text{ such that } t-t' \in J: \
\valuet(\varphi_2,H,\Pi,T,t') = \T \text{ and }\\&
\phantom{if }\text{for all } t'' \in (t',t]:\ \valuet(\varphi_1,H,\Pi,T,t'') \in \{\T,\U\} \text{ or } \\
&\phantom{\text{ if for all } t'' \in (t',t]:\ } \revised{t''-t \in J \text{ and }} \valuet(\varphi_2,H,\Pi,T,t'') = \T ,\\
\F & \text{otherwise}.
\end{cases}\]
\item If $\Psi  = *_i \varphi $, then
\[\valuet(\Psi,H,\Pi,T,t) =
\begin{cases}
\valuet(\varphi,H,\Pi,T[i\mapsto t],t) & \text{if } t \in \dom(\Pi(\pi)) \text{ for all } x_\pi^{*_i} \text{ that appear in } \varphi,\\
\U & \text{otherwise}.
\end{cases}
\]
where $T[i \mapsto t](j) \defeq t$ if $j=i$ and $T[i \mapsto t](j)  \defeq T(j)$ if $j \neq i$.
\item If $\Psi = \exists \pi.\Phi $, then
\[\valuet(\Psi,H,\Pi,T,t) = \begin{cases}
\T & \text{if } \text{for some } \mu \in H:\;
\valuet(\Phi,H,\Pi[\pi \mapsto \mu],T,t) = \T,\\
\F & \text{otherwise}.
\end{cases}\]
\item If $\Psi = \forall \pi.\Phi $, then
\[\valuet(\Psi,H,\Pi,T,t) = \begin{cases}
\T & \text{if for all } \mu \in H:\;
\valuet(\Phi,H,\Pi[\pi \mapsto \mu],T,t) \in \{\T,\U\},\\
\F & \text{otherwise}.
\end{cases}\]
\end{itemize}

\noindent
Note that if a formula is closed,  then its value is always either $\T$ or $\F$.

For a hybrid system $H$, a trace assignment $\Pi$, a register valuation $T$,  $t \in \realsnn$,  and a \hyperstlstar  formula $\Phi$ we can define the satisfaction relation $\models$ where
\[(H,\Pi,T,t)   \models \Phi \text{ if and only if } \valuet(\Phi,H,\Pi,T,t) = \T.\]

We say that a hybrid system \emph{$H$ satisfies a closed formula $\Phi$}, denoted $H \models \Phi$, if and only if it holds that $(H,\Pi_\emptyset,T_0,0) \models \Phi$, where $\Pi_\emptyset$ is the empty trace assignment and $T_0$ is the register valuation in which $0$ is stored at every index.

\begin{exa}
Let $H = \{\mu_1,\mu_2\}$  be a hybrid system that consists of two generalized timed traces,  $\mu_1$ with $\dom(\mu_1) = \{0,2,4\}$ and $\mu_2$ with $\dom(\mu_2) = \{0,1,3\}$, where  $\mu_1(t) = 0$ for all $t \in \{0,2,4\}$, and $\mu_2(0) = \mu_2(1) = 1$ and $\mu_2(3) = 0$.

Consider the \hyperstlstar formula $\Phi_1 = \forall \pi_1. \forall \pi_2. \;\LTLglobally_{[0,4]}(x_{\pi_1} = x_{\pi_2})$ that states that for every pair of timed traces and every time point in the interval $[0,4]$ the value of the two traces must be equal (i.e.,  they agree on the value of variable $x$).  We have that $H \not\models \Phi_1$ since the two traces differ at time point $t=0$. If, on the other hand we consider the formula $\Phi_2 = \forall \pi_1. \forall \pi_2. \;\LTLglobally_{[1,4]}(x_{\pi_1} = x_{\pi_2})$ obtained from $\Phi_1$ by replacing the interval $[0,4]$ by $[1,4]$, we have that $H \models\Phi_2$.  The justification behind this is that there is no time point in the interval $[1,4]$ where we witness a violation of the atomic predicate $x_{\pi_1} = x_{\pi_2}$. In particular, in the time interval $[1,4]$ there is no point at which both traces are defined.

Now, consider the \hyperstlstar formula $\Phi_3 = \forall \pi_1. \forall \pi_2. \;\LTLeventually_{[0,4]}(x_{\pi_1} = x_{\pi_2})$ that states that for every pair of traces, in the interval $[0,4]$ there \emph{exists} a time point where the values of the two traces are the same.  We have that $H \not\models \Phi_3$,  as expected, since there is no point in this interval where both traces are defined and have the same value.  Note that we also have $H \not \models \forall \pi_1. \forall \pi_2. \;\LTLeventually_{[1,4]}(x_{\pi_1} = x_{\pi_2})$ for the interval where the value  of $x_{\pi_1} = x_{\pi_2}$ is undefined.

Finally, let $\Phi_4 = \forall \pi_1. \forall \pi_2. \;\LTLeventually_{[0,3]}*_{1} \LTLeventually_{[0,1]} x_{\pi_1}^{*_1} = x_{\pi_2}$.  The formula $\Phi_4$ states that for every pair of traces there exists a time point $t_1$ in $[0,3]$ such that there is a time point $t_2$ at most $1$ time unit later, such that the value of the first trace at $t_1$ is equal to the value of the second trace at time $t_2$.  Here $t_1$ is the frozen time point per the semantics of the freeze operator $*$.  We have that $H \models \Phi_4$. To see this,  when $\pi_1$ is $\mu_1$ and $\pi_2$ is $\mu_2$ let $t_1= 2$ and $t_2 =3$, and when $\pi_1$ is $\mu_2$ and $\pi_2$ is $\mu_1$ let $t_1= 3$ and $t_2 =4$.
\end{exa}

\begin{rem}
Due to the generality of our semantic domain, which generalizes both continuous signals and timed words,  we have to address the issue of having to define the interpretation of \hyperstlstar formulas over all time points in $\realsnn$ while the considered traces might not be defined at all points.  Furthermore,  the semantics of the logic has to account for the fact that formulas,  even atomic propositions, refer to different traces which are possibly defined over different time domains.  To this end, we defined the function $\valuet$ that assigns values in the set $\{\T,\F,\U\}$.  For instance, if $\valuet(\alpha,H,\Pi,T,t) = \U$,  then $\valuet(\alpha \vee \neg\alpha,H,\Pi,T,t) = \U$.  For the temporal operators, our semantics is reminiscent of that in~\cite{Gazda2019},  in the sense that for evaluating $\varphi_1\LTLuntil\varphi_2$ the subformula $\varphi_1$ is evaluated only in time points where its value is defined,  and the time point where the obligation $\varphi_2$ must hold is one where its value is defined.  The treatment in $\LTLsince$ is analogous.
Our semantics interprets trace quantifiers over the traces for which the formula has a defined value.

Other temporal logics for timed hyperproperties face similar issues, which we discuss in Remark~\ref{rem:undefined}.
The logic \hyperstl~\cite{HyperSTL}, on the other hand is not affected by such difficulties, since its semantics is defined over continuous signals defined over a whole interval.
\end{rem}

\begin{rem}
In our definition of the semantics of $\varphi_1\LTLuntil_J \varphi_2$, similarly to~\cite{Deshmukh17} and~\cite{Gazda2019}, we account for the fact that in a dense time domain there might not exist a \emph{first} time point where $\varphi_2$ is satisfied.  Therefore we allow for $\varphi_1$ to be violated at intermediate time points as long as at those points the value of $\varphi_2$ is $\T$ \revised{and the constraint imposed by $J$ is satisfied.  More precisely, $\varphi_1\LTLuntil_J \varphi_2$ is $\T$ at time point $t$ if there exists a time point $t' \geq t$ such that $t' - t \in J$, $\varphi_2$ is $\T$ at $t'$,  and  for all intermediate points $t'' \in [t, t')$ it holds that if $\varphi_1$ is $\F$ at $t''$,  then $t''$ must be such that  $t'' - t \in J$ and $\varphi_2$ is $\T$ at $t''$.}
 The analogous holds for $\LTLsince$.
\end{rem}

\begin{rem}\label{rem:undefined}
Existing temporal logics for timed hyperproperties have also faced the challenge of dealing with timed traces that are defined over different sets of time points.

 In~\cite{HyperMTLother} this leads to the consideration of two different semantics of their logic \hypermtl: an asynchronous semantics that does not require the time stamps in two timed traces to match, and a synchronous semantics in which the range of quantifiers is restricted to the traces that synchronize with the current trace assignment.  The logic \hypermtl includes for each trace variable a Boolean constant $\top$ (true) indexed with that variable,  which allows for expressing syntactically in formulas the requirement that the current time point is in the domain of the corresponding trace.  In contrast,  in our logic  \hyperstlstar we account for undefined values on the semantic level in the definition of the value function,  and values at different points in time on different traces can be related via the freeze operator.

The authors of~\cite{HyperMTL} provide an alternative logic \hypermtl  by extending the logic \mtl with quantifiers over traces in the point-wise semantics.  The semantics of their logic has both a synchronous and an asynchronous layer.  At the synchronous layer, traces are compared at the same points in time,  and if a trace is undefined at a given point,  the value at the closest previous event is used.  At the asynchronous layer,  an asynchronous version of the $\LTLuntil$ operator allows for a bounded difference in the time points when the obligation of the \emph{Until} formula is fulfilled in  different traces.  Our logic, on the other hand,  allows for a general and flexible way of relating time points on different traces via the freeze operator.
\end{rem}

\subsection{Expressing robust and hybrid cleanness}\label{sec:cleanness-formulas}
Using the logic \hyperstlstar we can express trace and hybrid conformance and robust and hybrid cleanness.  We begin by first formalizing the conformance notions,  and then provide the characterization of cleanness.

Let $\pi_1$ and $\pi_2$ be two trace variables, and let $\tau$ and $\epsilon$ be non-negative rational constants.  We can express hybrid conformance with thresholds $\tau$ and $\epsilon$, i.e.,  $\hybc_{\te}$,  as follows:
\[
\begin{array}{lll}
\varphi^{\hybc}_{\tau,\epsilon} & = &
\Big(\LTLglobally *_1
\big(\LTLpastfinally_{[0,\tau]}
d_{\cal Y}(X_{\pi_2},X^{*_1}_{\pi_1}) \leq \epsilon\vee
\LTLeventually_{[0,\tau]}d_{\cal Y}(X_{\pi_2},X^{*_1}_{\pi_1}) \leq \epsilon\big)\Big) \wedge \\&&
\Big(\LTLglobally *_2
\big(\LTLpastfinally_{[0,\tau]}d_{\cal Y}(X_{\pi_1} ,X^{*_2}_{\pi_2}) \leq \epsilon \vee
\LTLeventually_{[0,\tau]}d_{\cal Y}(X_{\pi_1} ,X^{*_2}_{\pi_2}) \leq \epsilon \big)\Big)
\end{array}
\]
where $d_{\cal Y}(X,X')$ is an \revised{arithmetic expression characterizing the metric $d_{\cal Y}$}. Note that in the above formula the trace variables $\pi_1$ and $\pi_2$ are not quantified,  and hence it is not closed.

Intuitively, the formula states that for every time point on the trace described by $\pi_1$ it holds that within $\tau$ time units in the past or in the future, there exists a point on the trace described by $\pi_2$ where the value is $\epsilon$-close to the value of $\pi_1$ at the current time point, and symmetrically for the other direction with traces $\pi_1$ and $\pi_2$ swapped.

\begin{prop}\label{prop:correctness-hconf}
Let $H$  be a deterministic hybrid system
defined over a set of real-valued variables $X$
such that $0 \in \dom(\mu)$ for each $\mu \in H$.
Let $\tau,\epsilon \geq 0$ be rational constants,
and $\mu_1,\mu_2 \in  H$.
Let $\pi_1$ and $\pi_2$ be trace variables and
$\Pi=\{\pi_1 \mapsto \mu_1,\pi_2 \mapsto \mu_2 \}$.
Then
\begin{center}
$\hybc_{\te}(\mu_1,\mu_2)$ if and only if
$ (H,\Pi,T_0,0) \models \varphi^{\hybc}_{\tau,\epsilon}$.
\end{center}
\end{prop}
\begin{proof}
($\Longrightarrow$) First, suppose that $\hybc_{\te}(\mu_1,\mu_2)$.

By Definition~\ref{def:confrel}, we have that for all $t_1 \in \dom(\mu_1)$ there exists $t_2 \in \dom(\mu_2)$ such that
$|t_1-t_2|\leq \tau$ and $\dy(\mu_2(t_2),\mu_1(t_1))\leq\epsilon$.
Hence, when $t_1 \in \dom(\mu_1)$, we have that
$\valuet\big(
*_{1}\big(\LTLpastfinally_{[0,\tau]}
d_{\cal Y}(X_{\pi_2},X^{*_1}_{\pi_1}) \leq \epsilon\vee
\LTLeventually_{[0,\tau]}d_{\cal Y}(X_{\pi_2},X^{*_1}_{\pi_1}) \leq \epsilon\big),
H,\Pi,T_0,t_1
\big) = \T$.

If $t_1 \not \in \dom(\mu_1)$,
then, from the definition of the semantics of the operator $*_{1}$ we have that
$\valuet\big(
*_{1}\big(\LTLpastfinally_{[0,\tau]}
d_{\cal Y}(X_{\pi_2},X^{*_1}_{\pi_1}) \leq \epsilon\vee
\LTLeventually_{[0,\tau]}d_{\cal Y}(X_{\pi_2},X^{*_1}_{\pi_1}) \leq \epsilon\big),
H,\Pi,T_0,t_1
\big) = \U$.

Thus,
$\valuet\big(
\LTLglobally*_{1}\big(\LTLpastfinally_{[0,\tau]}
d_{\cal Y}(X_{\pi_2},X^{*_1}_{\pi_1}) \leq \epsilon\vee
\LTLeventually_{[0,\tau]}d_{\cal Y}(X_{\pi_2},X^{*_1}_{\pi_1}) \leq \epsilon\big),
H,\Pi,T_0,0
\big) = \T$.

$\valuet\big(
\LTLglobally *_2
\big(\LTLpastfinally_{[0,\tau]}d_{\cal Y}(X_{\pi_1} ,X^{*_2}_{\pi_2}) \leq \epsilon \vee
\LTLeventually_{[0,\tau]}d_{\cal Y}(X_{\pi_1} ,X^{*_2}_{\pi_2}) \leq \epsilon \big),
H,\Pi,T_0,0
\big) = \T$
can be shown by applying the same reasoning as above, this time for $\mu_2$.

From the two facts we showed, we can conclude that $ (H,\Pi,T_0,0) \models \varphi^{\hybc}_{\tau,\epsilon}$.

($\Longleftarrow$) For the other direction, assume that $ (H,\Pi,T_0,0) \models \varphi^{\hybc}_{\tau,\epsilon}$.

Let $t_1 \in \dom(\mu_1)$.
Since $ (H,\Pi,T_0,0) \models \varphi^{\hybc}_{\tau,\epsilon}$,
by the semantics of $\LTLglobally$ we have that
$\valuet\big(
*_{1}\big(\LTLpastfinally_{[0,\tau]}
d_{\cal Y}(X_{\pi_2},X^{*_1}_{\pi_1}) \leq \epsilon\vee
\LTLeventually_{[0,\tau]}d_{\cal Y}(X_{\pi_2},X^{*_1}_{\pi_1}) \leq \epsilon\big),
H,\Pi,T_0,t_1
\big) \in \{\T,\U\}$.
Taking into account that we are considering the case when $t_1 \in \dom(\mu_1)$,
and that by definition
\begin{itemize}
\item $\valuet\big(
\LTLpastfinally_{[0,\tau]}
d_{\cal Y}(X_{\pi_2},X^{*_1}_{\pi_1}) \leq \epsilon,
H,\Pi,T_0,t_1
\big) \neq U$
and
\item $\valuet\big(
\LTLeventually_{[0,\tau]}d_{\cal Y}(X_{\pi_2},X^{*_1}_{\pi_1}) \leq \epsilon,
H,\Pi,T_0,t_1
\big) \neq U$,
\end{itemize}
we conclude that
$\valuet\big(
*_{1}\big(\LTLpastfinally_{[0,\tau]}
d_{\cal Y}(X_{\pi_2},X^{*_1}_{\pi_1}) \leq \epsilon\vee
\LTLeventually_{[0,\tau]}d_{\cal Y}(X_{\pi_2},X^{*_1}_{\pi_1}) \leq \epsilon\big),
H,\Pi,T_0,t_1
\big) = \T$.
Therefore, there exists $t_2 \in \dom(\mu_2)$ such that
$|t_1-t_2|\leq \tau$ and $\dy(\mu_2(t_2),\mu_1(t_1))\leq\epsilon$.

The reasoning for the other conjunct when $t_2 \in \dom(\mu_2)$
is symmetric. We hence obtain $\hybc_{\te}(\mu_1,\mu_2)$, which concludes the proof.
\end{proof}

As a special case, we can express $\tracec_{\epsilon}$ as
\[
\begin{array}{lll}
\varphi^{\tracec}_{\epsilon} & = &
\Big(\LTLglobally *_1
\big(
\LTLeventually_{[0,0]}d_{\cal Y}(X_{\pi_2},X^{*_1}_{\pi_1}) \leq \epsilon\big)\Big) \wedge
\Big(\LTLglobally *_2
\big(
\LTLeventually_{[0,0]}d_{\cal Y}(X_{\pi_1} ,X^{*_2}_{\pi_2}) \leq \epsilon \big)\Big).
\end{array}
\]

Note that the formula
$
\varphi_{\epsilon}  =
\LTLglobally
d_{\cal Y}(X_{\pi_1},X_{\pi_2}) \leq \epsilon
$
\emph{does not}  characterize trace conformance as it does not assert the requirement that the time domains of the two traces must be the same.  The formula $\varphi^{\tracec}_{\epsilon}$, on the other hand,  requires that each time point where one of the traces is defined, must be matched by a value of the other trace at the same time point.

\begin{prop}\label{prop:correctness-tconf}
Let $H$  be a deterministic hybrid system defined over a set of real-valued variables $X$ such that $0 \in \dom(\mu)$ for each $\mu \in H$.
Let $\epsilon \geq 0$ be a rational constant, and $\mu_1,\mu_2 \in H$.
Let $\pi_1$ and $\pi_2$ be trace variables and
$\Pi=\{\pi_1 \mapsto \mu_1,\pi_2 \mapsto \mu_2 \}$. Then
\begin{center}
$\tracec_{\epsilon}(\mu_1,\mu_2)$ if and only if
$ (H,\Pi,T_0,0) \models \varphi^{\tracec}_{\epsilon}$.
\end{center}
\end{prop}
\begin{proof}
($\Longrightarrow$) First, suppose that $\tracec_{\epsilon}(\mu_1,\mu_2)$.

By Definition~\ref{def:confrel}, we have that $\dom(\mu_1) = \dom(\mu_2)$ and for all $t_1 \in \dom(\mu_1)$ it holds that $\dy(\mu_2(t_1),\mu_1(t_1))\leq\epsilon$.
This is equivalent to the conjunction of the following statements:
\begin{itemize}
\item for all $t_1 \in \dom(\mu_1)$, it holds that $t_1 \in \dom(\mu_2)$ and $\dy(\mu_2(t_1),\mu_1(t_1))\leq\epsilon$, and
\item for all $t_2 \in \dom(\mu_2)$, it holds that $t_2 \in \dom(\mu_1)$ and $\dy(\mu_1(t_2),\mu_2(t_2))\leq\epsilon$.
\end{itemize}

\noindent
Therefore, if $t \in \dom(\mu_1)=\dom(\mu_2)$ it holds that
\[
\begin{array}{l}
\valuet\big(
*_1
\big(
\LTLeventually_{[0,0]}d_{\cal Y}(X_{\pi_2},X^{*_1}_{\pi_1}) \leq \epsilon\big),
H,\Pi,T_0,t
\big) = \T  \text{ and}\\
\valuet\big(
*_2
\big(
\LTLeventually_{[0,0]}d_{\cal Y}(X_{\pi_1} ,X^{*_2}_{\pi_2}) \leq \epsilon \big),
H,\Pi,T_0,t
\big) = \T.
\end{array}
\]

When $t \not\in \dom(\mu_1)=\dom(\mu_2)$ we have that that
\[
\begin{array}{l}
\valuet\big(
*_1
\big(
\LTLeventually_{[0,0]}d_{\cal Y}(X_{\pi_2},X^{*_1}_{\pi_1}) \leq \epsilon\big),
H,\Pi,T_0,t
\big) = \U  \text{ and}\\
\valuet\big(
*_2
\big(
\LTLeventually_{[0,0]}d_{\cal Y}(X_{\pi_1} ,X^{*_2}_{\pi_2}) \leq \epsilon \big),
H,\Pi,T_0,t
\big) = \U.
\end{array}
\]

Thus,
$
\valuet\big(
\varphi^{\tracec}_{\epsilon},
H,\Pi,T_0,0
\big) = \T
$,
which is what we had to prove.

($\Longleftarrow$) For the other direction, assume that $ (H,\Pi,T_0,0) \models \varphi^{\tracec}_{\epsilon}$.

Let $t_1 \in \dom(\mu_1)$. We have that
$\valuet\big(
*_1
\big(
\LTLeventually_{[0,0]}d_{\cal Y}(X_{\pi_2},X^{*_1}_{\pi_1}) \leq \epsilon\big),
H,\Pi,T_0,t_1
\big) = \T
$.
This implies that $t_1 \in \dom(\mu_2)$ and
$\valuet\big(
d_{\cal Y}(X_{\pi_2},X^{*_1}_{\pi_1}) \leq \epsilon,
H,\Pi,\{1 \mapsto t_1\},t_1
\big) = \T.
$

Thus, we can conclude that for all $t_1 \in \dom(\mu_1)$, it holds that $t_1 \in \dom(\mu_2)$ and $\dy(\mu_2(t_1),\mu_1(t_1))\leq\epsilon$. Analogously, we can show that for all $t_2 \in \dom(\mu_2)$, it holds that $t_2 \in \dom(\mu_1)$ and $\dy(\mu_1(t_2),\mu_2(t_2))\leq\epsilon$.
Hence, by Definition~\ref{def:confrel}, $\tracec_{\epsilon}(\mu_1,\mu_2)$.
\end{proof}

We use the idea of the above encoding to define a closed \hyperstlstar formula that characterizes hybrid cleanness,  $\hybclean(\tau_I,\epsilon_I,\tau_O,\epsilon_O)$, for deterministic hybrid systems.

For the rest of the section we consider hybrid systems $H$ such that for every $\mu \in H$ it holds that $0 \in \dom(\mu)$,  that is, we assume that all traces are defined at time point $0$.

Furthermore,  we assume that the set of variables $X$ defining the states of the hybrid system $H$ contains an \emph{explicit clock variable $c$} representing the current time, that is never reset. That is,  $c$ simply captures the time-stamps of the values of the GTTs in $H$.  \revised{Formally,  for every GTT $\mu \in H$, and every $t \in \realsnn$, it holds that $\mu(t)(c) = t$.} With that,  the atomic propositions in \AP can refer to the current time point, and the freeze operator captures the current time-stamp together with the current values of the other variables.
\revised{Let $\alpha \in \AP$ be an atomic proposition and $r, r_i \in \realsnn$ for $i \in \mathcal{I}$ be non-negative real constants.  We denote by $\alpha[r,r_1,\ldots,r_{|\mathcal I|}]$ the atomic predicate obtained from $\alpha$ by replacing each variable $c_\pi$ by $r$, and each variable $c_\pi^{*_i}$ by $r_i$, for all $\pi \in V_{\mathit{trace}}$ and $i \in \mathcal I$.
By the definition of the clock variable $c$,  for every trace assignment $\Pi$, register valuation $T$, and $t \in \realsnn$ we have that
\[
\valuet(\alpha,H,\Pi,T, t) =
\valuet(\alpha[t,T(1),\ldots,T(|\mathcal I|)],H,\pi,T,t),
\]
when for every $\pi \in V_{\mathit{trace}}$ for which $c_\pi$ occurs in $\alpha$ it holds that $t \in \dom(\Pi(\pi))$ and for every $\pi \in V_{\mathit{trace}}$ and $i \in \mathcal I$ for which $c_\pi^{*_i}$ occurs in $\alpha$ it holds that $T(i) \in \dom(\Pi(\pi))$.
}

Let $\tau_I$ and $\epsilon_I$ be non-negative rational constants defining the threshold values for the input conformance relation,  and $\tau_O$ and $\epsilon_O$ be the ones for the output conformance.

Let $\pi$ and $\pi'$ be trace variables and $i,s,e \in \mathcal{I}$.  First, we define the formulas
\[
\begin{array}{lll}
\varphi^{\textsf{matchI}}_{\tau_I,\epsilon_I}(\pi,i,\pi',s,e)  & = &
*_i \big(
\LTLpastfinally_{[0,\tau_I]}
(d_{I}(X_{\pi'},X^{*_i}_{\pi}) \leq \epsilon_I
\wedge c_{\pi'} \geq c_{\pi'}^{*_s} \wedge c_{\pi'} \leq c_{\pi'}^{*_e}) \vee\\&&
\phantom{
*_i \big(
}
\LTLeventually_{[0,\tau_I]}
(d_{I}(X_{\pi'},X^{*_i}_{\pi}) \leq \epsilon_I
\wedge  c_{\pi'} \geq c_{\pi'}^{*_s} \wedge c_{\pi'} \leq c_{\pi'}^{*_e} )
\big),\\
\varphi^{\textsf{matchO}}_{\tau_O,\epsilon_O}(\pi,i,\pi',s,e) & =&
*_i \big(
\LTLpastfinally_{[0,\tau_O]}
(d_{O}(X_{\pi'},X^{*_i}_{\pi}) \leq \epsilon_O
\wedge c_{\pi'} \geq c_{\pi'}^{*_s} \wedge c_{\pi'} \leq c_{\pi'}^{*_e}) \vee\\&&
\phantom{
*_i \big(
}
\LTLeventually_{[0,\tau_O]}
(d_{O}(X_{\pi'},X^{*_i}_{\pi}) \leq \epsilon_O
\wedge  c_{\pi'} \geq c_{\pi'}^{*_s} \wedge c_{\pi'} \leq c_{\pi'}^{*_e} )
\big),
\end{array}
\]
where  $d_{I}(X,X')$ and $d_{O}(X,X')$ are \revised{the arithmetic expressions characterizing the metrics on the sets of input and output values of the considered hybrid system.}

Intuitively,  the formula $\varphi^{\textsf{matchI}}_{\tau_I,\epsilon_I}(\pi,i,\pi',s,e)$ evaluated at time point $t$ and register valuation $T$ is true if and only if there exists a time point $t' \in [t-\tau_I,t+\tau_I] \cap [T(s),T(e)]$ such that the input value at time $t$ on the trace represented by $\pi$ and the input value at time $t'$ on the trace represented by $\pi'$ are $\epsilon_I$-close.  The formula $\varphi^{\textsf{matchO}}_{\tau_O,\epsilon_O}(\pi,i,\pi',s,e)$ states the same  for the output values.  The need to constrain the time $t'$ where the match of the values at $t$ must be found comes from the fact that in the definition of cleanness in Section~\ref{sec:definition}, prefixes are compared using the predicate $\prefixconf$. Recall that $\prefixconf$ compares two prefixes $\mu_1$ and $\mu_2$ by requiring that there exist segments  $\mu_1[t_1^s\dots t_1^e]$ and $ \mu_2[t_2^s\dots t_2^e]$  obtained from them, that are conformant.  In the above formulas, the frozen values of the clock variable $c$  represent the end points of the interval for the trace assigned to $\pi'$.

Using the formula
$\varphi^{\textsf{matchI}}_{\tau_I,\epsilon_I}(\pi,i,\pi',s,e)$ we define the formula
$\varphi^{\textsf{PrefConfI}}_{\tau_I,\epsilon_I}(\pi_1, \pi_2)$
which is true if and only if the current time point defines a pair of prefixes of the traces represented by $\pi_1$ and $\pi_2$ for which there exist hybrid conforming segments obtained  from the prefixes by possibly removing a prefix/suffix of length at most $\tau_I$ or adding  a suffix of length at most $\tau_I$.

The formula is defined as a disjunction over the possible ways in which the end-points of the two segments are ordered on the time line.
Let $P$ be a set of $4$-tuples of the indices $\{3,4,5,6\} \subseteq \cali$ such that $(s_1,e_1,s_2,e_2) \in P$ if and only if $s_1,e_1,s_2,e_2$ are pairwise different, and
$s_1,s_2 \in\{3,4\}$ and $e_1,e_2 \in\{5,6\}$.  For $\overline p = (s_1,e_1,s_2,e_2) \in P$,  abusing notation we define $\overline p(s_1) = \overline p(e_1) =\pi_1$ and  $\overline p(s_2) = \overline p(e_2) =\pi_2$.  That is, the function $\overline p$ maps each endpoint index to the trace variable with which it is associated.  We now define
\[
\begin{array}{rll}
\varphi^{\mathsf{PrefConfI}}_{\tau_I,\epsilon_I}(\pi_1,\pi_2) & = &
\bigvee\limits_{\overline p = (s_1,e_1,s_2,e_2) \in P}
\big(*_{7} \varphi_{\pi_1} \wedge  *_{7}  \varphi_{\pi_2}\big)\\
\varphi_{\pi_1} &=&
\LTLpastfinally
\big(c_{\overline p(3)} \leq \tau_I \wedge *_{3}
\LTLfinally
\big(c_{\overline p(4)} \leq \tau_I \wedge \\&&
\phantom{
\LTLpastfinally
\big(
}
*_{4}
\LTLfinally
\big(c_{\overline p(5)} \geq (c_{\pi_1}^{*_7} - \tau_I) \wedge
c_{\overline p(5)} \leq (c_{\pi_1}^{*_7} + \tau_I) \wedge\\&&
\phantom{
\LTLpastfinally
\big(
}
*_{5}
\LTLfinally
\big(c_{\overline p(6)} \geq (c_{\pi_1}^{*_7} - \tau_I) \wedge
c_{\overline p(6)} \leq (c_{\pi_1}^{*_7} + \tau_I) \wedge
*_6\; \varphi^{\mathsf{ConfI}}
\big)
\big)\big)\big)\\
\varphi_{\pi_2} &=&
\LTLpastfinally
\big(c_{\overline p(3)} \leq \tau_I \wedge *_{3}
\LTLfinally
\big(c_{\overline p(4)} \leq \tau_I \wedge \\&&
\phantom{
\LTLpastfinally
\big(
}
*_{4}
\LTLfinally
\big(c_{\overline p(5)} \geq (c_{\pi_2}^{*_7} - \tau_I) \wedge
c_{\overline p(5)} \leq (c_{\pi_2}^{*_7} + \tau_I) \wedge\\&&
\phantom{
\LTLpastfinally
\big(
}
*_{5}
\LTLfinally
\big(c_{\overline p(6)} \geq (c_{\pi_2}^{*_7} - \tau_I) \wedge
c_{\overline p(6)} \leq (c_{\pi_2}^{*_7} + \tau_I) \wedge
*_6 \;  \varphi^{\mathsf{ConfI}}
\big)
\big)\big)\big)\\
\varphi^{\mathsf{ConfI}}& =&
\LTLpastglobally \Big(
\big(c_{\pi_1} \geq c_{\pi_1}^{*_{s_1}} \wedge c_{\pi_1} \leq c_{\pi_1}^{*_{e_1}})
\rightarrow  \varphi^{\textsf{matchI}}_{\tau_I,\epsilon_I}(\pi_1,1,\pi_2,s_2,e_2)
\Big)\wedge\\&&
\LTLpastglobally \Big(
\big(c_{\pi_2} \geq c_{\pi_2}^{*_{s_2}} \wedge c_{\pi_2} \leq c_{\pi_2}^{*_{e_2}})
\rightarrow   \varphi^{\textsf{matchI}}_{\tau_I,\epsilon_I}(\pi_2,2,\pi_1,s_1,e_1)
\Big).
\end{array}
\]
The conjunct $*_{7} \varphi_{\pi_1}$ handles the case when the current time point $t$ (i.e., the last time point for the considered prefixes) is in the domain of $\pi_1$. The second conjunct handles the case when $t$ is in the domain of $\pi_2$.  If neither is the case,  then the value of the conjunction is $\U$.  Since the formulas $*_{7} \varphi_{\pi_1}$ and $*_{7} \varphi_{\pi_2}$ differ only in the trace on which the time-point $t$ is frozen, if both of them have a defined value, than these values are necessarily the same.

Each of $ \varphi_{\pi_1}$ and $ \varphi_{\pi_2}$ asserts the existence of a sequence of time points defining the compared segments of the two traces  and their input conformance (formula $\varphi^{\mathsf{ConfI}}$).  \revised{The formula $\varphi^{\mathsf{ConfI}}$ captures the requirement that the two input prefixes ending at the current time point have segments (defined by the pairs of time points $c_{\pi_1}^{*_{s_1}}$ and $c_{\pi_1}^{*_{e_1}}$,  and $c_{\pi_2}^{*_{s_2}}$ and $c_{\pi_2}^{*_{e_2}}$, respectively) that are hybrid conformant, as in Definition~\ref{def:prefixconf}.}

The formula $\varphi^{\textsf{PrefConfO}}_{\tau_O,\epsilon_O}(\pi_1, \pi_2)$ is defined analogously using $\varphi^{\textsf{matchO}}_{\tau_O,\epsilon_O}(\pi,i,\pi',s,e)$.

\begin{prop}\label{prop:correctness-prefcon}
Let $H$  be a deterministic hybrid system defined over a set of real-valued variables $X$ that includes an explicit clock variable $c$,  and such that $0 \in \dom(\mu)$ for each $\mu \in H$.
Let $\tau_I,\tau_O,\epsilon_I,\epsilon_O \geq 0$ be rational constants,  \revised{and the predicates $\prefixconf_I$ and $\prefixconf_O$ be instantiated using
${\hybc}_{\tau_I,\epsilon_I}$ and ${\hybc}_{\tau_O,\epsilon_O}$ respectively.
That is,  let $\conf_I = {\hybc}_{\tau_I,\epsilon_I}$ and $\conf_O = {\hybc}_{\tau_O,\epsilon_O}$.}
Let $\mu_1,\mu_2 \in H$,
let $\pi_1$ and $\pi_2$ be trace variables and
$\Pi=\{\pi_1 \mapsto \mu_1,\pi_2 \mapsto \mu_2 \}$ a trace assignment.
\begin{enumerate}
\item If $t \in \dom(\mu_1) \cup \dom(\mu_2)$, then
\begin{center}
$\prefixconf_I(\mu_1\revised{^I},\mu_2\revised{^I},t)$ is true if and only if
$\valuet\big(\varphi^{\mathsf{PrefConfI}}_{\tau_I,\epsilon_I}(\pi_1,\pi_2),H,\Pi,T_0,t\big) \in \{\T,\U\}$.
\end{center}
\item If $t \in \dom(\mu_1) \cup \dom(\mu_2)$, then\\
%\begin{center}
$\prefixconf_O(\mu_1\revised{^O},\mu_2\revised{^O},t)$ is true if and only if
$\valuet\big(\varphi^{\mathsf{PrefConfO}}_{\tau_O,\epsilon_O}(\pi_1,\pi_2),H,\Pi,T_0,t\big) \in\{\T,\U\}$.
%\end{center}
\item If $t \not\in \dom(\mu_1) \cup \dom(\mu_2)$ then
$\valuet\big(\varphi^{\mathsf{PrefConfI}}_{\tau_I,\epsilon_I}(\pi_1,\pi_2),H,\Pi,T_0,t\big) = \U$.
\item If $t \not\in \dom(\mu_1) \cup \dom(\mu_2)$ then
$\valuet\big(\varphi^{\mathsf{PrefConfO}}_{\tau_O,\epsilon_O}(\pi_1,\pi_2),H,\Pi,T_0,t\big) = \U$.
\end{enumerate}
\end{prop}
\begin{proof}
We show (1), the proof for (2) is analogous.

($\Longrightarrow$)
Suppose that $\prefixconf_I(\mu_1\revised{^I},\mu_2,\revised{^I}t)$ is true.
By Definition~\ref{def:prefixconf}, there exist
$t_1^s, t_2^s \in [0,\tau_I]$ and
$t_1^e, t_2^e \in [t-\tau_I,t+\tau_I]$ such that
$\conf_I(\mu_1\revised{^I}[t_1^s\dots t_1^e], \mu_2\revised{^I}[t_2^s\dots t_2^e])$ is true.

Let $t_3,t_4,t_5,t_6$ be a permutation of
$t_1^s, t_2^s, t_1^e, t_2^e$ such that
$t_3 \leq t_4 \leq t_5 \leq t_6$.
Then, $t_3,t_4 \in [0,\tau_I]$ and $t_5,t_6 \in [t-\tau_I,t+\tau_I]$.
Let $\overline p = (s_1,e_1,s_2,e_2) \in P$ be defined according to the permutation $t_3,t_4,t_5,t_6$, that is,
$s_1 = i$ where $t_1^s$ is $t_i$, and so on. Let $t_7 =t$.
We define the register valuation
$T = \{i\mapsto t_i \mid i \in \{3,4,5,6,7\}\}$.

Since
$\conf_I(\mu_1\revised{^I}[t_1^s\dots t_1^e], \mu_2\revised{^I}[t_2^s\dots t_2^e])$,
we have that
$\valuet\big(\varphi^{\mathsf{ConfI}},H,\Pi,T,t_6\big) = \T$.

Since $t \in \dom(\mu_1) \cup \dom(\mu_2)$,
assume, without loss of generality that $t \in \dom(\mu_1)$.
By the choice of $t_3,t_4,t_5,t_6$ we have that
$t_3 \leq t_4 \leq t_5 \leq t_6$, and
$t_3,t_4 \in [0,\tau]$ and
$t_5,t_6 \in [t-\tau_I,t+\tau_I]$ which,
together with the definition of $\varphi_{\pi_1}$ implies that
$\valuet\big(\varphi_{\pi_1},H,\Pi,\{7 \mapsto t_7\},t_7\big)=\T$.

If $t \in \dom(\mu_2)$ we can show as above that the value of $*_{7}\varphi_{\pi_2}$ is $\T$.  Otherwise, the value of $*_{7}\varphi_{\pi_2}$ is $\U$.
Hence, we conclude that
$\valuet\big(\varphi^{\mathsf{PrefConfI}}_{\tau_I,\epsilon_I}(\pi_1,\pi_2),H,\Pi,T_0,t\big)\in \{\T,\U\}$.

($\Longleftarrow$)
Now, suppose that
$\valuet\big(\varphi^{\mathsf{PrefConfI}}_{\tau_I,\epsilon_I}(\pi_1,\pi_2),H,\Pi,T_0,t\big) \in \{\T,\U\}$.
Hence, there exists
$\overline p = (s_1,e_1,s_2,e_2) \in P$ for which we have
$\valuet\big(\big(*_{7} \varphi_{\pi_1} \wedge  *_{7} \varphi_{\pi_2}\big),
H,\Pi,T_0,t) \in \{\T,\U\}$.  Since $t \in \dom(\mu_1) \cup \dom(\mu_2)$,  assume, without loss of generality, that $t  \in \dom(\mu_1)$.  Then we must have that $\valuet\big(*_{7}\, \varphi_{\pi_1},H,\Pi,T_0,t\big) =\T $. Let $t_7 = t$. Then, $\valuet\big(\varphi_{\pi_1},H,\Pi,\{7 \mapsto t_7\},t\big) = \T$.
Therefore, according to the definition of $\varphi_{\pi_1}$ we have that there exist
$t_3 \leq t_4 \leq t_5 \leq t_6$ such that
$t_3,t_4 \in [0,\tau]$ and
$t_5,t_6 \in [t-\tau_I,t+\tau_I]$, and, furthermore,
$\valuet\big(\varphi^{\mathsf{ConfI}},H,\Pi,T,t_6\big) = \T$ for the valuation $T = \{i\mapsto t_i \mid i \in \{3,4,5,6,7\}\}$.

We define $t_1^s, t_2^s, t_1^e, t_2^e$ to be the permutation of $t_3,t_4,t_5,t_6$ determined by $\overline p$,
that is $t_1^s = t_{s_1}$, and so on.
Then, since
$\valuet\big(\varphi_{\pi_1},H,\Pi,\{7 \mapsto t_7\},t\big)=\T$, we have that $t_1^s \leq t_1^e$, $t_2^s \leq t_2^e$, and
$t_1^s, t_2^s \in [0,\tau_I]$ and
$t_1^e, t_2^e \in [t-\tau_I,t+\tau_I]$. Hence,
$\valuet\big(\varphi^{\mathsf{ConfI}},H,\Pi,T,t_6\big)=\T$ implies that
$\conf_I(\mu_1\revised{^I}[t_1^s\dots t_1^e], \mu_2\revised{^I}[t_2^s\dots t_2^e])$, which allows us to conclude that $\prefixconf_I(\mu_1\revised{^I},\mu_2\revised{^I},t)$.

\bigskip

(3) and (4) follow directly from the semantics of the operators $\wedge$ and $*$.
\end{proof}

We now define the formula characterizing hybrid cleanness as
\[\Phi^{\hybclean}_{\tau_I,\epsilon_I,\tau_O,\epsilon_O} =
\forall \pi_1. \forall\pi_2. \;
\LTLglobally \Big(
\big(\LTLpastglobally
\varphi^{\textsf{PrefConfI}}_{\tau_I,\epsilon_I}(\pi_1, \pi_2)
\big)
\rightarrow
\varphi^{\textsf{PrefConfO}}_{\tau_O,\epsilon_O}(\pi_1, \pi_2)
\Big)
.\]

\begin{prop}
Let $H$  be a deterministic hybrid system defined over a set of real-valued variables $X$ that includes an explicit clock variable $c$,  and such that $0 \in \dom(\mu)$ for each $\mu \in H$.  Let  $\tau_I,\tau_O,\epsilon_I,\epsilon_O \geq 0$ be rational constants.  It holds that
\begin{center}
$H$ is $\hybclean(\tau_I,\epsilon_I,\tau_O,\epsilon_O)$ if and only if
$ H \models \Phi^{\hybclean}_{\tau_I,\epsilon_I,\tau_O,\epsilon_O}$.
\end{center}
\end{prop}
\begin{proof}
($\Longrightarrow$)
First, suppose that
$H$ is $\hybclean(\tau_I,\epsilon_I,\tau_O,\epsilon_O)$.
Let $\mu_1,\mu_2 \in H$ be two arbitrarily chosen traces.
Let $\pi_1$ and $\pi_2$ be trace variables, and let
$\Pi=\{\pi_1 \mapsto \mu_1,\pi_2 \mapsto \mu_2 \}$.

Let $t \geq 0$ be an arbitrary time point.
If $\valuet\big(
\varphi^{\textsf{PrefConfO}}_{\tau_O,\epsilon_O}(\pi_1, \pi_2),
H,\Pi,T_0,t) \in \{\T,\U\}$,
then it holds that $\valuet\big(
\big(\LTLpastglobally \varphi^{\textsf{PrefConfI}}_{\tau_I,\epsilon_I}(\pi_1, \pi_2)
\big)
\rightarrow
\varphi^{\textsf{PrefConfO}}_{\tau_O,\epsilon_O}(\pi_1, \pi_2),
H,\Pi,T_0,t\big) \in \{\T,\U\}$.
If, on the other hand we have that
$\valuet\big(
\varphi^{\textsf{PrefConfO}}_{\tau_O,\epsilon_O}(\pi_1, \pi_2),
H,\Pi,T_0,t) = \F$,
then it holds that
$t \in \dom(\mu_1) \cup \dom(\mu_2)$ and
by Proposition~\ref{prop:correctness-prefcon} we have that
$\prefixconf_O(\mu_1\revised{^O},\mu_2\revised{^O},t)$ is false.
According to Definition~\ref{def:confclean-det}, this means that there exists $t' \leq t$ such that $\prefixconf_I(\mu_1\revised{^I},\mu_2\revised{^I},t')$ is false.
Applying Proposition~\ref{prop:correctness-prefcon} we obtain for the time point $t'$ that
$\valuet\big(
\varphi^{\textsf{PrefConfI}}_{\tau_I,\epsilon_I}(\pi_1, \pi_2),H,\Pi,T_0,t')  = \F$, and hence
$\valuet\big(
\LTLpastglobally\varphi^{\textsf{PrefConfI}}_{\tau_I,\epsilon_I}(\pi_1, \pi_2),H,\Pi,T_0,t)  = \F$.
This implies that
$\valuet\big(
\big(\LTLpastglobally
\varphi^{\textsf{PrefConfI}}_{\tau_I,\epsilon_I}(\pi_1, \pi_2)
\big)
\rightarrow
\varphi^{\textsf{PrefConfO}}_{\tau_O,\epsilon_O}(\pi_1, \pi_2)
,H,\Pi,T_0,t\big) = \T$.

Therefore, since  $t \geq 0$ was chosen arbitrarily,
$(H,\Pi,T_0,0) \models
\LTLglobally\big(
\big(\LTLpastglobally
\varphi^{\textsf{PrefConfI}}_{\tau_I,\epsilon_I}(\pi_1, \pi_2)
\big)
\rightarrow
\varphi^{\textsf{PrefConfO}}_{\tau_O,\epsilon_O}(\pi_1, \pi_2)
\big)$.
Since $\mu_1$ and $\mu_2$ were arbitrary, we conclude
$(H,\Pi_\emptyset,T_0,0) \models
\Phi^{\hybclean}_{\tau_I,\epsilon_I,\tau_O,\epsilon_O}$.

($\Longleftarrow$)
Now, suppose that $(H,\Pi_\emptyset,T_0,0) \models
\Phi^{\hybclean}_{\tau_I,\epsilon_I,\tau_O,\epsilon_O}$.
Let $\mu_1,\mu_2 \in H$ be two arbitrarily chosen traces.
Let $\pi_1$ and $\pi_2$ be trace variables, and define
$\Pi=\{\pi_1 \mapsto \mu_1,\pi_2 \mapsto \mu_2 \}$.

Let $t \in \dom(\mu_1) \cup \dom(\mu_2)$ be such that
for every $t' \leq t$  with $t' \in \dom(\mu_1) \cup \dom(\mu_2)$ it holds that $\prefixconf_I(\mu_1\revised{^I},\mu_2\revised{^I},t')$ is true.
By Proposition~\ref{prop:correctness-prefcon},
for every such $t'$ we have
$\valuet\big(\varphi^{\textsf{PrefConfI}}_{\tau_I,\epsilon_I}(\pi_1, \pi_2),H,\Pi,T_0,t'\big) \in \{\T,\U\}$.
If $t' \not\in \dom(\mu_1) \cup \dom(\mu_2)$,
we have $\valuet\big(\varphi^{\textsf{PrefConfI}}_{\tau_I,\epsilon_I}(\pi_1, \pi_2),H,\Pi,T_0,t'\big) = \U$.
Thus, $\valuet(\LTLpastglobally\varphi^{\textsf{PrefConfI}}_{\tau_I,\epsilon_I}(\pi_1, \pi_2),\Pi,T_0,t)  = \T$.

By assumption,
$\valuet\big(\big(\LTLpastglobally
\varphi^{\textsf{PrefConfI}}_{\tau_I,\epsilon_I}(\pi_1, \pi_2)
\big)
\rightarrow
\varphi^{\textsf{PrefConfO}}_{\tau_O,\epsilon_O}(\pi_1, \pi_2)
,H,\Pi,T_0,t\big) \in \{\T,\U\}$.
Therefore,
$\valuet\big(
\varphi^{\textsf{PrefConfO}}_{\tau_O,\epsilon_O}(\pi_1, \pi_2)
,H,\Pi,T_0,t\big) \in \{\T,\U\}$.
Since $t \in \dom(\mu_1) \cup \dom(\mu_2)$,
by Proposition~\ref{prop:correctness-prefcon}, this implies that
$\prefixconf_O(\mu_1\revised{^O},\mu_2\revised{^O},t)$ is true.
Since $t \in \dom(\mu_1) \cup \dom(\mu_2)$ was chosen arbitrarily,
we conclude that $\hybclean(\tau_I,\epsilon_I,\tau_O,\epsilon_O)$.
\end{proof}

In the special case when $\tau_I = \tau_O = 0$,  i.e., when we consider robust cleanness, the characterization of cleanness is simpler, since \revised{in the definition of $\prefixconf$ we only need to consider the actual compared prefixes (and not their truncated or extended versions)  because the timing deviation is $0$.
As per the definition of robust cleanness, we consider $\prefixconf$ instantiated with $\tracec$}.

Thus,  we can characterize robust cleanness as
\[\Phi^{\robclean}_{\epsilon_I,\epsilon_O} = \forall \pi_1. \forall\pi_2. \;
\LTLglobally \Big(
\big(\LTLpastglobally
\varphi^{I}_{\epsilon_I}
\big)
\rightarrow
\varphi^{O}_{\epsilon_O}
\Big),  \text{ where }\]
\[
\begin{array}{lll}
\varphi^{I}_{\epsilon_I} & = &
\big(*_1
\LTLeventually_{[0,0]}
d_{I}(X_{\pi_2},X^{*_1}_{\pi_1}) \leq \epsilon_I\big)
\wedge
\big(*_2
\LTLeventually_{[0,0]}
d_{I}(X_{\pi_1} ,X^{*_2}_{\pi_2}) \leq \epsilon_I\big) \text{ and}\\
\varphi^{O}_{\epsilon_O} & = &
\big(*_3
\LTLeventually_{[0,0]}
d_{O}(X_{\pi_2},X^{*_3}_{\pi_1}) \leq \epsilon_O\big)
\wedge
\big(*_4
\LTLeventually_{[0,0]}
d_{O}(X_{\pi_1} ,X^{*_4}_{\pi_2}) \leq \epsilon_O\big).
\end{array}
\]

\begin{prop}
Let $H$  be a deterministic hybrid system defined over a set of real-valued variables $X$ such that $0 \in \dom(\mu)$ for each $\mu \in H$.  Let  $\epsilon_I,\epsilon_O \geq 0$ be  rational constants.  It holds that
\begin{center}
$ H$ is $\robclean(\epsilon_I,\epsilon_O)$ if and only if
$ H \models \Phi^{\robclean}_{\epsilon_I,\epsilon_O}$.
\end{center}
\end{prop}
\begin{proof}
($\Longrightarrow$)
First, suppose that
$H$ is $\robclean(\epsilon_I,\epsilon_O)$.
Let $\mu_1,\mu_2 \in H$ be two arbitrarily chosen traces.
Let $\pi_1$ and $\pi_2$ be trace variables, and let
$\Pi=\{\pi_1 \mapsto \mu_1,\pi_2 \mapsto \mu_2 \}$.

Let $t \geq 0$ be an arbitrary time point.
If $\valuet\big(\varphi^{O}_{\epsilon_O},H,\Pi,T_0,t) \in \{\T,\U\}$,
then it holds that
$\valuet\big(
\LTLglobally \big(
\big(\LTLpastglobally
\varphi^{I}_{\epsilon_I}
\big)
\rightarrow
\varphi^{O}_{\epsilon_O}
\big),
H,\Pi,T_0,t\big) \in \{\T,\U\}$.
If, on the other hand we have that
$\valuet\big(\varphi^{O}_{\epsilon_O},H,\Pi,T_0,t) = \F$,
then it holds that at least one of the conjuncts of $\varphi^{O}_{\epsilon_O}$ is false. Suppose that, without loss of generality,
$\valuet\big(*_3
\LTLeventually_{[0,0]}
d_{O}(X_{\pi_2},X^{*_3}_{\pi_1}) \leq \epsilon_O,H,\Pi,T_0,t) = \F$.
Hence, $t \in \dom(\mu_1) \cup \dom(\mu_2)$ and
according to Definition~\ref{def:confclean-det},
this means that there exists $t' \leq t$ such that
$\valuet\big(
\varphi^{I}_{\epsilon_I},H,\Pi,T_0,t')  = \F$.
Therefore,
$\valuet\big(
\LTLpastglobally\varphi^{I}_{\epsilon_I},H,\Pi,T_0,t)  = \F$.
Thus, in this case we obtain that
$\valuet\big(
\big(\LTLpastglobally
\varphi^{I}_{\epsilon_I}
\big)
\rightarrow
\varphi^{O}_{\epsilon_O}
,H,\Pi,T_0,t\big) = \T$.

Therefore, since  $t \geq 0$ was chosen arbitrarily,
we have
$(H,\Pi,T_0,0) \models
\LTLglobally\big(
\big(\LTLpastglobally
\varphi^{I}_{\epsilon_I}
\big)
\rightarrow
\varphi^{O}_{\epsilon_O}
\big)$.
Since $\mu_1$ and $\mu_2$ were arbitrary, we can conclude
$(H,\Pi_\emptyset,T_0,0) \models
\Phi^{\robclean}_{\tau_I,\epsilon_I,\tau_O,\epsilon_O}$.

($\Longleftarrow$)
Now, suppose that $(H,\Pi_\emptyset,T_0,0) \models
\Phi^{\robclean}_{\tau_I,\epsilon_I,\tau_O,\epsilon_O}$.
Let $\mu_1,\mu_2 \in H$ be two arbitrarily chosen traces.
Let $\pi_1$ and $\pi_2$ be trace variables, and define
$\Pi=\{\pi_1 \mapsto \mu_1,\pi_2 \mapsto \mu_2 \}$.

Let $t \in \dom(\mu_1) \cup \dom(\mu_2)$ be such that
for every $t' \leq t$  with $t' \in \dom(\mu_1) \cup \dom(\mu_2)$ it holds that
$\valuet\big(
\varphi^{I}_{\epsilon_I},H,\Pi,T_0,t')  = \T$.
If $t' \not\in \dom(\mu_1) \cup \dom(\mu_2)$, we have
$\valuet\big(
\varphi^{I}_{\epsilon_I},H,\Pi,T_0,t')  = \U$.
Therefore, we have that $\valuet(\LTLpastglobally\varphi^{I}_{\epsilon_I},\Pi,T_0,t)  = \T$.

By assumption,
$\valuet\big(
\big(\LTLpastglobally
\varphi^{I}_{\epsilon_I}
\big)
\rightarrow
\varphi^{O}_{\epsilon_O}
,H,\Pi,T_0,t\big) \in \{\T,\U\}$.
Therefore, we obtain that
$\valuet\big(
\varphi^{O}_{\epsilon_O}
,H,\Pi,T_0,t\big) =\T$,
since $t \in \dom(\mu_1) \cup \dom(\mu_2)$.
Since $t$ was chosen arbitrarily,
we can conclude that $\robclean(\epsilon_I,\epsilon_O)$.
\end{proof}

\subsection{Monitoring \texorpdfstring{$\hyperstlstar_{\mathit{fin}}$}{HyperSTL*-fin} over finite-length real-valued signals}
We now consider the fragment $\hyperstlstar_{\mathit{fin}}$ and describe a method for offline monitoring of $\hyperstlstar_{\mathit{fin}}$ properties on finite sets of finite-length signals.

If the given traces are finite timed words, we can obtain from them piecewise linear signals by linear interpolation, or piecewise constant signals by fixing the value for each half-closed interval between time points to be the value at the starting point of this interval.

If the signals in the set are of different time length, we take the minimum length across the set, and consider the traces up to that length. Thus,  we ensure that for some $B$, all traces are defined over the interval $[0,B]$.  Furthermore,  we only consider formulas in $\hyperstlstar_{\mathit{fin}}$  for which the bounds of the temporal operators are such that every subformula has a defined  value when the formula is evaluated over traces with time domain $[0,B]$.

Our method handles the trace quantifiers similarly to the algorithm for offline monitoring of \hyperltl formulas on finite traces given in~\cite{DBLP:conf/rv/FinkbeinerHST17,MonitoringHyperproperties}.
The method iterates over tuples of generalized timed traces. The arity of the tuples is determined by the quantifier prefix of the formula.  For instance, for monitoring a formula of the form $\Phi = \forall \pi_1 \ldots \forall \pi_n\exists \pi_1' \ldots \exists \pi_m'. \varphi$ we will evaluate $\varphi$ on tuples of GTTs of arity $n+m$, to either determine that $\Phi$ is satisfied over the given set of traces, or return an $n$-tuple witnessing  a violation.
In order to evaluate the trace-quantifier-free formula $\varphi$ on an $(n+m)$- tuple of traces we compute a satisfaction evidence by using the method proposed in~\cite{FreezeSTL}.   \revised{Note that  unlike~\cite{DBLP:conf/rv/FinkbeinerHST17,MonitoringHyperproperties} we consider finite traces and formulas with bounded temporal operators.  Since we assume that the length of the traces is sufficient for all subformulas of a given $\hyperstlstar_{\mathit{fin}}$ formula of interest to have a defined value,  the truth value of the quantifier-free part of the formula is defined.  As we consider a fixed set of recorded traces,  we can check offline the satisfaction of $\hyperstlstar_{\mathit{fin}}$ formulas with arbitrary quantifier alternations similarly to~\cite{DBLP:conf/rv/FinkbeinerHST17,MonitoringHyperproperties},  as outlined above.}

Let $\varphi$ be a trace-quantifier-free formula. Let $H$ be a finite set of $\reals^l$-valued finite-length signals, and let $K \in H^k$ be a tuple of traces of arity $k$. Then, we can interpret $K$ as a real-valued signal $\kappa$ of order $k \times l$.
The \emph{satisfaction set} of the formula $\varphi$ over the signal $\kappa$ is defined analogously to~\cite{FreezeSTL}.  Similarly to~\cite{FreezeSTL}, we assume that the traces in $H$ are piecewise linear and that the atomic predicates are linear, in order to make the computation  of the satisfaction set tractable. Note that the atomic predicates used in the characterization of hybrid cleanness in Section~\ref{sec:cleanness-formulas} are linear when the expressions $d_{I}(X,X')$ and $d_{O}(X,X')$ are linear.  The explicit clock variable defines a linear signal.  Clearly, if the traces in  $H$ are piecewise linear, then so is the signal $\kappa$. Thus, the satisfaction set for $\varphi$ can be calculated effectively, represented as convex polytopes.

The satisfaction set for the formula $\varphi$ given the signal $\kappa$ is a subset of $\realsnn \times \realsnn^{\mathcal I}$, defined inductively with respect to structure of $\varphi$. Here, a signal $\kappa$ of order $k \times l$ is interpreted as a trace assignment for $k$ trace variables, in which each trace variable is assigned a GTT in the form of a real-valued signal of order $l$.
\[
\begin{array}{lll}
\mathsf{Sat}(\alpha,\kappa) & = &
\{(t,T) \in \realsnn \times \realsnn^{\mathcal I} \mid \valuet(\alpha,H,\kappa,T,t)=\T\};\\
\mathsf{Sat}(\top,\kappa) & = &
\realsnn \times \realsnn^{\mathcal I};\\
\mathsf{Sat}(\neg \varphi',\kappa) & = &
(\realsnn \times \realsnn^{\mathcal I}) \setminus \mathsf{Sat}(\varphi',\kappa);\\
\mathsf{Sat}(\varphi_1 \vee \varphi_2,\kappa) & = &
\mathsf{Sat}(\varphi_1,\kappa) \cup \mathsf{Sat}(\varphi_2,\kappa);\\
\mathsf{Sat}(\varphi_1\LTLuntil_{[a,b]}\varphi_2,\kappa) & = &
\{(t,T) \in \realsnn \times \realsnn^{\mathcal I} \mid
\exists t' \in [t+a,t+b]: (t',T) \in \mathsf{Sat}(\varphi_2,\kappa) \wedge \\&&
\phantom{\{}
\ \forall t'' \in [t,t'): (t'',T) \in \mathsf{Sat}(\varphi_1,\kappa) \cup \mathsf{Sat}(\varphi_2,\kappa) \revised{\cap [t+a,t+b]}\}; \\
\mathsf{Sat}(\varphi_1\LTLsince_{[a,b]}\varphi_2,\kappa) & = &
\{(t,T) \in \realsnn \times \realsnn^{\mathcal I} \mid
\exists t' \in [t-b,t-a]: (t',T) \in \mathsf{Sat}(\varphi_2,\kappa) \wedge \\&&
\phantom{\{(}
\ \forall t'' \in (t',t]: (t'',T) \in \mathsf{Sat}(\varphi_1,\kappa)\cup \mathsf{Sat}(\varphi_2,\kappa) \revised{\cap [t-b,t-a ]}\}; \\
\mathsf{Sat}(*_i \varphi',\kappa) & = &
\{(t,T) \in \realsnn \times \realsnn^{\mathcal I} \mid (t,T[i \mapsto t]) \in \mathsf{Sat}(\varphi',\kappa)\}.
\end{array}
\]
Once the satisfaction sets for the atomic predicates appearing in the given formula have been computed,  the satisfaction sets for the composite formulas can be constructed by following the inductive definition above.  Under the assumptions we made above about the signals and the atomic predicates, the satisfaction sets for the atomic predicates can be computed directly using the method described in~\cite{FreezeSTL}.  For further details,  we refer to~\cite{FreezeSTL}.

In order to perform monitoring for the formula $\Phi^{\hybclean}_{\tau_I,\epsilon_I,\tau_O,\epsilon_O}$ defined in Section~\ref{sec:cleanness-formulas},  we bound the temporal operators based on the signal length $B$,  and consider the case when $\tau_I>0$ and $\tau_O>0$ which ensures that the intervals in the operators are non-singular.

\section{Case Study}%
\label{sec:case_study}
In this section we evaluate the proposed notion of conformance-based cleanness in a real application context, known as the Diesel Emissions Scandal~\cite{DBLP:conf/qest/BiewerDH19,DBLP:conf/lpar/HermannsBDK18,DBLP:conf/rv/KohlHB18,DBLP:conf/cpsweek/BiewerDH18,DBLP:conf/esop/DArgenioBBFH17,DBLP:conf/sp/ContagLPDLHS17,BartheDFH16:isola}: Starting in fall 2015, millions of diesel cars were found  being equipped with defeat devices reducing the effectiveness of emission cleaning systems during real-world usage --- in contrast to the regulator-defined driving scenarios on a chassis dynamometer, where the amount of emitted pollutants stay well below the applicable limits.
%\looseness=-1

It was soon suspected, that the singularities of the testing procedure were straightforward to identify and hence made cheating easy; in particular, there was only a single test cycle for testing  emission cleaning systems, to be executed under very particular conditions. In the European Union, this was the \emph{New European Driving Cycle} (\NEDC)~\cite{nedc}, the speed profile of which is shown in Fig.~\ref{fig:NEDC}.
The \NEDC\ consists of four repetitions of an elementary \emph{urban driving cycle} (\UDC) followed by one \emph{extra urban driving cycle} (\EUDC).
%We split \UDC\ further in three subcycles, one for each of the ``hills''
Each test run is preceded by a preconditioning phase (\PreCon), in which three \EUDC s are driven consecutively.
Between \PreCon\ and the test, the vehicle has to cool down for 6 to 36 hours at an ambient temperature between 20 and 30 degrees Celsius.
%L 375/322, Section 5.3 (page 100)

Robust cleanness gives us a way of deriving additional test cycles that are reasonable w.r.t.\ the official \NEDC\@. For a concrete context, ``reasonable'' is defined by an accompanying formally defined \emph{contract}. The contributions of this paper enable us to go beyond previous work~\cite{DBLP:conf/qest/BiewerDH19} where a contract allowed inputs and outputs to deviate only in the value domains, but not in the time domain.

\subsection{Experimental Setup}
Our empirical studies apply the theory developed in the previous sections in a very specific setting. The system under test is a Nissan NV200 Evalia equipped with a Renault 1.5 dci (110hp) diesel engine and approved w.r.t.\ regulation \emph{Euro 6b}. All tests were conducted in November 2020.
As shown in Fig.~\ref{fig:teststand}, the car is fixed on a chassis dynamometer and attached to a portable emissions measurement system (PEMS) in preparation for the test.
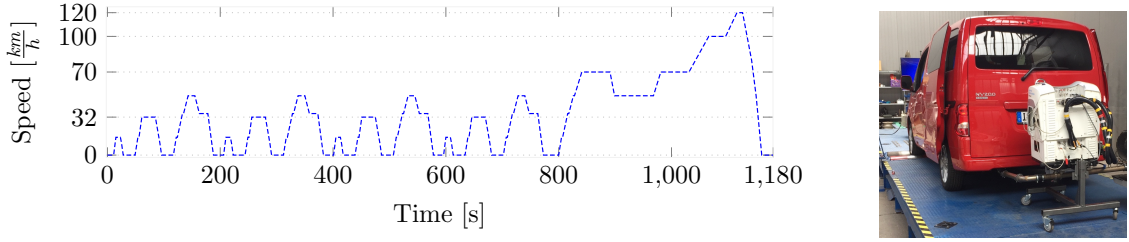
\begin{figure*}[t]
%\centering
\hspace{1cm}
\scalebox{0.9}{
\begin{tikzpicture}[trim axis left, scale=1]
\begin{axis}[
    width=0.75\textwidth,
    height=0.25\textwidth,
    xlabel={Time [s]},
    ylabel={Speed [$\frac{\mathit{km}}{h}$]},
    xmin=-2, xmax=1180,
    ymin=-2, ymax=125,
    xtick={0,200, 400, 600, 800, 1000, 1180},
    ytick={0, 32,  70, 100, 120}, %15, 50
    legend pos=north west,
    ymajorgrids=true,
    grid style=dotted,draw=gray!10,
]
 
\addplot[
    color=blue,
    line width=.5pt,
    dash pattern=on 2pt off 1pt,
    ]
    coordinates {
    (0,0)(6,0)(11,0)(15,15)(23,15)(25,10)(28,0)(44,0)(49,0)(54,15)(56,15)(61,32)(85,32)(93,10)(96,0)(112,0)(117,0)(122,15)(124,15)(133,35)(135,35)(143,50)(155,50)(163,35)(176,35)(178,35)(185,10)(188,0)(195,0)(201,0)(206,0)(210,15)(218,15)(220,10)(223,0)(239,0)(244,0)(249,15)(251,15)(256,32)(280,32)(288,10)(291,0)(307,0)(312,0)(317,15)(319,15)(328,35)(330,35)(338,50)(350,50)(358,35)(371,35)(373,35)(380,10)(383,0)(390,0)(396,0)(401,0)(405,15)(413,15)(415,10)(418,0)(434,0)(439,0)(444,15)(446,15)(451,32)(475,32)(483,10)(486,0)(502,0)(507,0)(512,15)(514,15)(523,35)(525,35)(533,50)(545,50)(553,35)(566,35)(568,35)(575,10)(578,0)(585,0)(591,0)(596,0)(600,15)(608,15)(610,10)(613,0)(629,0)(634,0)(639,15)(641,15)(646,32)(670,32)(678,10)(681,0)(697,0)(702,0)(707,15)(709,15)(718,35)(720,35)(728,50)(740,50)(748,35)(761,35)(763,35)(770,10)(773,0)(780,0)(800,0)(805,15)(807,15)(816,35)(818,35)(826,50)(828,50)(841,70)(891,70)(895,60)(899,50)(968,50)(981,70)(1031,70)(1066,100)(1096,100)(1116,120)(1126,120)(1142,80)(1150,50)(1160,0)(1180,0)
    };
\end{axis}
\end{tikzpicture}
}
\hspace{0.3cm}
\scalebox{0.9}{
\includegraphics[width=0.24\textwidth]{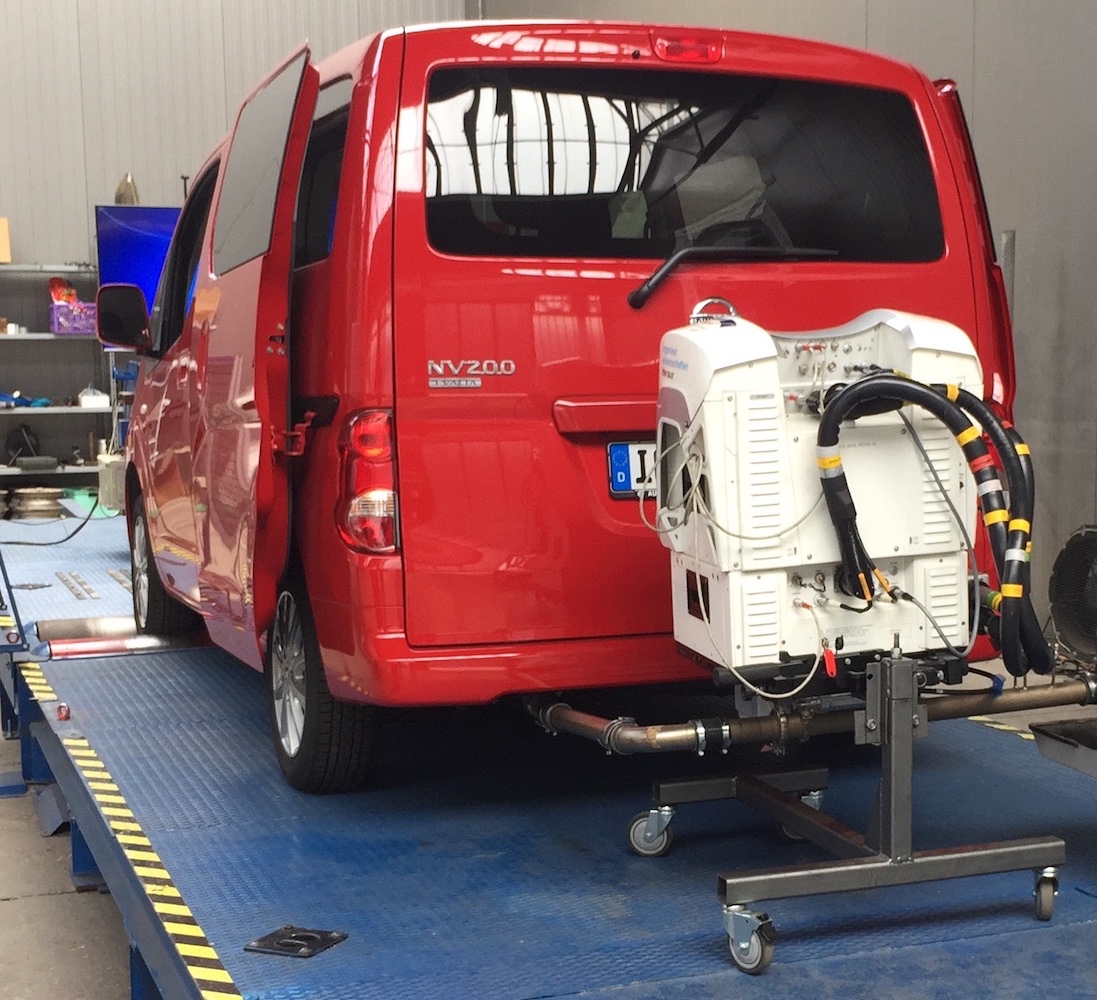}
}
%}
\vspace{-.3cm}
\caption{Left: New European Driving Cycle (NEDC); Right: Test Setup with Nissan NV200 Evalia on a chassis dynamometer attached to a PEMS\@.}%
\label{fig:NEDC}%
\label{fig:teststand}
\end{figure*}
The PEMS is connected to the onboard diagnostics (OBD) interface of the vehicle.
During a test, the PEMS measures the amount of several gases at the end of the car's exhaust pipe and logs the data received from OBD\@.
The PEMS is able to internally synchronise the times of gas measurements and OBD data.
We will not consider the internal PEMS retiming and instead analyse the final data set.
As input, we consider the OBD speed data, as output, the sum of emitted \NOgas\ and \NOtwo\ (abbreviated as \NOx).
The input and output is sampled by $1$ Hz. The amount of \NOx\ emitted along different runs is comparable only to a limited extend. This is because the emission cleaning system used can have internal regeneration phases, which --- from an external observer perspective --- are triggered nondeterministically. Hence, for the formal analysis we consider the accumulated output over $1180$ seconds (the length of a complete \NEDC). This is also the value decisive for type approval according to the official regulations.

\revised{Formally, a contract specifying \emph{clean} behaviour for a system is given as a tuple.
Previous work~\cite{DBLP:conf/qest/BiewerDH19} proposes the concrete contract $\mathcal{C}_t = \langle d_I, d_O, \epsilon_I, \epsilon_O \rangle$ for diesel doping tests, where $d_I(i_1, i_2) \eqdef \abs{i_1 - i_2}$ and $d_O(o_1, o_2) \eqdef \abs{o_1 - o_2}$, and $\epsilon_I = 15$ km/h and $\epsilon_O = 180$ mg/km.
%In this contract,
%the distance functions in $\mathcal{C}$ bound the absolute difference of speed, respectively \NOx, values, i.e.,
%$d_I(i_1, i_2) = \abs{i_1 - i_2}$ and $d_O(o_1, o_2) = \abs{o_1 - o_2}$.
%The value thresholds are $\epsilon_I = 15$ km/h and $\epsilon_O = 180$ mg/km.
$\mathcal{C}_t$ is based on robust cleanness, i.e., $\tracec_{\epsilon_I}$ and $\tracec_{\epsilon_O}$ conformance for inputs and outputs.
In the following, we will add the designated conformance notions to the tuple, i.e., instead of $\mathcal{C}_t$ we write $\mathcal{C} = \langle d_I, d_O, \tracec_{\epsilon_I}, \tracec_{\epsilon_O} \rangle$ (implicitly encoding $\epsilon_I$ and $\epsilon_O$).
In $\mathcal{C}$} the degree of deviation is constrained by a threshold on the pointwise difference of the new and the reference test input.
Comparisons of data at different time points are not possible.
The theory of conformance-based doping tests developed in this paper improves upon the previous testing methodology by enabling variation in the time domain. This gives us the possibility of  reordering \NEDC\ segments, lengthening a test beyond the time limits of the original \NEDC, and of increased tolerance regarding human-caused input distortions during test execution.
%To this end, the contract $\mathcal{C}$ needs to be adapted with tailored conformance notions. In the sequel we describe  how~$\mathcal{C}$ is adapted for each scenario.
\revised{With the parameters $d_I, d_O, \epsilon_I, \epsilon_O$ and $\tracec_{\epsilon_O}$ in $\mathcal{C}$ fixed, we show how adaptations of $\mathcal{C}$ with different input conformance relations are of benefit for software doping analysis.}

\begin{description}
\item[NEDC Permutations]
Based on the conformance notions in this work, we propose a new test cycle \PermNEDC\ in which \NEDC\ segments are permuted on the time axis.
Fig.~\ref{fig:PermNEDC} shows the test cycle. In each of the four \UDC\ segments the three non-zero speed-phases are permuted.
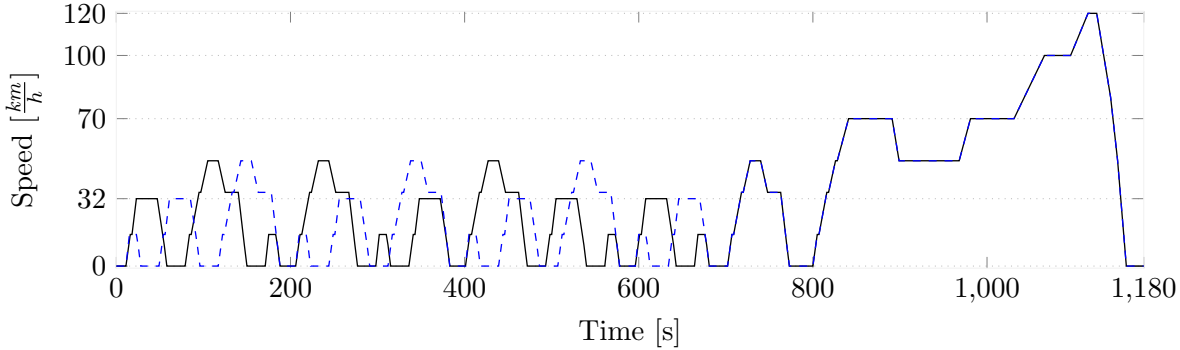
\begin{figure}[t]
\showplot{\begin{tikzpicture}[trim axis left, scale=1]
  \begin{axis}[
    width=\textwidth,
    height=5cm,
    xlabel={Time [s]},
    ylabel={Speed [$\frac{\mathit{km}}{h}$]},
    xmin=0, xmax=1180,
    xtick={0,200,400,600,800,1000,1180,2360},
    ytick={0, 32,70,100,120},
    legend pos=north west,
    ymajorgrids=true,
    grid style=dotted,draw=gray!10,
    ymin=-1,
    ymax=121  ]    \addplot [color=black, style=solid,line width=.5pt] coordinates {(0,0)(6,0)(11,0)(16,15)(18,15)(23,32)(47,32)(55,10)(58,0)(74,0)(79,0)(84,15)(86,15)(95,35)(97,35)(105,50)(117,50)(125,35)(138,35)(140,35)(147,10)(150,0)(166,0)(171,0)(175,15)(183,15)(185,10)(188,0)(195,0)(201,0)(206,0)(211,15)(213,15)(222,35)(224,35)(232,50)(244,50)(252,35)(265,35)(267,35)(274,10)(277,0)(293,0)(298,0)(302,15)(310,15)(312,10)(315,0)(331,0)(336,0)(341,15)(343,15)(348,32)(372,32)(380,10)(383,0)(390,0)(396,0)(401,0)(406,15)(408,15)(417,35)(419,35)(427,50)(439,50)(447,35)(460,35)(462,35)(469,10)(472,0)(488,0)(493,0)(498,15)(500,15)(505,32)(529,32)(537,10)(540,0)(556,0)(561,0)(565,15)(573,15)(575,10)(578,0)(585,0)(591,0)(596,0)(601,15)(603,15)(608,32)(632,32)(640,10)(643,0)(659,0)(664,0)(668,15)(676,15)(678,10)(681,0)(697,0)(702,0)(707,15)(709,15)(718,35)(720,35)(728,50)(740,50)(748,35)(761,35)(763,35)(770,10)(773,0)(780,0)(800,0)(805,15)(807,15)(816,35)(818,35)(826,50)(828,50)(841,70)(891,70)(895,60)(899,50)(968,50)(981,70)(1031,70)(1066,100)(1096,100)(1116,120)(1126,120)(1142,80)(1150,50)(1160,0)(1180,0)};

\addplot [color=blue, style=dashed,line width=.5pt] coordinates {(0,0)(6,0)(11,0)(15,15)(23,15)(25,10)(28,0)(44,0)(49,0)(54,15)(56,15)(61,32)(85,32)(93,10)(96,0)(112,0)(117,0)(122,15)(124,15)(133,35)(135,35)(143,50)(155,50)(163,35)(176,35)(178,35)(185,10)(188,0)(195,0)(201,0)(206,0)(210,15)(218,15)(220,10)(223,0)(239,0)(244,0)(249,15)(251,15)(256,32)(280,32)(288,10)(291,0)(307,0)(312,0)(317,15)(319,15)(328,35)(330,35)(338,50)(350,50)(358,35)(371,35)(373,35)(380,10)(383,0)(390,0)(396,0)(401,0)(405,15)(413,15)(415,10)(418,0)(434,0)(439,0)(444,15)(446,15)(451,32)(475,32)(483,10)(486,0)(502,0)(507,0)(512,15)(514,15)(523,35)(525,35)(533,50)(545,50)(553,35)(566,35)(568,35)(575,10)(578,0)(585,0)(591,0)(596,0)(600,15)(608,15)(610,10)(613,0)(629,0)(634,0)(639,15)(641,15)(646,32)(670,32)(678,10)(681,0)(697,0)(702,0)(707,15)(709,15)(718,35)(720,35)(728,50)(740,50)(748,35)(761,35)(763,35)(770,10)(773,0)(780,0)(800,0)(805,15)(807,15)(816,35)(818,35)(826,50)(828,50)(841,70)(891,70)(895,60)(899,50)(968,50)(981,70)(1031,70)(1066,100)(1096,100)(1116,120)(1126,120)(1142,80)(1150,50)(1160,0)(1180,0)};
  \end{axis}
\end{tikzpicture}}
\caption{\PermNEDC\ (solid, black line) compared to \NEDC\ (dashed, blue line)}%
\label{fig:PermNEDC}
\end{figure}
The transformation from \NEDC\ to \PermNEDC\ can be described by a retiming function $r_p$.
An explicit definition of $r_p$ is space consuming, hence we omit it.
Along with the new cycle, we propose
two suitable variants of contract~$\mathcal{C}$ \revised{with different input conformances.}
%Both keep the trace conformance for outputs, but  adapt the input conformances to the new setting.
Neither input conformance is constrained by a time threshold; in other words, $\tau = \infty$, so we omit $\tau$ in the index.
\begin{itemize}
\item Contract $\mathcal{C}_a$ \revised{is as $\mathcal{C}$, but} entails input conformance $\conf^{\,\ret_a}_{\epsilon_I}$, where
$\ret_a = \{  (r, r^{-1}) \ \mid \ r \in \mathcal{T} \rightarrow \mathcal{T}  \text{ and } r \text{ is total and bijective}   \}$
is the family of retimings that allows any reordering of the \NEDC\ inputs.
Notably, no inputs can be added or removed.
\item Contract $\mathcal{C}_p$ \revised{adjusts $\mathcal{C}$ by enforcing input conformance $\conf^{\,\ret_p}_{\epsilon_I}$, where $\ret_p = \{ (r_p, r^{-1}_p) \}$ is the family of retimings that only allows the particular retiming $r_p$ used to design the test cycle as discussed above.
This input conformance is stricter than $\conf^{\,\ret_a}_{\epsilon_I}$} above; it enforces that \PermNEDC\ is not permuted any further by the driver.
\end{itemize}

\item[NEDC Lengthening]
Conformance-based doping tests can run longer than the \NEDC;\@ this is not possible with robust cleanness.
We propose the test cycle \DoubleNEDC, which consists of two consecutive \NEDC s.
\begin{figure}[t]
\showplot{\begin{tikzpicture}[trim axis left, scale=1]
  \begin{axis}[
    width=\textwidth,
    height=3cm,
    xlabel={Time [s]},
    ylabel={Speed [$\frac{\mathit{km}}{h}$]},
    xmin=0, xmax=2360,
    xtick={0,200,400,600,800,1000,1180,2360},
    ytick={0, 32,70,100,120},
    legend pos=north west,
    ymajorgrids=true,
    grid style=dotted,draw=gray!10,
    ymin=-1,
    ymax=121  ]    \addplot [color=black, style=solid,line width=.5pt] coordinates {(0,0)(6,0)(11,0)(15,15)(23,15)(25,10)(28,0)(44,0)(49,0)(54,15)(56,15)(61,32)(85,32)(93,10)(96,0)(112,0)(117,0)(122,15)(124,15)(133,35)(135,35)(143,50)(155,50)(163,35)(176,35)(178,35)(185,10)(188,0)(195,0)(201,0)(206,0)(210,15)(218,15)(220,10)(223,0)(239,0)(244,0)(249,15)(251,15)(256,32)(280,32)(288,10)(291,0)(307,0)(312,0)(317,15)(319,15)(328,35)(330,35)(338,50)(350,50)(358,35)(371,35)(373,35)(380,10)(383,0)(390,0)(396,0)(401,0)(405,15)(413,15)(415,10)(418,0)(434,0)(439,0)(444,15)(446,15)(451,32)(475,32)(483,10)(486,0)(502,0)(507,0)(512,15)(514,15)(523,35)(525,35)(533,50)(545,50)(553,35)(566,35)(568,35)(575,10)(578,0)(585,0)(591,0)(596,0)(600,15)(608,15)(610,10)(613,0)(629,0)(634,0)(639,15)(641,15)(646,32)(670,32)(678,10)(681,0)(697,0)(702,0)(707,15)(709,15)(718,35)(720,35)(728,50)(740,50)(748,35)(761,35)(763,35)(770,10)(773,0)(780,0)(800,0)(805,15)(807,15)(816,35)(818,35)(826,50)(828,50)(841,70)(891,70)(895,60)(899,50)(968,50)(981,70)(1031,70)(1066,100)(1096,100)(1116,120)(1126,120)(1142,80)(1150,50)(1160,0)(1180,0)(1186,0)(1191,0)(1195,15)(1203,15)(1205,10)(1208,0)(1224,0)(1229,0)(1234,15)(1236,15)(1241,32)(1265,32)(1273,10)(1276,0)(1292,0)(1297,0)(1302,15)(1304,15)(1313,35)(1315,35)(1323,50)(1335,50)(1343,35)(1356,35)(1358,35)(1365,10)(1368,0)(1375,0)(1381,0)(1386,0)(1390,15)(1398,15)(1400,10)(1403,0)(1419,0)(1424,0)(1429,15)(1431,15)(1436,32)(1460,32)(1468,10)(1471,0)(1487,0)(1492,0)(1497,15)(1499,15)(1508,35)(1510,35)(1518,50)(1530,50)(1538,35)(1551,35)(1553,35)(1560,10)(1563,0)(1570,0)(1576,0)(1581,0)(1585,15)(1593,15)(1595,10)(1598,0)(1614,0)(1619,0)(1624,15)(1626,15)(1631,32)(1655,32)(1663,10)(1666,0)(1682,0)(1687,0)(1692,15)(1694,15)(1703,35)(1705,35)(1713,50)(1725,50)(1733,35)(1746,35)(1748,35)(1755,10)(1758,0)(1765,0)(1771,0)(1776,0)(1780,15)(1788,15)(1790,10)(1793,0)(1809,0)(1814,0)(1819,15)(1821,15)(1826,32)(1850,32)(1858,10)(1861,0)(1877,0)(1882,0)(1887,15)(1889,15)(1898,35)(1900,35)(1908,50)(1920,50)(1928,35)(1941,35)(1943,35)(1950,10)(1953,0)(1960,0)(1980,0)(1985,15)(1987,15)(1996,35)(1998,35)(2006,50)(2008,50)(2021,70)(2071,70)(2075,60)(2079,50)(2148,50)(2161,70)(2211,70)(2246,100)(2276,100)(2296,120)(2306,120)(2322,80)(2330,50)(2340,0)(2360,0)};

\addplot [color=blue, style=dashed,line width=.5pt] coordinates {(0,0)(6,0)(11,0)(15,15)(23,15)(25,10)(28,0)(44,0)(49,0)(54,15)(56,15)(61,32)(85,32)(93,10)(96,0)(112,0)(117,0)(122,15)(124,15)(133,35)(135,35)(143,50)(155,50)(163,35)(176,35)(178,35)(185,10)(188,0)(195,0)(201,0)(206,0)(210,15)(218,15)(220,10)(223,0)(239,0)(244,0)(249,15)(251,15)(256,32)(280,32)(288,10)(291,0)(307,0)(312,0)(317,15)(319,15)(328,35)(330,35)(338,50)(350,50)(358,35)(371,35)(373,35)(380,10)(383,0)(390,0)(396,0)(401,0)(405,15)(413,15)(415,10)(418,0)(434,0)(439,0)(444,15)(446,15)(451,32)(475,32)(483,10)(486,0)(502,0)(507,0)(512,15)(514,15)(523,35)(525,35)(533,50)(545,50)(553,35)(566,35)(568,35)(575,10)(578,0)(585,0)(591,0)(596,0)(600,15)(608,15)(610,10)(613,0)(629,0)(634,0)(639,15)(641,15)(646,32)(670,32)(678,10)(681,0)(697,0)(702,0)(707,15)(709,15)(718,35)(720,35)(728,50)(740,50)(748,35)(761,35)(763,35)(770,10)(773,0)(780,0)(800,0)(805,15)(807,15)(816,35)(818,35)(826,50)(828,50)(841,70)(891,70)(895,60)(899,50)(968,50)(981,70)(1031,70)(1066,100)(1096,100)(1116,120)(1126,120)(1142,80)(1150,50)(1160,0)(1180,0)};
  \end{axis}
\end{tikzpicture}}
\caption{\DoubleNEDC\ (solid, black line) compared to \NEDC\ (dashed, blue line)}%
\label{fig:DoubleNEDC}
\end{figure}
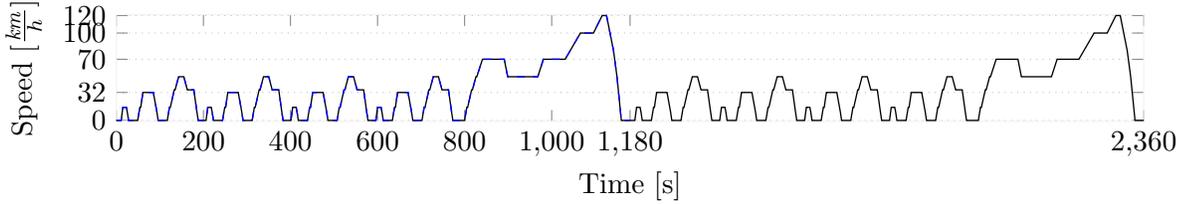
In contrast to all other test cycles in this paper, \DoubleNEDC\ produces two outputs: the first after $1180$ seconds, and the second after $2360$ seconds.
The first half of this cycle is a classical ``cold'' \NEDC\ (i.e., the engine cooled down before the test execution).
The second half is a ``hot'' \NEDC, since the cool-down phase was implicitly skipped.
Also, the \PreCon\ phase is skipped implicitly; there is only a single \EUDC\ (instead of three)  prior to the second \NEDC\@.
The inputs of both \NEDC s can be compared by using the retiming functions $\mathsf{id}$ and $r_d = \lambda\;t.\;t \textit{ mod }1180$. % chktex 35
$r_d$ maps time points of the global ``test case clock'' to the local time point in the \NEDC\ time domain.
\DoubleNEDC\ requires to adapt contract $\mathcal{C}$ to $\mathcal{C}_d$ by replacing robust cleanness by cleanness with synchronised retiming:
\revised{it entails input conformance $\conf^{\,\ret_d}_{\epsilon_I}$, which allows only $\mathsf{id}$ and $r_d$ as retiming functions, i.e.,
$\ret_d = \{ (\mathsf{id}, r_d) \}$.
Similarly, $\mathcal{C}_d$ must include the synchronisation retiming function $\synch_d(r_1, r_2) = (\mathsf{id}, r_d)$, which} enforces that both \DoubleNEDC\ outputs are compared to the single \NEDC\ output (independent of the input retimings $r_1$ and $r_2$).
%At this point, we could also define $\synch$ to be the identity function, as the only possible input retiming is $(\mathsf{id}, r_d)$.
%However, with our definition we will be able to change the input retiming family later on without disturbing the output comparison.

\item[Human \revised{Time} Imprecision Tolerance]
Diesel doping tests are executed by humans driving a car.
Humans tend to make mistakes when driving.
Mistakes can be the over- or undershooting of the targeted speed (the error is on the value axis), or accelerations or decelerations happening too early or too late (the error is a shift on the time axis), or superpositions thereof.
\revised{To compensate for both value and time errors, we use hybrid conformance.
As a formal contract, this would be expressed by a variant of $\mathcal{C}$, in which the input conformance is replaced by $\hybc_{\tau_I,\epsilon_I}$ for some $\tau_I > 0$.
For the purpose of demonstration, we will later analyse several such variants of $\mathcal{C}$, each variant with a unique value for $\tau_I$ and $\epsilon_I$, i.e., we consider the contract $\mathcal{C}(\tau_I, \epsilon_I)$ parametrised in $\tau_I$ and $\epsilon_I$.
%, $\mathcal{C}_p$ and $\mathcal{C}_d$, which are parametrised by $\tau_I$ and $\epsilon_I$.
%Derived from $\mathcal{C}$, the parametrised contract $\mathcal{C}(\tau_I, \epsilon_I)$ substitutes the conformance relation for inputs with $\hybc_{\tau_I,\epsilon_I}$.
Concrete values for $\tau_I$ and $\epsilon_I$ must be specified when using the contract.}

A test cycle that \revised{reflects drivings rich of} acceleration and deceleration phases\revised{---and is hence particularly prone to human driving errors---}is \SineNEDC~\cite{DBLP:conf/qest/BiewerDH19}.
\begin{figure}[t]
\showplot{\input{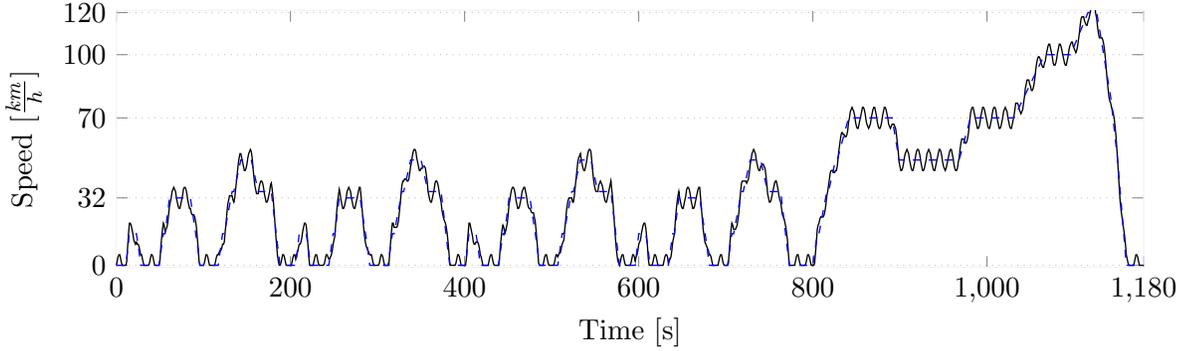}}
\caption{\SineNEDC\ (solid, black line) compared to \NEDC\ (dashed, blue line)}%
\label{fig:SineNEDC}
\end{figure}
\SineNEDC\ is defined as the \NEDC\ superimposed by a sine curve, formally
\[
    \text{\SineNEDC}(t) = \max\{0, \text{\NEDC}(t) + 5\sin(0.5t) \},
\]
with a maximum input deviation from \NEDC\  of $ 5\text{km/h}$, compare Fig.~\ref{fig:SineNEDC}.
\revised{We will evaluate \SineNEDC{} under several variants of $\mathcal{C}(\tau_I, \epsilon_I)$.}

Human \revised{time} imprecision is as yet not considered in test cycles \PermNEDC\ and \DoubleNEDC,
both cycles require a cycle-specific conformance predicate.
However, tolerance for human imprecision can be added to these predicates by means of conformance and retiming composition.
Let $\ret^{(1)}$ and $\ret^{(2)}$ be two families of retimings.
Then $\ret^{(2)} \circ \ret^{(1)} \eqdef \{ (r^{(2)}_1 \circ r^{(1)}_1  ,  r^{(2)}_2 \circ r^{(1)}_2 \ \mid \ (r^{(2)}_1, r^{(2)}_2) \in \ret^{(2)} \text{ and } (r^{(1)}_1 , r^{(1)}_2) \in \ret^{(1)}  \}$ is the component-wise function composition.
The definition for conformance composition is $\conf^{\,\ret^{(2)}}_{\tau_2,\epsilon} \circ \conf^{\,\ret^{(1)}}_{\tau_1,\epsilon} \eqdef \conf^{\,\ret^{(3)}}_{\infty,\epsilon}$, where $\ret^{(3)} = \ret^{(2)}_{\tau_2} \circ \ret^{(1)}_{\tau_1}$ composes the individual retimings.
The $\tau_1$- and $\tau_2$-constraints on $\ret^{(1)}$ and $\ret^{(2)}$ are applied before the composition.
It is not necessary to apply further timing constraints to the resulting retiming, hence we allow infinite $\tau$.
To overcome the human imprecisions for \PermNEDC\ and \DoubleNEDC,  we use the \revised{parametrised contracts $\mathcal{C}_p(\tau_I, \epsilon_I)$ and $\mathcal{C}_d(\tau_I, \epsilon_I)$, adaptations of $\mathcal{C}_p$ and $\mathcal{C}_d$}, with input conformances $\hybpc_{\tau_I,\epsilon_I} = \hybc_{\tau_I,\epsilon_I} \circ \conf^{\,\ret_p}_{\epsilon_I}$ and $\hybdc_{\tau_I,\epsilon_I} = \hybc_{\tau_I,\epsilon_I} \circ \conf^{\,\ret_d}_{\epsilon_I}$, respectively.
\revised{As for hybrid conformance in $\mathcal{C}(\tau_I, \epsilon_I)$, we will specify concrete $\tau_I$ and $\epsilon_I$ upon usage of the contracts.}
Notably, for \DoubleNEDC, this does not have effects on the output conformance, because \revised{$\synch_d$} does not consider the input retiming.
This is important, because outputs are available only at time points $1180$ and $2360$ and must not be moved to  time points different than that.
We do not compose $\conf^{\,\ret_a}_{\epsilon_I}$ and hybrid conformance, because $\ret_a$ allows any possible \NEDC\ permutation, which naturally reduces the effect of timing imprecisions.
\end{description}
\revised{Table~\ref{fig:contracts} summarises the contracts presented above. }

\begin{table}[t]
\revised{
\begin{alignat*}{7}
& \mathcal{C}\ & = \langle d_I,\ & d_O,\ && \tracec_{\epsilon_I},\ && \tracec_{\epsilon_O} & & \rangle \\
& \mathcal{C}_a\ & = \langle d_I,\ & d_O,\ && \conf^{\,\ret_a}_{\epsilon_I},\ && \tracec_{\epsilon_O} & & \rangle \\
& \mathcal{C}_p\ & = \langle d_I,\ & d_O,\ && \conf^{\,\ret_p}_{\epsilon_I},\ && \tracec_{\epsilon_O} & & \rangle \\
& \mathcal{C}_d\ & = \langle d_I,\ & d_O,\ && \conf^{\,\ret_d}_{\epsilon_I},\ && \tracec_{\epsilon_O},\ & \synch_d & \rangle \\
& \mathcal{C}(\tau_I, \epsilon_I)\ & = \langle d_I,\ & d_O,\ && \hybc_{\tau_I,\epsilon_I},\ && \tracec_{\epsilon_O} & & \rangle \\
& \mathcal{C}_p(\tau_I, \epsilon_I)\ & = \langle d_I,\ & d_O,\ && \hybc_{\tau_I,\epsilon_I} \circ \conf^{\,\ret_p}_{\epsilon_I},\ && \tracec_{\epsilon_O} & & \rangle \\
& \mathcal{C}_d(\tau_I, \epsilon_I)\ & = \langle d_I,\ & d_O,\ && \hybc_{\tau_I,\epsilon_I} \circ \conf^{\,\ret_d}_{\epsilon_I}, \ && \tracec_{\epsilon_O},\ & \synch_d & \rangle \\
\end{alignat*}
\vspace{-10mm}
\caption{Overview of the evaluated contracts. For $\mathcal{C}$, $\mathcal{C}_a$, $\mathcal{C}_p$ and $\mathcal{C}_d$, $\epsilon_I = 15$ km/h. For all contracts, $\epsilon_O = 180$ mg/km, $d_I(i_1, i_2) = \abs{i_1 - i_2}$ and $d_O(o_1, o_2) = \abs{o_1 - o_2}$. $\ret_a$, $\ret_p$ and $\ret_d$ are defined as explained in the text above.}%
\label{fig:contracts}
}
\end{table}

\subsection{Computing Parameters of Hybrid Conformance}
In some experiments we compute, for a fixed time threshold $\tau$ and two test cycles, the minimal value error $\epsilon$ such that hybrid conformance holds for the input.
The implementation of this computation is inspired by the the $\hyperstlstar$ formula $\varphi^{\hybc}_{\tau,\epsilon}$, i.e., it computes $\min_\epsilon\ \varphi^{\hybc}_{\tau,\epsilon}$ for two test executions $\pi_1$ and $\pi_2$.
The algorithm is sketched below.
%\begin{align*}
%\mathsf{localMin}(t_1, x_1, x_2, \tau) &= \min\ \{  \abs{x_1[t_1] - x_2[t_2]} \ \mid \ t_2 \in [t_1-\tau; t_1+\tau] \cap \mathcal{T}_2  \} \\
%\mathsf{globalMin}(x_1,x_2,\tau) &= \max \ \{  \mathsf{localMin}(t_1, x_1, x_2, \tau) \ \mid \ t_1 \in \mathcal{T}_1  \}\\
%\epsilon_\text{min}(x_1,x_2,\tau) &= \max \ \{ \mathsf{globalMin}(x_1,x_2,\tau), \mathsf{globalMin}(x_2,x_1,\tau) \}
%\end{align*}
\begin{align*}
 \mathsf{localMin}(t_1, \mu_1, \mu_2, \tau) &= \min\ \{  d_{\cal Y}(\mu_2[t_2], \mu_1[t_1]) \ \mid \ t_2 \in [t_1-\tau; t_1+\tau] \cap \dom(\mu_2)  \} \\
\mathsf{globalMin}(\mu_1,\mu_2,\tau) &= \max \ \{  \mathsf{localMin}(t_1, \mu_1, \mu_2, \tau) \ \mid \ t_1 \in \dom(\mu_1)  \}\\
\epsilon_\text{min}(\mu_1,\mu_2,\tau) &= \max \ \{ \mathsf{globalMin}(\mu_1,\mu_2,\tau), \mathsf{globalMin}(\mu_2,\mu_1,\tau) \}
\end{align*}
Here, $\mathsf{localMin}(t_1, x_1, x_2, \tau)$ computes the minimal $\epsilon$ for subformula
$\LTLpastfinally_{[0,\tau]}
d_{\cal Y}(X_{\pi_2},X^{*_1}_{\pi_1}) \leq \epsilon\vee
\LTLeventually_{[0,\tau]}d_{\cal Y}(X_{\pi_2},X^{*_1}_{\pi_1}) \leq \epsilon$,
where the value of $X^{*_1}_{\pi_1}$ is frozen at time $t_1$.
$\mathsf{globalMin}(x_1,x_2,\tau)$ reflects the \emph{Globally} and \emph{Freeze} operator: it finds the maximum by quantifying over all $t_1 \in \dom(x_1)$ and by calling $\mathsf{localMin}$ with the frozen time value $t_1$.
To reflect the complete formula, $\epsilon_\text{min}(x_1,x_2,\tau)$ returns the maximum of the conjuncts, which are the results of $\mathsf{globalMin}$ for both combinations of $x_1$ and $x_2$.

The computations for $\hybpc$ and $\hybdc$ proceed in two steps.
Both $\ret_p$ and $\ret_d$ are singleton sets; it is known which retiming must be applied first.
For two traces $\mu_1$ and $\mu_2$ and retiming $(r_1, r_2)$, there are shifted traces $\mu'_1 = \mu_1 \circ r_2$ and $\mu'_2 = \mu_2 \circ r_1$.
The minimal $\epsilon$ for hybrid conformance is given by
$\epsilon^{(r_1,r_2)}_\text{min}(\mu_1,\mu_2,\tau) = \max \ \{ \mathsf{globalMin}(\mu_1,\mu'_2,\tau), \mathsf{globalMin}(\mu_2,\mu'_1,\tau) \}$.

% FORMULA:
%\[
%\begin{array}{lll}
%\varphi^{\hybc}_{\tau,\epsilon} & = &
%\Big(\LTLglobally *_1
%\big(\LTLpastfinally_{[0,\tau]}
%d_{\cal Y}(X_{\pi_2},X^{*_1}_{\pi_1}) \leq \epsilon\vee
%\LTLeventually_{[0,\tau]}d_{\cal Y}(X_{\pi_2},X^{*_1}_{\pi_1}) \leq \epsilon\big)\Big) \wedge \\&&
%\Big(\LTLglobally *_2
%\big(\LTLpastfinally_{[0,\tau]}d_{\cal Y}(X_{\pi_1} ,X^{*_2}_{\pi_2}) \leq \epsilon \vee
%\LTLeventually_{[0,\tau]}d_{\cal Y}(X_{\pi_1} ,X^{*_2}_{\pi_2}) \leq \epsilon \big)\Big)
%\end{array}
%\]

\subsection{Test Results \& Verdicts}
We executed each of \NEDC, \PermNEDC, \DoubleNEDC\ and \SineNEDC\ two times.
We identify a concrete test execution by a suffix -1 or -2 to test cycle identifier (e.g., \NEDC-1 is the first and \NEDC-2 the second execution of \NEDC).
Raw data and the implementation of the analysis is available online~\cite{supplementary}.
For \NEDC, we combined the result of both executions to an average value of $182$ mg/km of \NOx.
Notably, the Euro 6b regulation (to which our car is supposed to conform to) allows at most $80$ mg/km, and the car under test is certified with $60.8$ mg/km according to its documentation. The car is 3 years old.

For doping detection, a test verdict is only meaningful if its input trace is conformant to that of the average \NEDC\ execution; otherwise, the test is trivially passed.
\revised{We will first evaluate \PermNEDC{} w.r.t. $\mathcal{C}_a$ and $\mathcal{C}_p$, \DoubleNEDC{} w.r.t. $\mathcal{C}_d$, and \SineNEDC{} w.r.t. $\mathcal{C}$.
To demonstrate the effects of hybrid conformance, we then analyse the experiments w.r.t.\ the parametrised variants of the contracts $\mathcal{C}$, $\mathcal{C}_p$ and $\mathcal{C}_d$, respectively.
By definition of the test cycles, the nominal value difference for \PermNEDC{} and \DoubleNEDC{} after retiming is zero, and for \SineNEDC{} it is $5$ km/h.
Though, due to human imprecisions, the actual differences are significantly higher.}
\begin{itemize}
\item The executions of \PermNEDC\ are shown in Fig.~\ref{fig:Perm1NEDCResults} and~\ref{fig:Perm2NEDCResults}. The amount of emitted \NOx\ were $392$ mg/km for \PermNEDC-1 and $316$ mg/km for \PermNEDC-2.
%Here, $\conf^{\,\ret_a}_{\epsilon_I}$ does hold for $\epsilon_I \geq 3$ km/h for both executions; one of them is an instance of detected doping according to contract $\mathcal{C}_a$, which defines $\epsilon_I = 15$ km/h, as we will explain.
$\conf^{\,\ret_a}_{\epsilon_I}$ does hold for $\epsilon_I \geq 3$ km/h for both executions; with contract $\mathcal{C}_a$, which defines $\epsilon_I = 15$ km/h, drastic deviations of \NOx\ can be detected as doping.
It is detected for \PermNEDC-1, i.e., the cleanness test fails, as the difference of \NOx\  (compared to \NEDC) is $210$ mg/km and hence greater than $\epsilon_O = 180$ mg/km defined by $\mathcal{C}_a$.
%Cleanness test \PermNEDC-1 fails (i.e., doping is detected), as the difference of \NOx\  (compared to \NEDC) is $210$ mg/km and hence greater than $\epsilon_O = 180$ mg/km defined by $\mathcal{C}_a$.
Test \PermNEDC-2 passes with an \NOx\ difference of $134$ mg/km which is within the contract.
\begin{figure}[t]
\showplot{\input{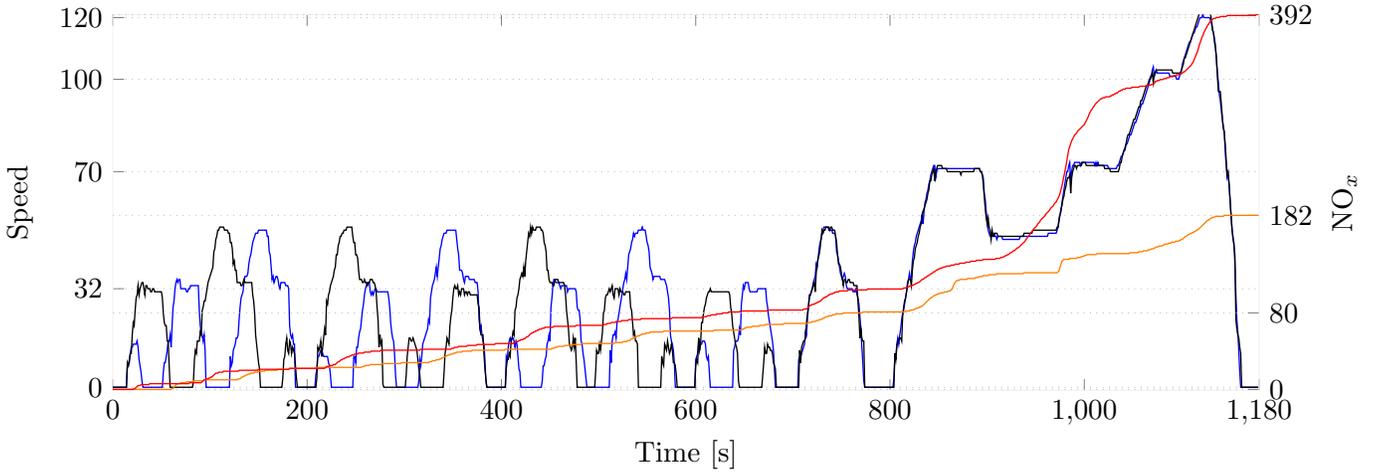}}
\caption{\PermNEDC-1 speed (black) and \NEDC\ speed (blue) in km/h, and accumulated \NOx\ for \PermNEDC-1 (red) and \NEDC\ (orange) in mg/km.}%
\label{fig:Perm1NEDCResults}
\end{figure}
\begin{figure}[t]
\showplot{\input{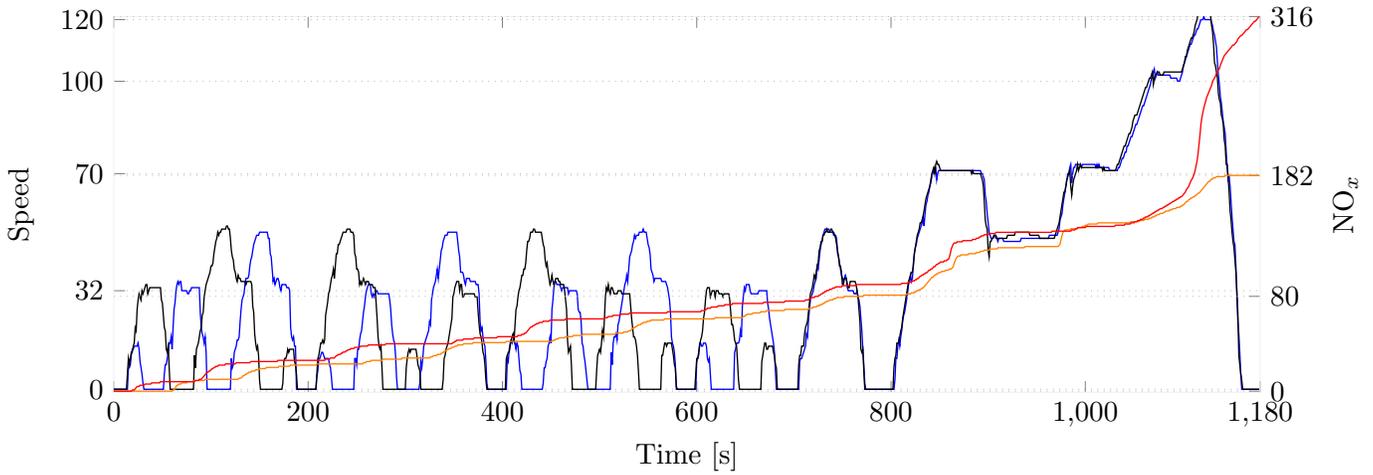}}
\caption{\PermNEDC-2 speed (black) and \NEDC\ speed (blue) in km/h, and accumulated \NOx\ for \PermNEDC-2 (red) and \NEDC\ (orange) in mg/km.}%
\label{fig:Perm2NEDCResults}
\end{figure}

With contract $\mathcal{C}_p$ and input conformance $\conf^{\,\ret_p}_{\epsilon_I}$, the test verdict for \PermNEDC-1 is different. $\conf^{\,\ret_p}_{\epsilon_I}$ would only hold for $\epsilon_I \geq 16$ km/h, which is above the contract defined threshold of $15$ km/h.
Hence, \PermNEDC-1 is not adduced and the test trivially passed.

\item \DoubleNEDC-1 and 2, shown in Fig.~\ref{fig:Double1NEDCResults} and~\ref{fig:Double2NEDCResults}, lead to an  average emission of $305$ mg/km, respectively $308$ mg/km of \NOx.
Executions of \DoubleNEDC\ are twice as long as regular \NEDC\ tests and produce two outputs.
The measurements for \DoubleNEDC-1 report $(229, 382)$ mg/km, for \DoubleNEDC-2 $(207, 408)$ mg/km.
\begin{figure}[t]
\showplot{\input{plots/double1nedcAvg}}
\caption{\DoubleNEDC-1 speed (black) and \NEDC\ speed (blue) in km/h, and accumulated \NOx\ for \DoubleNEDC-1 (red) and \NEDC\ (orange) in mg/km.}%
\label{fig:Double1NEDCResults}
\end{figure}
\begin{figure}[t]
\showplot{\input{plots/double2nedcAvg}}
\caption{\DoubleNEDC-2 speed (black) and \NEDC\ speed (blue) in km/h, and accumulated \NOx\ for \DoubleNEDC-2 (red) and \NEDC\ (orange) in mg/km.}%
\label{fig:Double2NEDCResults}
\end{figure}
To determine the verdicts for contract $\mathcal{C}_d$, we first check if $\conf^{\,\ret_d}_{\epsilon_I}$ holds.
This turns out not to hold for \DoubleNEDC-2, because we observed value deviations of up to $25$ km/h. This test is therefore trivially passed.
For \DoubleNEDC-1 all value deviations remain below the $15$ km/h threshold; this test run is thus to be considered relevant for output comparison.
According to the retiming synchronisation in $\mathcal{C}_d$, each of the outputs $229$ and $382$ must be compared to the \NEDC\ output $182$.
The output conformance is violated for the second output, with a difference of $200$ mg/km, exceeding the  allowed $\epsilon_O = 180$ mg/km threshold.
Hence, \DoubleNEDC-1 fails --- doping is detected.

\item During the test executions of \SineNEDC, we measured $483$ mg/km and $632$ mg/km.
The test progression is shown in Fig.~\ref{fig:Sine1NEDCResults} and~\ref{fig:Sine2NEDCResults}.
\begin{figure}[t]
\showplot{\input{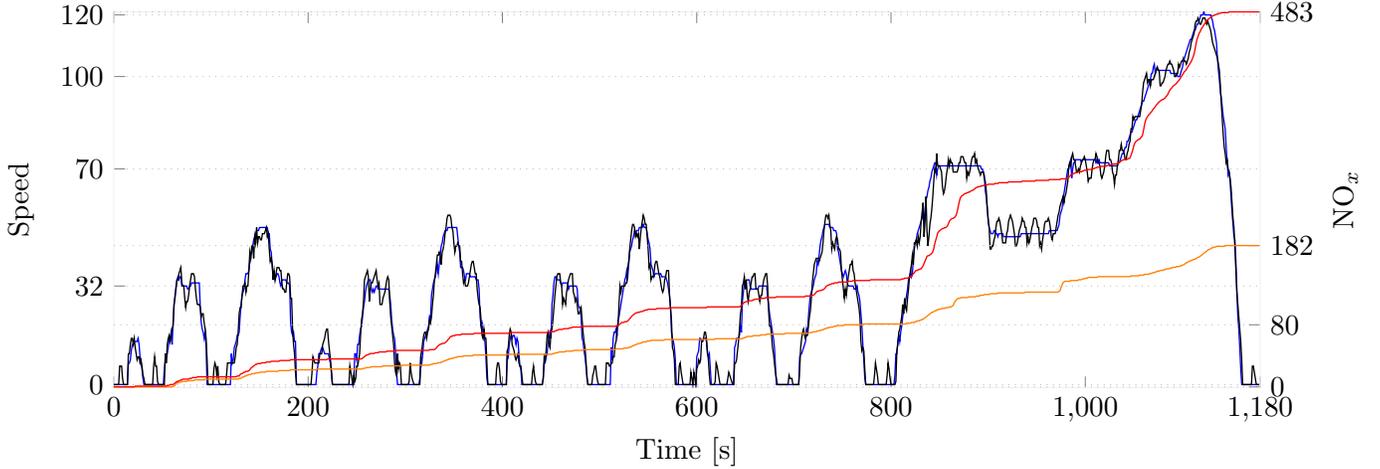}}
\caption{\SineNEDC-1 speed (black) and \NEDC\ speed (blue) in km/h, and accumulated \NOx\ for \SineNEDC-1 (red) and \NEDC\ (orange) in mg/km.}%
\label{fig:Sine1NEDCResults}
\end{figure}
\begin{figure}[t]
\showplot{
\input{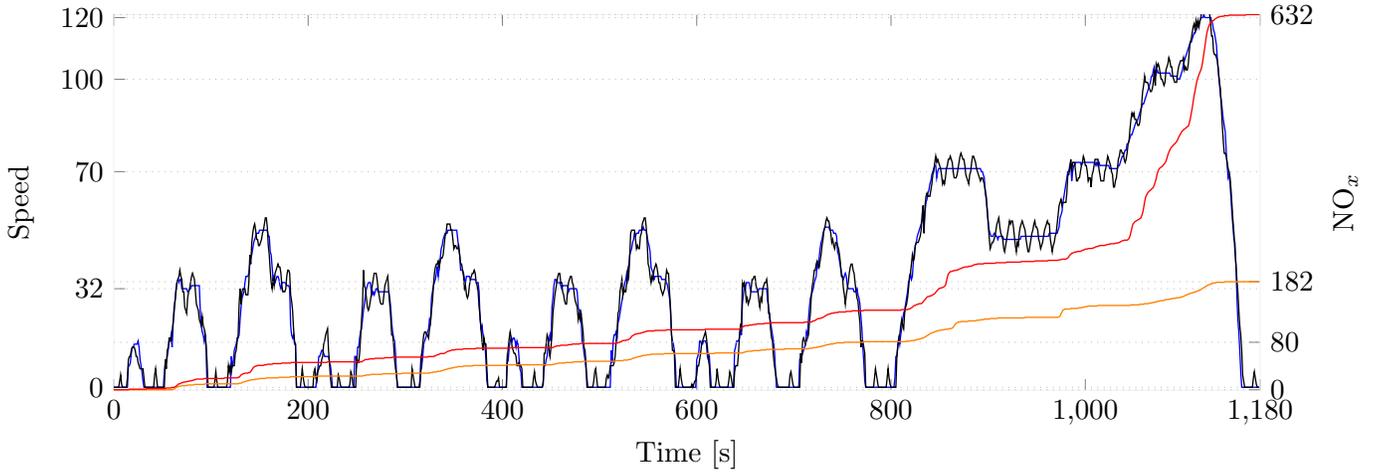} }
\caption{\SineNEDC-2 speed (black) and \NEDC\ speed (blue) in km/h, and accumulated \NOx\ for \SineNEDC-2 (red) and \NEDC\ (orange) in mg/km.}%
\label{fig:Sine2NEDCResults}
\end{figure}
In \SineNEDC-1, speed values deviate by up to $18$ km/h, which exceeds the $\epsilon_I$ threshold \revised{in $\mathcal{C}$}, so this test run is trivially passed.
\SineNEDC-2 respects the $\epsilon_I$ threshold because inputs never deviate by more than $13$ km/h.
Consequently, \SineNEDC-2 convicts our test car of doping, as the output difference of $450$ mg/km is $2.5$ times the allowed threshold $\epsilon_O$.

\item As discussed, we use hybrid conformance to compensate for human driving imprecisions.
In this context, Table~\ref{tbl:teCombinedNew} details the effect of a choice of $\tau$ on the maximal value error.
%\begin{figure}
%\includegraphics[width=0.9\textwidth]{images/placeholder-table}
%\caption{Hybrid Conformance Placeholder Table}
%\label{tbl:teCombinedNew}
%\end{figure}
\begin{table*}[t]\scriptsize
  \begin{tabular}{@{}@{\,}r@{\,}@{\,}c@{\,}@{\,}l@{\,}@{\:}l@{\:}@{\:}l@{\:}@{\:}l@{\:}@{\:}l@{\:}@{\:}l@{\:}@{\:}l@{\:}@{\:}l@{\:}@{\:}l@{\:}@{}}
    \textit{Test Name
    } & \textit{\revised{Contract}} & \textit{\revised{Input Conformance}} & $\tau\subi = 0$ & $\tau\subi = 1$ & $\tau\subi = 2$ & $\tau\subi = 3$ & $\tau\subi = 5$ & $\tau\subi = 10$ & $\tau\subi = 15$ & $\tau\subi = 20$  \\ \toprule
    \PermNEDC-1 & $\mathcal{C}_p(\tau_I, \epsilon_I)$ & \hybpc & $\epsilon\subi = 16$ & $\epsilon\subi = 16$ & $\epsilon\subi = 16$ & $\epsilon\subi = 11$ & $\epsilon\subi = 8$ & $\epsilon\subi = 8$ & $\epsilon\subi = 8$ & $\epsilon\subi = 8$ \\[1mm]
    \PermNEDC-2 & $\mathcal{C}_p(\tau_I, \epsilon_I)$ & \hybpc   & $\epsilon\subi = 11$ & $\epsilon\subi = 10$ & $\epsilon\subi = 7$ & $\epsilon\subi = 7$ & $\epsilon\subi = 7$ & $\epsilon\subi = 7$ & $\epsilon\subi = 7$ & $\epsilon\subi = 7$ \\[1mm]
    \DoubleNEDC-1 & $\mathcal{C}_d(\tau_I, \epsilon_I)$ & \hybdc   & $\epsilon\subi = 15$ & $\epsilon\subi = 12$ & $\epsilon\subi = 11$ & $\epsilon\subi = 9$ & $\epsilon\subi = 6$ & $\epsilon\subi = 6$ & $\epsilon\subi = 6$ & $\epsilon\subi = 6$ \\[1mm]
    \DoubleNEDC-2 & $\mathcal{C}_d(\tau_I, \epsilon_I)$ & \hybdc   & $\epsilon\subi = 25$ & $\epsilon\subi = 18$ & $\epsilon\subi = 10$ & $\epsilon\subi = 8$ & $\epsilon\subi = 8$ & $\epsilon\subi = 8$ & $\epsilon\subi = 8$ & $\epsilon\subi = 8$ \\[1mm]
    \SineNEDC-1 & $\mathcal{C}(\tau_I, \epsilon_I)$ & \hybc   & $\epsilon\subi = 18$ & $\epsilon\subi = 16$ & $\epsilon\subi = 15$ & $\epsilon\subi = 12$ & $\epsilon\subi = 9$ & $\epsilon\subi = 7$ & $\epsilon\subi = 6$ & $\epsilon\subi = 6$ \\[1mm]
    \SineNEDC-2 & $\mathcal{C}(\tau_I, \epsilon_I)$ & \hybc   & $\epsilon\subi = 13$ & $\epsilon\subi = 11$ & $\epsilon\subi = 9$ & $\epsilon\subi = 9$ & $\epsilon\subi = 7$ & $\epsilon\subi = 7$ & $\epsilon\subi = 7$ & $\epsilon\subi = 7$ \\
    \end{tabular}
%    \caption{Examples for $\hybc_{\te}(\NEDCdr,i)$ for $i\in\{\PowerNEDCdr,\SineNEDCdr\}$ and fixed $\tau$.}

%TORESOLVE

\caption{Comparison of minimal value thresholds $\epsilon\subi$ for fixed $\tau\subi$. Values are given as km/h and time in seconds.}%
\label{tbl:teCombinedNew}
\end{table*}
We fix a maximum value that we allow for the time offset $\tau\subi$.
For this $\tau\subi$ we analyse our dataset to find the minimal $\epsilon\subi$ such that for the combination of $\tau\subi$ and $\epsilon\subi$ the input traces under consideration satisfy the cycle-specific hybrid conformance.
For $\tau\subi = 0$ we get exactly the $\epsilon\subi$ for which the two traces satisfy $\conf^{\,\ret_p}_{\epsilon_I}$ (for \PermNEDC), $\conf^{\,\ret_d}_{\epsilon_I}$ (for \DoubleNEDC),  and $\tracec_{\epsilon_O}$ (for \SineNEDC).
Table~\ref{tbl:teCombinedNew} shows the computed $\epsilon\subi$ values for $\tau\subi = 0, 1, 2, 5,10,15$ and $20$ seconds.
As expected,  an increasing $\tau\subi$ induces the minimal $\epsilon\subi$ to decrease.
At $\tau\subi = 5$ the decrease in the value error reduces notably. This happens because the error is only partially caused by the incorrect timing of the driver.
From the values reported in Table~\ref{tbl:teCombinedNew} we see that if we allow deviation for the input $\tau\subi=2$, and keep $\epsilon\subi=15$, then we have that $\hybdc_{\tau_I,\epsilon_I}(\text{\NEDC}, \text{\DoubleNEDC-2})$ and $\hybc_{\tei}(\text{\NEDC}, \text{\SineNEDC-1})$ hold.
For time threshold $\tau\subi=3$ seconds $\hybpc_{\tau_I,\epsilon_I}(\text{\NEDC}, \text{\PermNEDC-1})$ also holds.
Thus, under hybrid conformance these pairs of traces will be considered in the cleanness test \revised{for contracts $\mathcal{C}_d(2, 15)$,  $\mathcal{C}(2, 15)$ and $\mathcal{C}_p(3, 15)$, respectively, while under their original contract and input} conformance they are to be dismissed.

\end{itemize}

\subsection{Evaluation and Discussion}
The amounts of emitted \NOx\ observed during our experiments provide clear indications of software doping regarding the car's emission cleaning system. The conformance-based contracts provide the formal basis for this verdict, as discussed above. We here complement this fact with a more intuitive explanation of the behaviour observed.

\begin{itemize}
\item
\PermNEDC\ slightly reorders \NEDC\ segments in the \UDC\ part of the test cycle.
During this part, the measured \NOx\ does not significantly differ from the \NEDC\ reference.
However, during the (unmodified) \EUDC\ part, the amount of emissions grows significantly.
It is very unlikely to find a physical explanation for the \NOx\ increase; and very likely, that the cleaning system is optimised specifically for the \NEDC\@.

\item
The \DoubleNEDC\ executions appear to reveal that the emission cleaning system optimisation can also rely on engine temperature or execution time instead of speed data.
Physically, many of the common emission cleaning techniques require a hot engine to work properly (and none of them requires a cold engine).
Therefore, a lower \NOx\ value can be expected if the \NEDC\ is run with a hot engine.
In our experiments, however, the \NOx\ emissions in the hot half are almost two times higher than in the initial cold part.
In other words, the emission cleaning performance is reduced after the first \NEDC\ execution. There is no physical explanation for this behaviour.
Inside the software, detecting the end of an \NEDC\ trip can be implemented very easily, for instance with a timer counting from $1180$ --- the length of \NEDC\ --- to zero.

\item
With \SineNEDC, we test the cleaning system during driving behaviour which is rich in accelerations and decelerations.
An increased amount of \NOx\ can possibly be explained by physical phenomena.
However, we measured an increase of factors $2.7$ and $3.5$; these numbers can be safely considered as too high for a trustworthy emission cleaning system.
\end{itemize}

\noindent
Software doping theory provides the basis for  detecting software behaviour violating a formal contract. In this, physical aspects of the emission cleaning system should be considered during the construction of test cases, and test cycles for which drastically higher emissions can be explained physically, should not be considered.
The test cycles we used for our experiments were picked with automotive expertise to avoid physically stressful cycles.
If test cases are generated automatically from a contract, the physical constraints could be captured by the contract.

The contracts we use for our experiments can be interpreted as very generous in favour of the manufacturers.
Input thresholds such as $15$ km/h and $2$ seconds appear as reasonable values, keeping all tests close enough to the original \NEDC\@.
For the output threshold, we use a very large deviation value of $180$ mg/km, which allows \NOx\ emissions to  almost double compared to the original \NEDC\ value.
Despite the generosity of the contracts, our experiments have been able to reveal doping for all experiments except \PermNEDC-2.

The analysis of the data shows that it is indeed necessary to not only consider a deviation of value, but to also allow for timing deviations.
Considering value and timing deviations offers a rich set of potential test cycles for doping tests and allows to realistically verify conformance of a test cycle and a reference cycle;
especially when the quality of the studied driving tests suffers from the human-caused input distortions.
In this regard, cleanness notions entailing hybrid conformance are more adequate than conformance notions demanding punctual test executions, such as robust cleanness.
Without hybrid conformance, more of the doping cases we have detected would slip through.

Finally, while hybrid conformance is central to the  case study considered here, our generic theory of conformance-based cleanness allows for using other conformance notions as appropriate for the CPS under test.

\section*{Acknowledgments}
The work of Sebastian Biewer and Holger Hermanns is supported  by the ERC Grant 695614 (POWVER) and by the Deutsche Forschungsgemeinschaft (DFG, German Research Foundation) grant 389792660 as part of TRR~248, see \url{https://perspicuous-computing.science}. The work of Sebastian Biewer is supported by the Saarbr\"ucken Graduate School of Computer Science. The work of Holger Hermanns is supported by  the Key-Area Research and Development Program Grant 2018B010107004 of Guangdong Province.
The work of Mohammad Reza Mousavi has been partially supported by the UKRI Trustworthy Autonomous Systems Node in Verifiability, with Grant Award Reference EP/V026801/1.

\bibliographystyle{alphaurl}
\bibliography{bibliography}

\end{document}